\documentclass[aps,preprint,superscriptaddress]{revtex4-2}
\usepackage{amsmath}
\usepackage{amsthm}
\usepackage{amssymb} 
\usepackage{graphicx}
\usepackage{algorithm}
\usepackage{algpseudocode}
\usepackage{booktabs}
\usepackage{tabularx}
\usepackage{multirow} 
\usepackage{float}  
\usepackage[utf8]{inputenc}
\usepackage{array} 
\usepackage{siunitx}  
\usepackage{collcell}
\usepackage{longtable}
\usepackage{placeins}
\usepackage{threeparttable} 
\usepackage{needspace}
\newtheorem{proposition}{Proposition}[section] 
\begin{document}


\title{Comprehensive pKa Data Augmentation from Limited Real Data \\through an Engineered Models–Quantum Framework}



\author{Wang Rui}
\email[]{r-wang18@tsinghua.org.cn}     
\homepage[]{https://orcid.org/0000-0002-2362-3214} 
\thanks{Contact author}
\affiliation{Department of Chemistry, Tsinghua University, Beijing 100084, P. R. China}
\affiliation{Department of Chemical Engineering, Tsinghua University, Beijing 100084, P. R. China}

\author{Liu Dinghao}
\affiliation{School of Science, China Pharmaceutical University, Nanjing 211198, P. R. China}





\date{\today}

\begin{abstract}
Proton dissociation constants (pKa) are critical for functional molecule discovery and molecular modeling. Building on iBonD, the largest experimental pKa database established, we and other researchers have developed several methods including machine-learning-based empirical prediction and high-accuracy energy calculations. Despite this foundation, the rapid augmentation of high-quality  pKa data remains fundamentally constrained. As part of this work, we performed large-scale regression-based pKa prediction on unlabeled molecular datasets using a collection of extensively optimized machine-learning models. The results indicate that, since the feature distributions of unlabeled molecular datasets, the pKa data distribution approximates normality, with extreme scarcity of tail-region samples. Although such augmentation is highly valuable for improving overall data availability and predictive modeling, it remains insufficient for efficiently discovering molecules with broad-spectrum pKa properties. To address this, we explore the targeted generation of molecules with sparse pKa properties from the vast chemical space. Given that traditional continuous latent space VAE-RNN methods for molecular generation suffer from insufficient stability and fail to demonstrate clear advantages in complementing sparse data, we design and implement a quantum-assisted sparse-pKa molecular generation. Feasibility is validated on a simulated quantum annealer, and superior extreme-value sampling is further achieved  on physical coherent Ising machines (CIMs). Combined with established high-accuracy pKa prediction models and confidence filtering, the proposed workflow yields a substantial collection of novel high-confidence molecules and broadens the accessible pKa distribution beyond what can be achieved through regression-based augmentation alone. By integrating machine-learning-driven data augmentation with quantum-assisted exploration, this work establishes a practical paradigm for high-confidence physicochemical augmentation data augmentation from limited experimental observations. More broadly, it represents the first realistic molecular discovery scenario in which quantum-inspired optimization provides a genuinely complementary advantage for extending sparse property distributions, while simultaneously offering practical insights into the deployment and optimization of physical CIMs for scientific discovery.
\end{abstract}

\maketitle
\section{INTRODUCTION}
The proton dissociation constant (pKa) fundamentally governs the ionization state of molecules in solution, exerting a decisive influence on physicochemical properties such as lipophilicity, solubility, membrane permeability, and molecular recognition \cite{1, 2, 3, 4, 5, 6, 7}. Accurate knowledge of pKa values is therefore indispensable for functional molecule discovery, drug design, and the rationalization of reaction mechanisms in both aqueous and non-aqueous media. Over decades of research, substantial efforts have been devoted to the experimental determination and computational prediction of pKa. On the experimental front, curated databases such as iBonD \cite{8}, the largest public repository of experimental pKa data established by our group \cite{9} (similar to other benchmark labeled databases like QM7 and QM9 that have profoundly advanced quantum-chemical machine learning  \cite{10, 11, 12}) have provided an invaluable foundation for data-driven modeling  \cite{13, 14, 15, 16, 17}. In parallel, computational approaches have evolved along two complementary directions: high-accuracy quantum mechanical calculations based on density functional theory (DFT) and ab initio methods, which offer physical rigor but at prohibitive computational cost  \cite{18, 19, 20, 21, 22}, and empirical machine-learning-based predictions, which trade some accuracy for throughput.

By virtue of pKa databases, especially iBonD—the largest experimental pKa database constructed — we and other researchers have established diverse approaches that combine machine learning-based empirical prediction with high-precision energy calculation methods. However, despite these advances, the rapid augmentation of high-quality experimental pKa data remains constrained. Even the recent breakthroughs in large language models, including multimodal learning and optical character recognition  \cite{23, 24, 25}, have not fundamentally overcome the bottleneck of acquiring new, reliable experimental measurements.

\begin{figure*}
\centering
\includegraphics[width=1.0\textwidth]{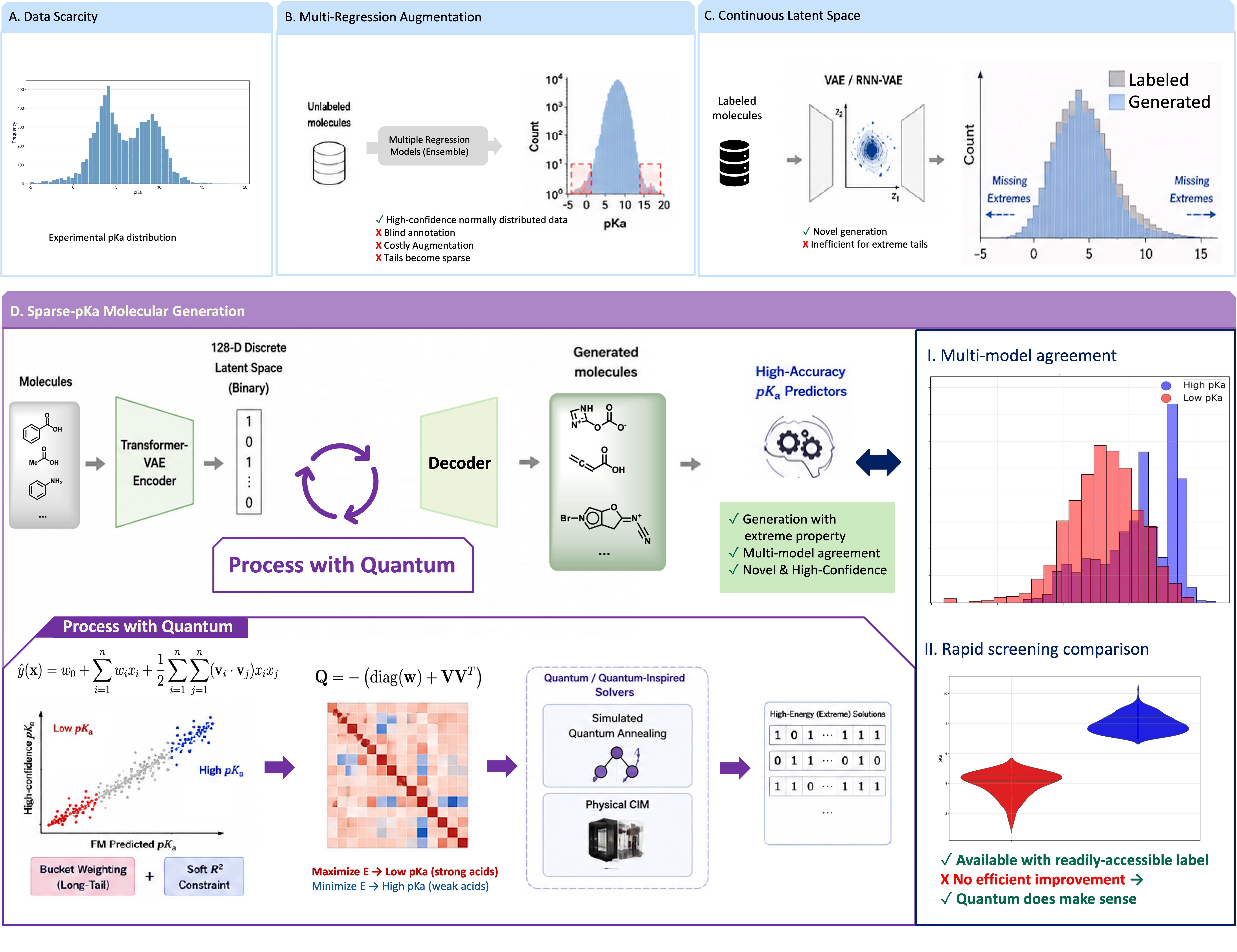}  
\caption{Broad-Spectrum and High-Confidence pKa Data Augmentation}
\label{fig1}
\end{figure*}

To alleviate data scarcity (FIG. 1A), we first pursued a large-scale data augmentation strategy based on supervised regression. By training multiple high-performance models on iBonD and applying them to extensive unlabeled molecular collections, we generated a substantially enlarged dataset of high-confidence computational pKa annotations. The combination of carefully engineered descriptors and powerful learning algorithms such as XGBoost produced strongly consistent predictions across multiple models, significantly increasing the amount of available pKa data beyond experimentally accessible scales.

Importantly, this data augmentation effectively enriched the dominant regions of the pKa distribution and laid a solid foundation for subsequent modeling and analysis. However, analysis of the final distribution revealed an inherent limitation of regression-driven data augmentation. The unlabeled molecular datasets adopted for augmentation are inherently biased toward molecules following a strong normal distribution. Since regression-based augmentation cannot actively target distribution regions, the newly generated data mostly fell within the conventional range. Consequently, while the total volume of high-confidence pKa data increased substantially, insufficient coverage arose in the tail regions corresponding to strong acids and extremely weak acids (FIG. 1B). This finding demonstrates that complementary approaches are required alongside regression-based augmentation to realize full-spectrum molecular property prediction.

To further extend the coverage of the pKa distribution, we sought methods capable of actively exploring underrepresented regions of chemical space. Such a task is fundamentally challenging because the target regions are defined by scarce and extreme property values. Conventional generative approaches are typically trained on data dominated by the central distribution and therefore tend to preferentially reproduce already well-sampled regions. Effective exploration of sparse-pKa regimes consequently requires not only generative capability but also a principled mechanism for navigating the discrete and combinatorially vast molecular landscape toward rare property objectives. This consideration motivated our integration of generative modeling with quantum-assisted optimization.

In preliminary exploration, we also investigated continuous latent-variable models for molecular generation \cite{26} (FIG. 1C). Although autoencoder-based architectures such as Variational Autoencoders (VAEs) and their recurrent variants provide an elegant framework for molecular representation learning, our experiments consistently revealed practical limitations for sparse-property exploration. Maintaining a useful balance between latent-space regularization and faithful molecular reconstruction proved difficult. Excessive regularization resulted in posterior collapse and loss of chemically meaningful information, whereas prioritizing reconstruction yielded noisy and weakly structured latent spaces. More importantly, exploration trajectories remained strongly influenced by the density structure of the training distribution, making targeted discovery of rare-property molecules inefficient. Consequently, while such models are effective for interpolation within densely sampled regions, they offered limited advantages for directed exploration of sparse pKa regimes.

To overcome these limitations, we developed a synergistic Transformer–VAE autoencoder that maps molecules into a compact 128-dimensional discrete latent space through a regularized binary bottleneck. The binary bottleneck is enforced via differentiable Gumbel–Softmax binarization  \cite{27, 28, 29}, compressing molecular representations into fixed-length binary codes while retaining disentangled structural and physicochemical features \cite{30}.

To construct an explicit energy landscape linking this discrete space to pKa, we trained a factorization machine (FM) surrogate model \cite{31, 32, 33} to predict pKa directly from binary latent codes. The FM captures pairwise interactions among latent variables and accurately models nonlinear relationships between discrete representations and molecular properties. To address severe long-tail sparsity, we incorporated pKa-bucket-specific sample weighting, a soft $R^2$ regularization term, and carefully tuned pairwise interaction parameters. Together, these strategies substantially improved both long-tail fitting performance and overall predictive accuracy. We additionally examined higher-order surrogate models amenable to QUBO extraction, including cubic FM variants and conditional Boltzmann machines, but neither provided meaningful improvements over the quadratic FM formulation.

From the optimized FM surrogate, a QUBO matrix can be extracted directly, with low-energy and high-energy solutions corresponding to low-pKa and high-pKa molecules, respectively. This transformation converts molecular design into a combinatorial optimization problem naturally compatible with quantum and quantum-inspired solvers. Molecular generation was subsequently performed using both simulated quantum annealing and physical coherent Ising machines (CIMs) \cite{34, 35, 36, 37, 38} operating under 8-bit and 14-bit integer precision modes.

The resulting QUBO formulation enables direct optimization within a discrete latent molecular representation. Annealing-based solvers are particularly attractive in this setting because they naturally support large-scale combinatorial exploration and diverse solution sampling, both of which are beneficial for discovering molecules residing in sparsely represented regions of the property distribution. Both simulated and physical solvers generated molecules with high validity and novelty, while the physical CIM consistently exhibited superior extreme-value sampling performance. For the best-fitted long-tail buckets, the Pearson correlation between FM-predicted pKa and QUBO energy remained stably above 0.9 in both low-pKa and high-pKa regions under both 8-bit and 14-bit CIM configurations, substantially exceeding the correlations obtained using simulated quantum annealing. Furthermore, comparison across multiple FM variants demonstrated that surrogate models specifically optimized for long-tail fitting produced markedly stronger sampling correlations, directly validating the importance of tailored surrogate design for rare-property discovery.

Both the 8-bit and 14-bit CIM configurations were operated exclusively in sampling mode. Compared with deterministic quota-based optimization, sampling-based operation introduced beneficial stochasticity that substantially enriched molecular diversity. With sufficiently large sample counts, the sampling mode was even capable of discovering lower-energy solutions than those obtained through deterministic optimization. A notable distinction emerged between the two precision modes. Although both achieved strong Pearson correlations ($> 0.9$), the 14-bit mode provided substantially higher Spearman rank correlations, indicating more faithful preservation of the property ranking encoded in the energy landscape. In contrast, the 8-bit mode generated greater molecular diversity and more effectively populated sparse pKa regions. These observations reveal a practical exploration–exploitation trade-off, where lower precision sacrifices ranking fidelity but introduces additional stochasticity that enhances coverage of long-tail chemical regions.

Multi-model evaluation confirms that CIM-generated candidates more frequently fall within high-confidence sparse-pKa regions, with low-pKa generation consistently outperforming high-pKa generation. Analysis of the extracted QUBO matrices and solution outcomes further reveals the primary limitation: despite effective tail-targeted generation, the confidence for the most extreme molecules remains moderate. This is partly due to the information loss inherent in discrete latent compression, but more importantly, it arises because pKa values cannot be adequately embedded through a latent space that is predominantly tied to molecular scaffolds. Precisely because of this, quantum-assisted optimization enables the generation of novel molecules that breaks free from the heavy reliance on high-accuracy regression fitting and expert knowledge. Building on this insight, we conduct feature importance and XGBoost Gain analyses using several newly developed high-accuracy regression models. We find that screening based on high XGBoost-Gain, readily available labels provides a viable route for rapidly pre-screening low- and high-pKa groups. Nevertheless, the results presented in this work demonstrate that, even with these limitations, direct physical-machine solving still retains clear advantages over such low-cost screening approaches. While the high XGBoost-Gain screening can serve as a promising complement, the core value of quantum-assisted generation remains evident.

Overall, this workflow (FIG. 1B\&1D, algorithm of 1D summarized in TABLE I) integrates large-scale high-confidence data augmentation with targeted quantum-assisted exploration. The resulting framework not only generates a substantial collection of novel molecules across underrepresented pKa regions but also broadens the accessible high-confidence pKa dataset derived from limited experimental observations. More generally, it establishes a transparent and potentially generalizable paradigm for combining machine-learning-driven data augmentation with quantum-assisted optimization to address challenging property-discovery problems in chemistry.
\begin{table*}[h!]
\centering
\caption{Sparse-pKa Molecular Generation \textit{via} FM-QUBO-CIM}
\label{tab:alg_workflow}
\footnotesize
\renewcommand{\arraystretch}{0.75}
\begin{tabular}{l}
\toprule
\textbf{Input:} \\
\quad Experimental pKa dataset $D_\text{exp}$ \\
\quad Unlabeled molecular library $D_\text{unlab}$ \\
\midrule
\textbf{Output:} \\
\quad High-confidence sparse-pKa molecules $M_\text{sparse}$ \\
\midrule
\textbf{Algorithm Steps:} \\
1.~~Train holistic pKa regressors on $D_\text{exp}$ \\
2.~~Generate pseudo-labels for $D_\text{unlab}$ \\
3.~~Apply cross-model agreement filtering \\
4.~~Construct high-confidence augmented dataset $D_\text{aug}$ \\
5.~~Train Transformer-VAE with SELFIES encoding \\
6.~~Encode molecules into 128-bit latent vectors $\boldsymbol{z}$ \\
7.~~Train long-tail optimized FM surrogate: $\hat{y}(\boldsymbol{z}) = \text{FM}(\boldsymbol{z})$ \\
8.~~Extract QUBO matrix $\boldsymbol{Q}$ from FM parameters \\
9.~~Solve QUBO using: \\
\qquad a) Simulated annealing \\
\qquad b) Physical CIM hardware \\
10.~Decode sampled binary solutions $\boldsymbol{z}^*$ into molecular structures \\
11.~Predict pKa using ensemble regressors \\
12.~Apply multi-model confidence filtering \\
13.~Return retained sparse-pKa candidates \\
\bottomrule
\end{tabular}
\end{table*}
\section{Regression‑Driven Data Augmentation}
\subsection{Holistic pKa Model Refinement and Paradigm Determination}
In our prior work, the decision-tree ensemble model XGBoost demonstrated superior performance in pKa prediction \cite{13, 16}. After systematically re‑evaluating multiple regression architectures in the present study—including deep neural networks, random forests, and support vector regression—we again adopted XGBoost as the core predictive engine, confirming its robustness and generalization capacity on the iBonD benchmark. Building on this foundation, we undertook a comprehensive paradigm exploration of feature selection and model combination strategies to establish a general‑purpose pKa prediction framework.

For feature engineering, we constructed a hybrid descriptor pool comprising 3,986 mixed‑cost descriptors, integrating multiple complementary molecular representations. The full descriptor set consists of 3,776 fingerprint bits, 208 RDKit physicochemical descriptors, and 2 categorical encoding features specifying solvent conditions and measurement type. The detailed fingerprint-based composition is summarized in TABLE II.

\begin{table}[htbp]
\centering
\footnotesize
\renewcommand{\arraystretch}{0.8}
\begin{threeparttable}
\caption{Fingerprint-Based Composition}
\label{tab:fingerprint_composition}
\begin{tabular}{lllcl}
\toprule
 Feature Type & Feature Range & Count & Description \\
\midrule
 Morgan fingerprint (radius 2) & fp\_0 \(\sim\) fp\_1023 & 1024 & nBits = 1024 \\
 Morgan fingerprint (radius 3) & fp\_1024 \(\sim\) fp\_2047 & 1024 & nBits = 1024 \\
 MACCS fingerprint & fp\_2048 \(\sim\) fp\_2214 & 167 & Standard 166-bit MACCS (+1 starter bit) \\
 Topological fingerprint & fp\_2215 \(\sim\) fp\_2726 & 512 & Actual 64 bits, remaining zeros padded \\
 AtomPairs fingerprint & fp\_2727 \(\sim\) fp\_3750 & 1024 & First 1024 bits used \\
 Physicochemical descriptors & fp\_3751 \(\sim\) fp\_3775 & 25 & See in Appendix TABLE A1\\
\bottomrule
\end{tabular}

\vspace{0.1em}
\flushleft
\scriptsize
\textbf{Total:} 3776 fingerprint features (fp\_0 \(\sim\) fp\_3775).
\end{threeparttable}
\end{table}

In addition to the fingerprint-based features listed above, the model further incorporated 208 RDKit descriptors, including EState indices, quantitative estimate of drug-likeness (QED), BCUT2D metrics, VSA‑series descriptors, fragment counts, and other graph‑based topological and constitutional properties. These descriptors capture essential physicochemical determinants of acid–base behavior—such as charge distribution, molecular connectivity, and solvation‑relevant surface properties—that complement the structural information encoded in the fingerprint features. The two categorical encoding features specify the solvent environment (e.g., water, ethanol, DMSO) and the measurement type (e.g., potentiometric vs. spectrophotometric), accounting for systematic variations in experimental conditions that are known to affect pKa determination.

\begin{figure*}
\centering
\includegraphics[width=0.9\textwidth]{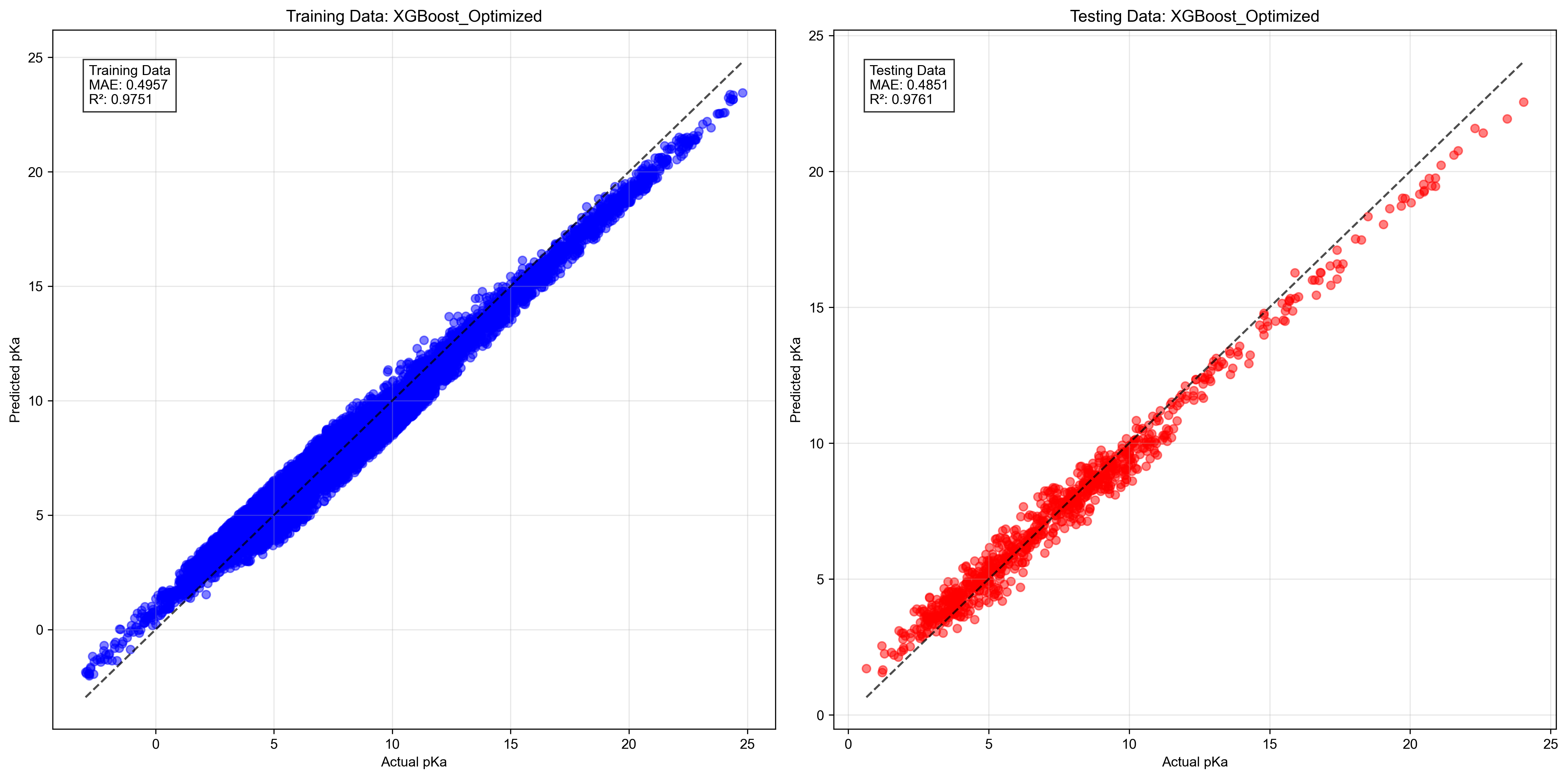}  
\caption{Final  Results of Optimized Holistic Regression Model }
\label{fig3}
\end{figure*}

To construct an optimal predictive model, we employed a multi‑stage feature selection pipeline combining K‑best univariate selection, recursive feature elimination (RFE), and mutual‑information‑based filtering. After systematic evaluation, the best‑performing feature subsets retained between 1,200 and 1,500 fingerprint bits together with the full set of 208 RDKit descriptors and the two categorical encoding features, resulting in a final dimensionality of approximately 1,410–1,710 features. Detailed feature selection outcomes and the composition of the retained subsets are provided in the Appendix A2\&A3.

With the pruned feature sets in place, we further optimized the XGBoost model itself through extensive hyperparameter tuning guided by cross‑validation. The refined models achieved a training mean absolute error (MAE) of 0.46–0.47 and a coefficient of determination ($R^2$ ) of 0.97. On independent test sets, these models exhibited strong generalization, with MAE around 0.6 and $R^2$ exceeding 0.94 (details in Appendix FIG.A1 and FIG.A2, K-best and RFE both selected exactly the same 1200 features). By rigorously employing 5‑fold cross‑validation throughout the training protocol, we obtained a final holistic pKa predictor that consistently delivers both training and test MAE below 0.5 and $R^2$ reaching 0.97, thereby preliminarily advancing the state‑of‑the‑art for holistic pKa prediction (detailed model training parameters are provided in the Appendix A2).

\subsection{Aqueous Specialized Model Construction}

\begin{table}[htbp]
\centering
\footnotesize
\renewcommand{\arraystretch}{0.8}
\begin{threeparttable}
\caption{Grouped Results with Holistic Model}
\label{tab:grouped_results}
\begin{tabular}{llccccc}
\toprule
Solvent               & Measurement & N      & MAE    & $R^2$   & Mean pKa & SD \\
\midrule
acetone (5)           & PTM (23) & 362    & 0.2400 & 0.9219  & 6.60     & 1.18 \\
acetonitrile (67)     & PTM (23) & 464    & 0.3008 & 0.9600  & 9.51     & 2.22 \\
dmso (50)             & PTM (23) & 1229   & 0.3499 & 0.9292  & 5.89     & 1.82 \\
acetonitrile (34)     & PTM (23) & 389    & 0.4352 & 0.9669  & 10.55    & 3.02 \\
dmso (50)             & UV (26)            & 257    & 0.4905 & 0.9628  & 6.54     & 3.17 \\
methanol (30)         & PTM (23) & 333    & 0.5015 & 0.9642  & 10.71    & 3.68 \\
\textbf{water (60)}            & \textbf{PTM (23)} & \textbf{5587}   & \textbf{0.5315} & \textbf{0.9339}  & \textbf{6.24}     & \textbf{2.69} \\
\textbf{water (60)}            & \textbf{NMR (19)}           & \textbf{954}    & \textbf{0.5485} & \textbf{0.9348}  & \textbf{6.38}     & \textbf{2.67} \\
\textbf{water (60)}            & \textbf{UV (26)}            & \textbf{2611}   & \textbf{0.5528} & \textbf{0.9267}  & \textbf{6.35}     & \textbf{2.68} \\
\bottomrule
\end{tabular}

\vspace{0.1em}
\flushleft
\scriptsize
\renewcommand{\arraystretch}{0.6}
\textbf{Abbreviations:} DMSO = dimethyl sulfoxide; PTM = potentiometry; UV = ultraviolet spectrophotometry; NMR = nuclear magnetic resonance; \\
N = number of data points; MAE = mean absolute error; SD = standard deviation. \\
Numbers in parentheses: solvent ID / method ID.

\end{threeparttable}
\end{table}

The model optimization paradigm established above yielded multiple high-performance regressors suitable for cross-validation. Stratified analysis by solvent–measurement type combinations revealed that, while the overall $R^2$ exceeded 0.9 for all subsets—comparable or superior to our previously reported holistic-purpose predictors—the mean absolute error varied systematically with the diversity of the underlying pKa distribution. Combinations with the narrowest pKa standard deviations (SD), such as acetone–potentiometric measurements (PTM), DMSO–PTM, and DMSO–PTM, achieved the lowest MAE (0.24–0.35). In contrast, for aqueous-phase measurements, which arguably hold great practical value \cite{39, 40}, the MAE of the holistic model remained in the range of 0.5315–0.5528.

\begin{figure}
\centering
\includegraphics[width=0.85\textwidth]{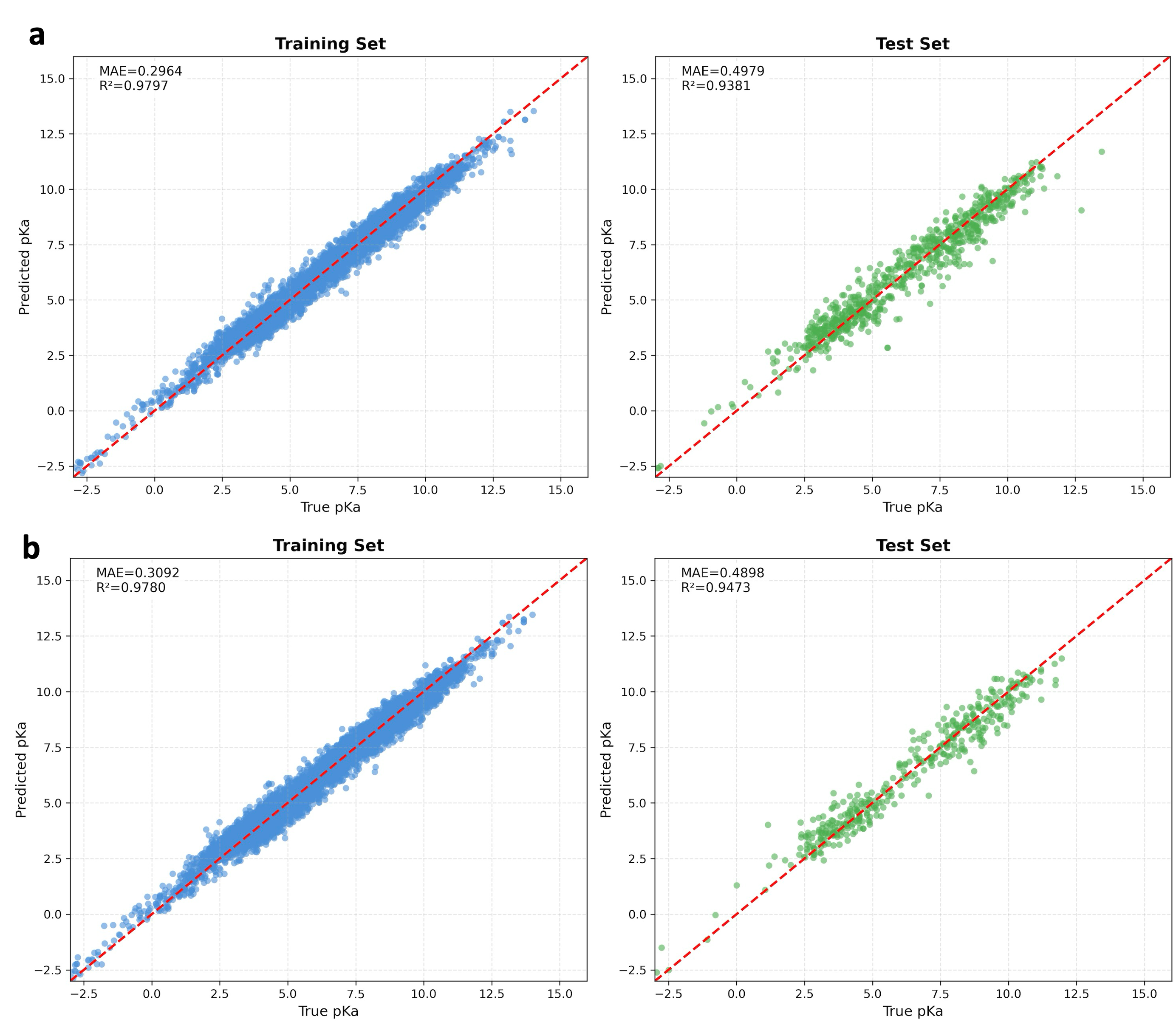}  
\caption{Aqueous Specialized Model. (a) WaterBest, (b) WaterOptimized}
\label{fig3}
\end{figure}

Motivated by this gap, we applied the identical feature-selection–model-optimization workflow to construct a specialized aqueous-phase predictor. The solvent condition was fixed to aqueous, and the solvent categorical encoding features were consequently removed from the input vector, while measurement-related features were retained. We selected the three measurement types with the largest data volumes to ensure robust regression, and extended the feature selection process to the 208 RDKit descriptors, which had been retained without pruning in the holistic model. Among the candidate models generated during this process, one variant exhibited clear overfitting (with a training MAE of 0.1074 and $R^2$ of 0.9973 but a test MAE of 0.4795 and $R^2$ of 0.9321, Appendix FIG. A3) was excluded. The two retained selected models, denoted WaterBest and WaterOptimized, demonstrated robust generalization, with training MAE around 0.3 and test MAE around 0.49 (FIG. 3, detailed model training parameters and results for different measurements are provided in the Appendix TABLE A6 \& FIG. A5).

The two aqueous models retained 89.4–93.3$\%$ of the original 208 RDKit descriptors, confirming that even carefully curated physicochemical descriptors benefit from task-specific pruning. Notably, the overlap between the two feature sets was modest: only 722 fingerprint bits were shared, corresponding to a mere 48$\%$  overlap. In contrast, 180 descriptors (86.5$\%$  of the total 208) were common to both models, with WaterBest retaining 6 unique descriptors and WaterOptimized retaining 14 unique descriptors (Appendix A3).

\subsection{Data Augmentation with Cross-Model Validation}

The two aqueous-specialized models, WaterBest and WaterOptimized, were applied to a randomly selected subset of over 100,000 molecules from the ZINC‑250K database (dictated by the computationally expensive descriptor generation required for full‑library processing and the already‑observed tendency of the augmented data toward an increasingly normal distribution).The pairwise predictions from the two models yielded 134,633 data points with a Pearson correlation of 0.9285 and a mean absolute deviation of 0.2772, indicating strong mutual consistency (FIG. 4a). To ensure annotation reliability, we discarded sample pairs with an absolute difference exceeding 0.5, resulting in a high‑consensus subset of 114,399 molecules, for which the WaterOptimized predictions were retained as the final pseudo‑labels.
\begin{figure}[H]
\centering
\includegraphics[width=0.7\textwidth]{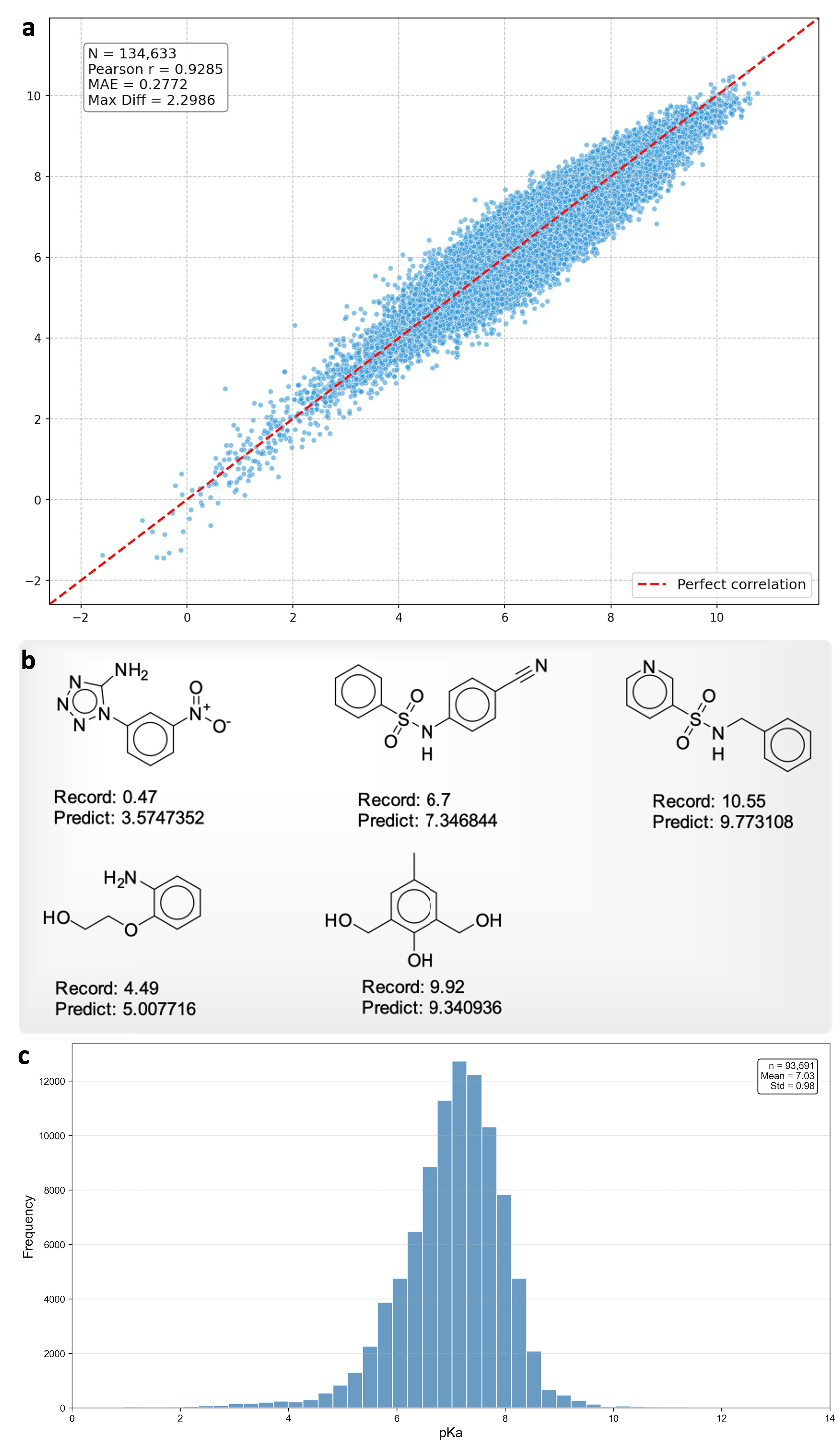}  
\caption{Regression‑Driven Data Augmentation. (a) Comparison of models, (b) Few anomalies, (c) Final distribution}
\label{fig4}
\end{figure}
To further validate this augmented dataset, we cross‑referenced it with the recently published aqueous pKa compilation by Dávid Bajusz et al \cite{9}. Among 637 overlapping molecules, only one exhibited a substantial deviation of 3.1 pKa units; the remaining cases with discrepancies above 0.5 were all below 0.7, confirming strong external agreement(FIG. 4b, The Predict data refer to the pKa values predicted by the WaterOptimized model). Additionally, we subjected the consensus predictions to a multi‑model filtering pipeline that included the updated holistic model developed in this work and the previously established SPOC predictor. Only molecules for which all independent models agreed within 0.5 pKa units were retained. This stringent multi‑model consensus procedure yielded approximately 100,000 high‑confidence aqueous pKa annotations, substantially enriching the available data for this therapeutically relevant domain and demonstrating the feasibility of large‑scale, high‑quality pseudo‑label generation through specialized regression models.

Notably, the augmented dataset presents an approximately normal distribution (FIG. 4c), which prompted us to halt further data augmentation. After continuous augmentation, molecular pKa values became densely distributed around the mean, whereas the sparse tails and extreme pKa regions suffered from data scarcity. Considering the inherent randomness and high computational cost of descriptor computation in passive pseudo-labeling, we accordingly revised our technical roadmap. We discarded the conventional passive data augmentation scheme based on blind prediction, and proposed a novel molecular generation strategy tailored for sparse property distribution. This method is dedicated to supplementing molecules with extreme pKa values in the tail regions, which cannot be effectively enriched by traditional regression-based augmentation approaches.

Another striking finding was uncovered through validated SHAP analysis. The features dominating pKa prediction are not simple functional group counts such as fr\_COO2, but BCUT2D descriptors, atom-pair fingerprints and topological torsion fingerprints (Appendix FIG. A4). All these features are closely associated with molecular skeletons and overall topological structures. This finding strongly indicates that latent representations based on molecular scaffolds contain intrinsic information relevant to pKa properties, even in the absence of fine-grained electronic descriptors.It also establishes a feasible basis for supplementing data within the sparse pKa intervals derived from regression fitting in our follow-up work.

\section{Non-Quantum Generative Approaches}
The uncovered limitation of passive regression-based data augmentation for extreme pKa regimes motivates us to explore generative models as a complementary approach toward comprehensive pKa property profiling. Before turning to quantum‑inspired discrete optimization, we systematically evaluated conventional continuous latent variable models, which represented the predominant non‑quantum paradigm for property‑guided molecular generation.
Following architectures and training protocols reported in the literature \cite{26}, we constructed variational autoencoders (VAEs) that mapped molecular string representations into continuous latent spaces. Two molecular encoding schemes—SMILES and SELFIES—were each combined with three latent dimensionalities (16, 32, and 48), yielding six model configurations in total. Critically, we explicitly incorporated pKa values as conditioning inputs during training: the molecular representation and its corresponding pKa label were jointly encoded, and the decoder was conditioned on the same pKa value, so that the latent space was directly structured by the property of interest. All models were trained on the iBonD molecular corpus, and both reconstruction fidelity and latent space regularity were monitored throughout hyperparameter optimization.
After extensive tuning, the continuous VAE models attained several capabilities consistent with prior reports. They generated chemically valid and structurally novel molecules when sampling from the learned latent distributions. Interpolation in the latent space produced chemically plausible molecules with properties that varied smoothly between endpoints, and arithmetic operations enabled property‑guided editing of molecular scaffolds. In densely populated pKa regions, the models performed well: randomly generated molecular distributions broadly resembled the original training data distribution, representing a clear improvement over the regression‑based pseudo‑labeling approach.
However, all six configurations consistently failed when tasked with targeted generation toward the extreme pKa tails. Even the conditionally trained variants, where pKa was explicitly fed as a conditioning signal, could not steer the decoder toward the sparse extreme regimes. Optimization in the latent space—whether driven by gradient‑based traversal along a property predictor or by conditioning on desired pKa ranges—invariably either collapsed the output toward the densely populated central region or produced invalid, undecodable strings. This behavior proved independent of latent dimensionality, string representation, and regularization strength, confirming the fundamental dilemma outlined earlier: stronger regularization induced posterior collapse and discarded chemically meaningful variation, while weaker regularization preserved reconstruction but yielded a noisy, unstructured latent space in which exploration remained effectively blind.
These systematic failures across multiple non‑quantum generative configurations underscored a critical limitation of continuous latent representations: without a principled combinatorial search mechanism, they could not escape the dominant central distribution to reach the sparse property extremes. This interim conclusion directly motivated the discrete binary bottleneck and the quantum‑inspired optimization framework developed in the following section, where the generative task was recast as a combinatorial optimization problem amenable to coherent Ising machines. Detailed training configurations, hyperparameter settings, and full generative performance metrics for all six continuous VAE models were provided in the Appendix A4.

\section{Engineered Models-Quantum Optimization}

\subsection{Discrete Latent Space Construction \textit{via} SELFIES Vocabulary Encoding}
To circumvent the inherent limitations of continuous latent generative models in sampling molecules with sparse tail pKa, we constructed an engineered generative framework tailored for quantum and quantum-inspired hardware. The core design adopted a discrete latent paradigm that standardizes molecular representation into fixed-bit binary codes, enabling direct compatibility with physical quantum annealers and coherent Ising machines (CIMs).
Herein, we employed SELFIES as the molecular string grammar to ensure structural validity and chemical robustness. All molecular graphs were first converted into standardized SELFIES sequences, followed by vocabulary-based token embedding and fixed-length padding to unify sequence dimensionality for batch-wise Transformer training.
The padded token sequences integrated with positional encoding were fed into a symmetric Transformer encoder to extract hierarchical structural and physicochemical features. To convert continuous sequence embeddings into quantum-ready representations, we adopt a differentiable Gumbel-Softmax binary bottleneck, encoding high-dimensional semantic features into a compact  128-dimensional binary latent spacedenoted as $\bm{x}\in\{0,1\}^{128}$. During training, the Gumbel-Softmax temperature is linearly annealed 
to balance gradient estimability and binary constraints; at inference, one-hot argmax binarization yields deterministic 128-bit latent codes, where each bit encodes disentangled molecular substructures and property contributions.

The Transformer-VAE is optimized with a composite loss consisting of reconstruction loss and KL divergence regularization toward an isotropic Bernoulli prior:
\begin{equation}
\mathcal{L}_{\text{total}} = \mathcal{L}_{\text{recon}} + \beta \cdot \max\left(0, \mathcal{L}_{\text{KL}} - \lambda_{\text{fb}}\right)
\end{equation}

A $\beta$-annealing strategy is applied to gradually strengthen latent regularization, while a free-bit threshold $\lambda_{\text{fb}}$ avoids over-regularization and chemical semantic erosion. Trained on unlabeled molecular datasets including iBonD,  ZINC and Enamine, the model achieves superior molecular reconstruction accuracy and valid novel molecule generation performance. The strong linear correlation between the latent space and the synthetic accessibility (SA) score of molecular additive properties \cite{41, 42, 43} further verifies the physical interpretability of the 128-dimensional binary latent space (Appendix FIG. A13). Detailed training specifications are provided in the appendix A5.

\paragraph*{Pseudocode.}
Algorithm summarizes the entire encoding and decoding process for a single molecule during training.

{\footnotesize
\setlength{\parskip}{0pt}
\setlength{\parindent}{0pt}
\renewcommand{\baselinestretch}{0.92}\selectfont

\noindent\rule{\linewidth}{0.8pt}

\vspace{2pt}
\textbf{Binary SELFIES Autoencoder Training and Reconstruction}
\vspace{2pt}

\noindent\rule{\linewidth}{0.4pt}

\textbf{Input:}
SELFIES string $S$;
maximum sequence length $L_{\max}$;
temperature $\tau$;
free-bit threshold $\lambda_{\mathrm{fb}}$;
KL weight $\beta$.

\textbf{Output:}
Reconstructed SELFIES $\hat{S}$;
total loss $\mathcal{L}_{\mathrm{total}}$.

\vspace{2pt}
\noindent\rule{\linewidth}{0.4pt}

\textbf{Encoding}

1. Tokenize $S$ into $(s_1,\ldots,s_L)$ and pad to length $L_{\max}$ using [PAD].

2.
\[
H^{(0)}
=
\mathrm{Embed}(\tilde{s})
+
\mathrm{PE}
\]

3. \textbf{for} $\ell=1$ to $N$ \textbf{do}

\hspace*{1em}
$H' \leftarrow \mathrm{LayerNorm}(H^{(\ell-1)})$

\hspace*{1em}
$A \leftarrow \mathrm{MultiHeadAttention}(H',H',H')$

\hspace*{1em}
$H_{\mathrm{interim}}
\leftarrow
H^{(\ell-1)}
+
\mathrm{Dropout}(A)$

\hspace*{1em}
$H_{\mathrm{ffn}}
\leftarrow
\mathrm{SwiGLU}
(
\mathrm{LayerNorm}
(H_{\mathrm{interim}})
)$

\hspace*{1em}
$H^{(\ell)}
\leftarrow
H_{\mathrm{interim}}
+
\mathrm{Dropout}(H_{\mathrm{ffn}})$

4. \textbf{end for}

5. $c \leftarrow \mathrm{AttentionPooling}(H^{(N)})$

6. $q_{\mathrm{logits}}
\leftarrow
cW_{\mathrm{to\_latent}}$

\textbf{Binarization}

7. \textbf{for} $j=1$ to $D_L$ \textbf{do}

\hspace*{1em}
$p_{j,0}\leftarrow1,\;
p_{j,1}\leftarrow\exp(q_{\mathrm{logits},j})$

\hspace*{1em}
Sample
$g_0,g_1\sim\mathrm{Gumbel}(0,1)$

\hspace*{1em}
\[
\pi_{j,k}
=
\frac{
\exp((\log p_{j,k}+g_k)/\tau)
}{
\sum_m
\exp((\log p_{j,m}+g_m)/\tau)
}
\]

\hspace*{1em}
Forward:
$x_j\leftarrow\mathrm{onehot\_argmax}(\pi_j)$

\hspace*{1em}
Backward:
use $\pi_j$ as continuous relaxation

8. \textbf{end for}

9. $x=(x_1,\ldots,x_{D_L})$

\textbf{Decoding}

10. $z\leftarrow xW_{\mathrm{from\_latent}}$

11. $M_x\leftarrow \mathrm{Repeat}(z,L_{\max})$

12. Initialize $y_1=[\mathrm{SOS}]$

13. \textbf{for} $t=2$ to $L_{\max}$ \textbf{do}

\hspace*{1em}
Run decoder layers with masked self-attention and cross-attention to $M_x$

\hspace*{1em}
$l_t\leftarrow\mathrm{Linear}(h_{\mathrm{dec},t})$

\hspace*{1em}
$\hat{s}_t\leftarrow\arg\max(l_t)$

\hspace*{1em}
\textbf{if}
$\hat{s}_t=[\mathrm{EOS}]$
\textbf{then break}

14. \textbf{end for}

\textbf{Loss Computation}

15.
\[
\mathcal{L}_{\mathrm{recon}}
=
-\sum_t
\log p(s_t^{\mathrm{true}}
\mid
\hat{s}_{<t},x)
\]

16.
\[
\mathcal{L}_{\mathrm{KL}}
=
\sum_{j=1}^{D_L}
\left[
\sigma(q_j)
\log\frac{\sigma(q_j)}{0.5}
+
(1-\sigma(q_j))
\log\frac{1-\sigma(q_j)}{0.5}
\right]
\]

17.
\[
\mathcal{L}_{\mathrm{total}}
=
\mathcal{L}_{\mathrm{recon}}
+
\beta
\max
\left(
0,
\mathcal{L}_{\mathrm{KL}}
-
\lambda_{\mathrm{fb}}
\right)
\]

18. Return $(\hat{S},\mathcal{L}_{\mathrm{total}})$

\vspace{2pt}
\noindent\rule{\linewidth}{0.8pt}
}

Notably, the 128-bit fixed binary design is purposely optimized for physical quantum hardware deployment. Binary latent variables naturally align with qubit states, Ising spins and QUBO binary variables, enabling direct one-to-one mapping between latent bits and computing units of quantum annealers or physical CIMs. This design allowed seamless migration from classical generative modeling to physical quantum sampling.

Overall, the integrated pipeline unifies SELFIES tokenization, vocabulary encoding, sequence padding, Transformer feature learning and Gumbel-Softmax discretization into an end-to-end module. It avoids the posterior collapse of generative models, and provides a hardware-friendly latent foundation for subsequent factorization machine surrogate fitting, QUBO matrix construction, and quantum-driven sparse-pKa molecular generation.

\subsection{Factorization Machine Surrogate Fitting with Long-Tail Optimization}
To construct an explicit energy landscape linking the discrete latent space to pKa, we trained a factorization machine (FM) as a surrogate predictor. The FM models pairwise feature interactions between binary latent bits, enabling fitting of nonlinear relationships between discrete representations and pKa values. 

The FM directly predicts the pKa value from the binary latent code $\mathbf{x} \in \{0,1\}^{128}$ as a sum of linear and pairwise interaction terms:
\begin{equation}
\hat{y}(\mathbf{x}) = w_0 + \sum_{i=1}^n w_i x_i + \sum_{i=1}^n \sum_{j=i+1}^n (\mathbf{v}_i \cdot \mathbf{v}_j)\, x_i x_j,
\label{eq:fm_ij}
\end{equation}
where $\hat{y}(\mathbf{x})$ is the predicted pKa, $w_0$ is the bias, $\mathbf{w} \in \mathbb{R}^n$ are linear weights, and $\mathbf{V} \in \mathbb{R}^{n \times k}$ is the latent factor matrix with each feature $i$ corresponding to a latent vector $\mathbf{v}_i$.

For later algebraic manipulation, it is convenient to symmetrize the pairwise sum. Because each pair $(i,j)$ with $i<j$ appears once, writing it as a double sum over all ordered pairs would count each interaction twice; this double‑counting is corrected by a factor $1/2$ and an explicit exclusion of self‑interactions ($i \neq j$):
\begin{equation}
\hat{y}(\mathbf{x}) = w_0 + \sum_{i=1}^n w_i x_i + \frac{1}{2} \sum_{i=1}^n \sum_{\substack{j=1 \\ j \neq i}}^n (\mathbf{v}_i \cdot \mathbf{v}_j) x_i x_j.
\label{eq:fm_pred}
\end{equation}

For efficient training, we use an algebraically equivalent reformulation that reduces the complexity from $O(n^2)$ to $O(nk)$. The equivalence relies on expanding the sum over all $i,j$ (including $j=i$) and then subtracting the self‑interaction terms:
\begin{equation}
\frac{1}{2} \sum_{i=1}^n \sum_{\substack{j=1 \\ j \neq i}}^n (\mathbf{v}_i \cdot \mathbf{v}_j) x_i x_j 
= \frac{1}{2} \sum_{i=1}^n \sum_{j=1}^n (\mathbf{v}_i \cdot \mathbf{v}_j) x_i x_j - \frac{1}{2} \sum_{i=1}^n (\mathbf{v}_i \cdot \mathbf{v}_i) x_i^2.
\label{eq:fm_expand}
\end{equation}

Using the identity $\sum_{i=1}^n\sum_{j=1}^n (\mathbf{v}_i \cdot \mathbf{v}_j) x_i x_j = \sum_{f=1}^k \big[(\sum_{i=1}^n v_{i,f} x_i)^2 - \sum_{i=1}^n v_{i,f}^2 x_i^2 \big]$, we obtain the linear‑complexity form:
\begin{equation}
\hat{y}(\mathbf{x}) = w_0 + \sum_{i=1}^n w_i x_i + \frac{1}{2} \sum_{f=1}^k \left[ \Bigl(\sum_{i=1}^n v_{i,f} x_i\Bigr)^2 - \sum_{i=1}^n v_{i,f}^2 x_i^2 \right].
\label{eq:fm_linear}
\end{equation}

Note that the $\sum_{i=1}^n (\mathbf{v}_i \cdot \mathbf{v}_i) x_i^2$ term in Eq.~\eqref{eq:fm_expand} corresponds exactly to $\sum_{i=1}^n \sum_{f=1}^k v_{i,f}^2 x_i^2$, so the subtraction in Eq.~\eqref{eq:fm_linear} correctly removes the self‑interactions. Thus, Eqs.~\eqref{eq:fm_pred} and \eqref{eq:fm_linear} are mathematically identical; we adopt Eq.~\eqref{eq:fm_linear} for training and inference to ensure efficiency without approximation error.

We used a hybrid training corpus combining the high-quality aqueous potentiometric measurement (PTM) subset from the iBonD experimental database and the high-confidence aqueous pKa pseudo-labels generated in Section II.

Initial evaluations of standard quadratic FMs with factorization dimensions k = 8-128 yielded test set \(R^2 < 0.4\) and severe underfitting in extreme pKa tails. Higher-order surrogates amenable to QUBO extraction, including cubic FM and conditional Boltzmann machines, showed no improvement over quadratic FMs and exhibited worse tail generalization.

To address long-tail sparsity, we first implemented pKa-bucket-specific sample weighting, which improved tail prediction but degraded overall \(R^2\) to below 0.2. We then introduced a balanced loss function with a soft \(R^2\) constraint to resolve the trade-off between tail performance and global fit:
\begin{equation}
\mathcal{L} = \alpha \cdot \frac{1}{N_{\text{total}}} \sum_{i=1}^{N_{\text{total}}} w_i \big(y_i - \hat{y}_i\big)^2 - (1-\alpha) R^2,
\end{equation}

Here, $y_i$ and $\hat{y}_i$ denote experimental and predicted pKa values, respectively; $w_i$ is the bucket-specific sample weight defined as 
\[
w_i = \sqrt{\frac{N_{\text{total}}}{n_{\text{bucket},i} \times N_{\text{buckets}}}}
\]

where $N_{\text{total}}$ is the total number of training samples, $n_{\text{bucket},i}$ is the sample count in the $i$-th pKa bucket, and $N_{\text{buckets}}$ is the total number of buckets). Attempts with weaker reweighting, such as cube-root scaling (i.e., $w_i \propto n_{\text{bucket},i}^{-1/3}$), failed to correct the extreme tail underfitting. $R^2$ is the coefficient of determination, and $\alpha \in [0, 1]$ controls the trade-off between weighted prediction error and global fitting accuracy. After systematic hyperparameter tuning, we identified the optimal configuration for our FM surrogate: the trade-off parameter $\alpha = 0.3$ and factorization inner dimension $k = 64$ yielded the best predictive performance (full pipeline and representative results are summarized in FIG.5, the full training algorithm is summarized in Table~\ref{tab:alg_fm_train}).

\begin{figure}[H]
\centering
\includegraphics[width=0.98\textwidth]{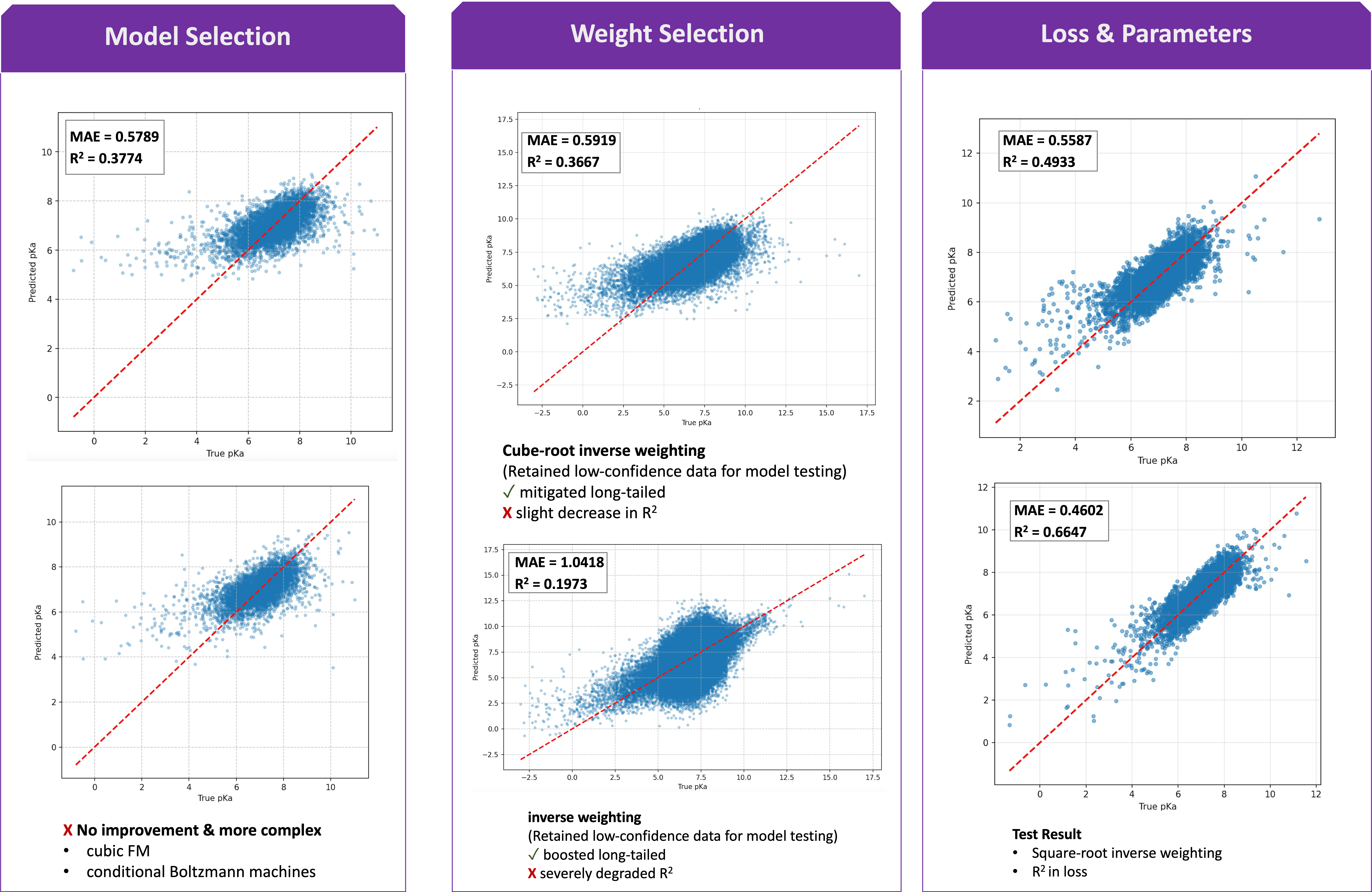}  
\caption{Pipeline \& Representative Results}
\label{fig5}
\end{figure}

For engineering efficiency, we reimplemented this complex loss pipeline in PyTorch with the Adam optimizer, achieving a significant speedup over NumPy. To enable rapid post-hoc evaluation of quantum sampling outputs, we additionally trained a suite of regression models that directly map binary latent codes to pKa values for downstream engineering use. Among all tested architectures, extensively optimized XGBoost regressors achieved the highest performance, yet even the best variant only yielded a test set \(R^2\)  just above 0.7 (detailed results in Appendix A6), a result that reflects the enbedded capacity limit of the 128‑dimensional discrete latent space. Operating under the same constraint, our FM model was pushed to its performance limits through extensive model tuning, and subsequent quantum sampling experiments confirmed that this level of surrogate fidelity was both indispensable and adequate for our target task. This engineering reality also underscores that the current end‑to‑end workflow still relies on high‑cost descriptor‑based models for final validation.

\begin{table}[t]
\centering
\caption{Long-tail Optimized FM Training}
\label{tab:alg_fm_train}
\small
\renewcommand{\arraystretch}{0.8} 
\begin{tabular}{l}
\toprule
\textbf{Input:} \\
\quad Binary latent vectors $\boldsymbol{z}$ \\
\quad Experimental/pseudo-label pKa values $y$ \\
\midrule
\textbf{Output:} \\
\quad Optimized FM parameters \\
\midrule
\textbf{Algorithm Steps:} \\
1.~~Partition pKa values into buckets \\
2.~~Compute bucket weights: $w_i \propto 1 / n_{\text{bucket}}^\alpha$ \\
3.~~Initialize FM parameters: \\
\qquad linear weights $\boldsymbol{w}$ \\
\qquad latent vectors $\boldsymbol{V}$ \\
4.~~\textbf{while} not converged \textbf{do} \\
5.~~\quad Compute FM prediction $\hat{y}$ \\
6.~~\quad Compute weighted MSE \\
7.~~\quad Compute soft $R^2$ regularization \\
8.~~\quad Update parameters using Adam \\
9.~~\textbf{end while} \\
10.~Return optimized FM model \\
\bottomrule
\end{tabular}
\end{table}

\subsection{From FM Surrogate to Quantum-Inspired Solvers}

The quadratic structure of the optimized FM model enabled direct, lossless extraction of a quadratic unconstrained binary optimization (QUBO) matrix, which served as the input to both classical simulated quantum annealers and physical coherent Ising machines. Using the FM prediction $\hat{y}(\mathbf{x})$ defined in \eqref{eq:fm_pred}, we cast molecular generation as a minimization problem compatible with quantum optimizers by defining the QUBO energy function as the negative of the FM prediction shifted by the bias:

\begin{equation}
E(\mathbf{x}) = \mathbf{x}^T \mathbf{Q} \mathbf{x} = -\hat{y}(\mathbf{x}) + w_0
\end{equation}
The QUBO matrix $\mathbf{Q}$ was constructed directly from the trained FM parameters:
\begin{equation}
\mathbf{Q} = -\left( \mathrm{diag}(\mathbf{w}) + \mathbf{V} \mathbf{V}^T \right)
\end{equation}
where diagonal elements $Q_{ii} = -w_i$ correspond to linear terms, and off-diagonal elements $Q_{ij} = -\mathbf{v}_i \cdot \mathbf{v}_j$ (for $i < j$) capture pairwise feature interactions. This transformation establishes a monotonic inverse relationship between QUBO energy and pKa: minimizing $E(\mathbf{x})$ drives the search toward high pKa values (very weak acids), while maximizing $E(\mathbf{x})$ drives the search toward low pKa values (strong acids).

\begin{figure}[H]
\centering
\includegraphics[width=0.98\textwidth]{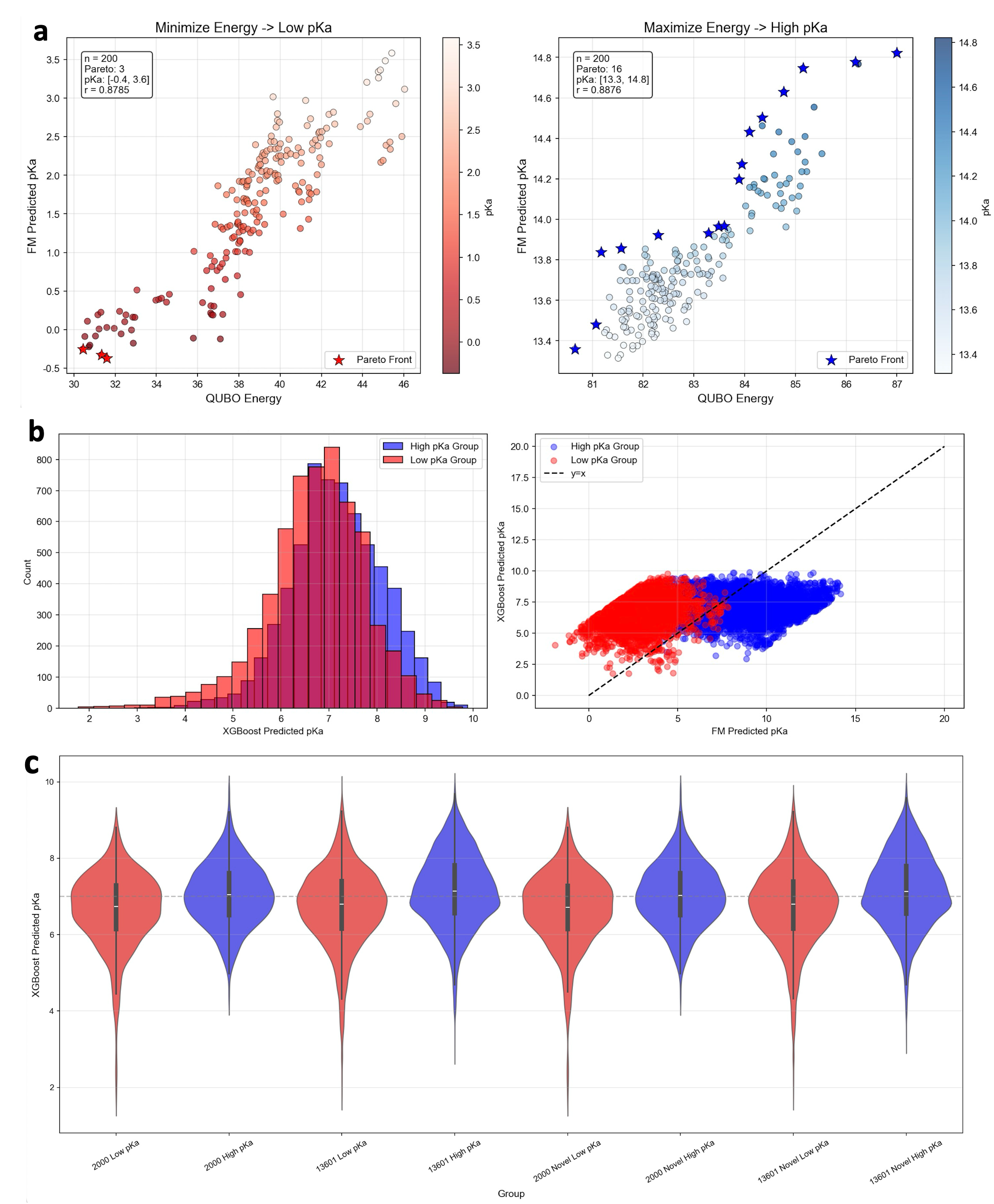}  
\caption{Preliminary Results Using SA Solution}
\label{fig6}
\end{figure}

We first validated the pipeline using a classical simulated quantum annealer implemented \textit{via} the Kaiwu SDK. For the long-tail-optimized FM model, we observed strong linear correlations between QUBO energy and FM-predicted pKa values across both extreme regions: a Pearson correlation coefficient of $r = 0.8785$ for the low-pKa region ($[-0.4, 3.6]$) and $r = 0.8876$ for the high-pKa region ($[13.3, 14.8]$) when sampling 200 solutions (FIG. 6a). These correlations improved further with increased sampling size, exceeding 0.92 for 10,000+ samples. In contrast, QUBO matrices extracted from FM variants without proper long-tail optimization exhibited negligible correlations in extreme regions, confirming the critical importance of our tailored surrogate design (detailed parameters and results provided in the Appendix A7).

Solutions obtained from the annealer were decoded back to molecular structures via the pre-trained SELFIES decoder. Notably, the fraction of samples possessing novel latent representations relative to the full train-test dataset together with valid final SMILES exceeds 90 \% for the simulated quantum annealer (distinct latent representations do not always yield unique molecules, but this metric remains a reasonable indicator of molecular novelty; details are discussed later). This substantially outperforms the novel-and-valid generation ratio from RNN-VAE in Section III. 

For each generated molecule, we computed the full set of molecular descriptors and evaluated its pKa using our ensemble of high-accuracy regression models (WaterOptimized and WaterBest). A multi-stage confidence filtering pipeline was applied: only molecules with inter-model agreement within 1.0 pKa unit were retained as high-confidence candidates. This workflow yielded a substantial collection of novel molecules in the sparse pKa tails, achieving moderate improvements over both the regression-based pseudo-labeling strategy and conditional RNN-VAE generative model, which struggle with labor-intensive sampling to populate these sparse regions. However, the distribution of high-confidence candidates still exhibited a residual normal tendency around $\mathrm{p}K_\mathrm{a}=7$ (FIG. 6 b\&c). On one hand, scaffold-dependent prediction accuracy is limited. Although sampled points agree closely with FM predictions and cluster within FM-targeted extreme regions, most candidates cannot pass joint validation from multiple high-precision regression models. On the other hand, this feature plausibly stems from relatively uniform sampling of the simulated quantum annealer. In view of the obtained results and existing limitations, we further investigate sampling performance using physical quantum CIM hardware.

\subsection{CIM Optimization with 8-bit and 14-bit Precision}

\begin{figure}[H]
\centering
\includegraphics[width=0.98\textwidth]{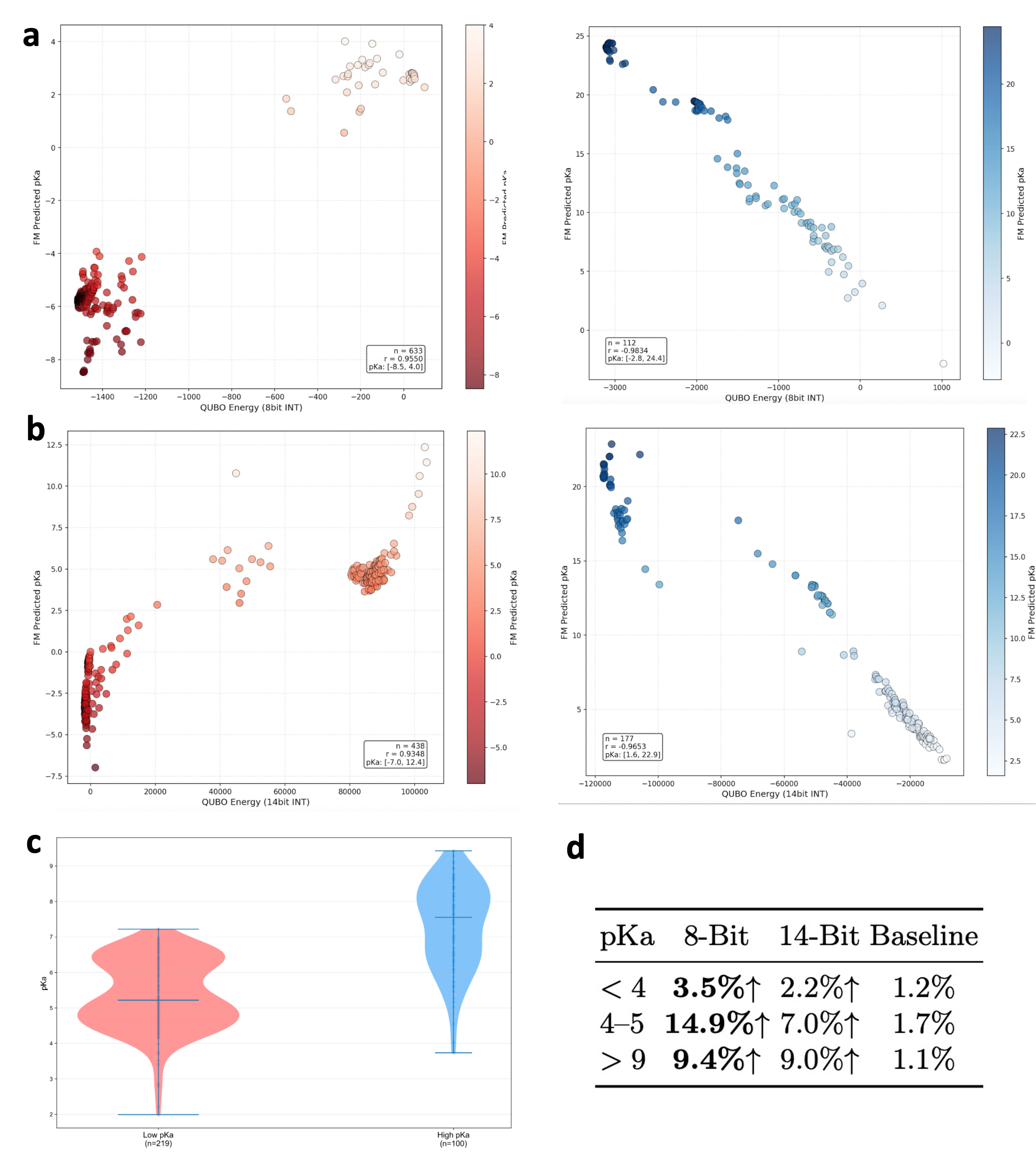}  
\caption{Results with CIM}
\label{fig7}
\end{figure}

We deployed the QUBO matrices extracted from the best-performing FM surrogate on a physical coherent Ising machine (CIM) operating at both 8-bit and 14-bit integer precision. Since the CIM hardware natively solves only minimization problems, we implemented a matrix inversion strategy to enable bidirectional extreme-value sampling: the original QUBO matrix was used to search for high-pKa molecules (minimum energy solutions), while the negated matrix $-\mathbf{Q}$ was used to search for low-pKa molecules (equivalent to maximizing the original energy). Both configurations were operated exclusively in sampling mode rather than deterministic Quota mode, as the stochastic nature of sampling substantially enriched the diversity of latent-space-to-molecule mappings and enabled discovery of more extreme energy solutions.

The physical CIM delivered orders-of-magnitude faster sampling than classical simulated annealing, acquiring 1,000 candidate solutions in under 30 ms across all precision settings. For the long-tail-optimized FM model, the Pearson correlation between FM-predicted pKa and QUBO energy remained stably above 0.9 for both low-pKa and high-pKa regions under both precision modes. The sampled results are mainly concentrated at extreme energy points, and part of these extreme values originate from out-of-distribution (OOD) samples introduced by reverse data, which will be elaborated in subsequent ablation experiments (FIG. 7a and FIG. 7b present the results under 8-bit and 14-bit precision, respectively).

\begin{table}[h!]
\centering
\caption{CIM-based Sparse-pKa Molecular Sampling}
\label{tab:alg_cim_sampling}
\small
\renewcommand{\arraystretch}{0.75}
\begin{tabular}{l}
\toprule
\textbf{Input:} \\
\quad QUBO matrix $\boldsymbol{Q}$ \\
\quad Precision mode $p$ \\
\quad Number of samples $N$ \\
\midrule
\textbf{Output:} \\
\quad High-confidence molecules $\boldsymbol{M}$ \\
\midrule
\textbf{Algorithm Steps:} \\
1.~~Configure CIM sampler with precision $p$ \\
2.~~\textbf{for} $k = 1$ to $N$ \textbf{do} \\
3.~~\quad Sample binary solution $\boldsymbol{z}_k$ \\
4.~~\quad Decode $\boldsymbol{z}_k$ into molecule $\boldsymbol{m}_k$ \\
5.~~\quad Validate molecular structure \\
6.~~\textbf{end for} \\
7.~~Predict pKa using: \\
\qquad a) FM surrogate \\
\qquad b) WaterBest \\
\qquad c) WaterOptimized \\
8.~~Retain molecules satisfying: $|\text{pKa}_i - \text{pKa}_j| < \text{threshold}$ \\
9.~~Remove duplicates and invalid structures \\
10.~Return filtered candidate set $\boldsymbol{M}$ \\
\bottomrule
\end{tabular}
\end{table}

A clear trade-off emerged between the two precision configurations. The 14-bit mode provided substantially superior Spearman rank correlation (0.887 versus 0.380 for 8-bit), faithfully preserving the relative ordering of pKa values in the energy landscape, with FM-predicted pKa ranges spanning $[-7.0, 12.41]$ for low-pKa sampling and $[1.6, 22.9]$ for high-pKa sampling across 1,000 solutions. In contrast, the 8-bit mode, despite lower rank fidelity, exhibited higher stochasticity and broader extreme-value coverage, generating solutions spanning $[-8.5, 4.0]$ for low-pKa and $[-2.8, 24.4]$ for high-pKa. This demonstrates that lower precision sacrifices ranking reliability but enhances exploration, enabling more effective probing of the edges of the FM-predicted property landscape. Further analysis revealed that 8-bit quantization acts as an implicit hard-threshold sparsification operator, naturally filtering out approximately 34.5\% of weak coupling edges from the original 406-edge QUBO matrix and resulting in a 266-edge sparse matrix. This natural sparsification may provided favorable regularization compared to both full-precision operation and active edge-selection methods. Notably, the 14-bit configuration retaining 400 edges exhibited slightly lower performance than the 8-bit 266-edge version, suggesting that preserving all weak interactions may introduce more noise than useful information for extreme-value sampling in this task (detailed spherical coupling topology was shown in FIG. 8).

All CIM-generated solutions were decoded to molecular structures via the SELFIES decoder and evaluated using our ensemble of high-accuracy regression models. Multi-model evaluation confirmed that CIM-generated candidates outperformed those from classical simulated annealing in populating high-confidence sparse-pKa regions: the number of high-confidence molecules with pKa $<$ 5 was significantly increased, and a substantial number of candidates with pKa $>$ 9 were also obtained. Although the fraction of genuinely novel molecules within extreme pKa ranges remains constrained, the corresponding proportion has been markedly elevated. The violin plots summarizing final high-confidence screening results are provided in FIG. 7c (full algorithm is summarized in Table~\ref{tab:alg_cim_sampling}), the resulting performance markedly surpasses that from high-gain-based filtering (detailed comparisons provided in Appendix A8). Notably, our core focus lies on the practical cost of discovering extreme-pKa molecules \textit{via} the proposed pipeline. We therefore further quantified the fraction of extreme-range samples among total generated candidates, as all such specimens require costly descriptor extraction in practical engineering workflows. The computed metrics demonstrate clear superiority over all prior generation strategies (FIG. 7d).

\subsection{In-Depth Research \& Future Prospects}

\subsubsection{Absolute Deviation Regression Trials}
We further conducted comparative regression experiments in which the target response was the absolute deviation of the experimental $\mathrm{p}K_\mathrm{a}$ from a fixed reference value of 3:
\[
d_i = |\,y_i - 3\,|,
\]

where $y_i$ is the experimentally measured $\mathrm{p}K_\mathrm{a}$ of molecule $i$. A separate factorization machine was trained to predict $d$ directly from the binary latent code $\mathbf{x}$:
\begin{equation}
\hat{d}(\mathbf{x}) = w_0 + \sum_{i=1}^n w_i x_i + \frac{1}{2} \sum_{i=1}^n \sum_{j=1}^n (\mathbf{v}_i \cdot \mathbf{v}_j) x_i x_j.
\label{eq:fm_abs}
\end{equation}

Note that $\hat{d}(\mathbf{x})$ is the output of a standard quadratic FM fitted to the absolute deviation values; it is not obtained by post‑processing the $\mathrm{p}K_\mathrm{a}$ predictor, and therefore retains the exact quadratic form. Training with this absolute‑deviation response did not significantly optimize the property distribution of high‑confidence molecules obtained \textit{via} quantum solving and generation pipelines, despite yielding better statistical metrics; the apparent improvement originated from the numerical folding of the distribution around $\mathrm{p}K_\mathrm{a} = 3$ rather than from a genuine structure–property mapping (detailed results in Appendix A9). Regression trials on a target that contains a quadratic term in $\mathrm{p}K_\mathrm{a}$ also failed to produce any meaningful fit.

\subsubsection{Negative Data Ablation Experiments}

To further verify the contribution of high-pKa negative samples to latent-space diversity, we performed an ablation test by removing all samples with $\mathrm{p}K_\mathrm{a}>8$ before FM training. After refitting the factorization machine and extracting updated QUBO matrix for CIM sampling, the resulting sparse-pKa molecular distribution was markedly degraded compared with the full-dataset training scheme (details in Appendix A10).

This observation demonstrates that high-pKa extreme samples provide complementary, asymmetric latent-space features for low-pKa counterparts; the latent embedding space does not follow symmetric interaction rules across the full property range. The counterintuitive out-of-distribution (ODD) extreme values on the FM energy landscape essentially originate from antagonistic latent constraints introduced by these high-pKa negative samples. Although retaining full-range data expands the energy landscape beyond initial statistical confidence bounds and yields apparent ODD outliers, removing high-pKa data truncates latent-bit flipping diversity in discrete latent space. Restricted single-bit perturbation destroys complementary opposite-direction constraints and impairs downstream quantum sampling performance. Collectively, joint training on full-spectrum pKa data benefits latent feature complementarity and stabilizes the generative surrogate model.

\subsubsection{Proposed Unified Insight}
\begin{figure}[H]
\centering
\includegraphics[width=1\textwidth]{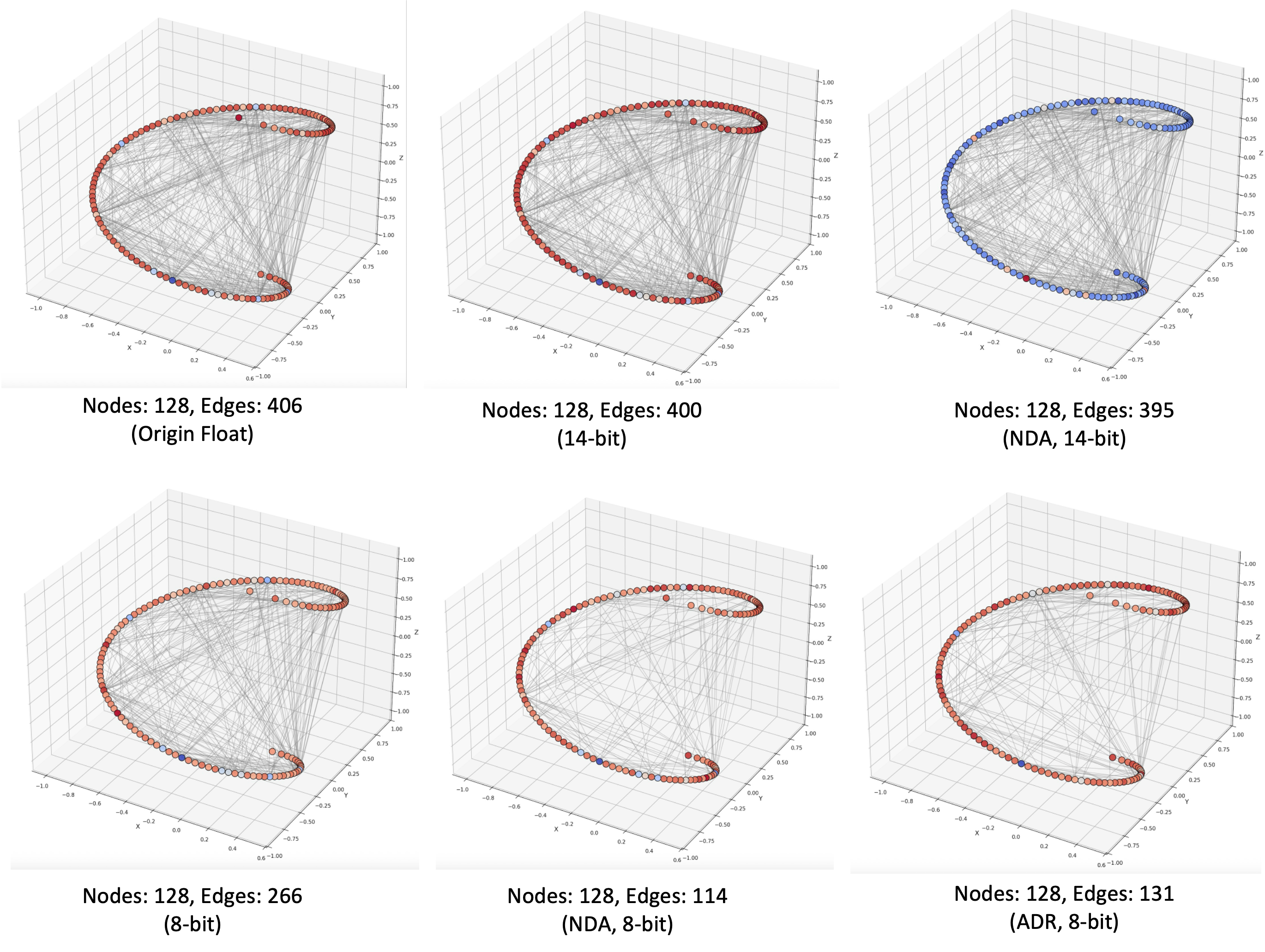}  
\caption{Spherical Coupling Topology of Different QUBO Matrices}
\label{fig8}
\end{figure}

We summarize the experimental results of two previous strategies, namely negative-data ablation (NDA) and Absolute Deviation Regression (ADR) Trials , and rationalize all observed performance discrepancies from the viewpoint of spherical coupling topology. Under 8-bit quantization, both pruning schemes compress the QUBO coupling edges down to 114 and 131. Excessive pruning fractures the intrinsic spherical connectivity skeleton and ruins the global correlation constructed by weak couplings, leading to deteriorated molecular sampling performance.
At the 14-bit quantization setting, nearly 400 coupling edges are retained after pruning, yet no performance gain is achieved relative to the original 14-bit baseline. Analogous to the trend observed at 8 bit, the inherent diversity of the energy landscape is already compromised, and the reserved 395 edges mainly consist of redundant noisy couplings that distort the well-defined spherical topology.
Based on these comparative findings, we draw a tentative unified conclusion: aggressive artificial pruning inevitably corrupts spherical coupling architectures, manifested as excessive sparsification at 8 bit. By contrast, the spontaneous sparsification emerging naturally from 8-bit quantization preserves intact spherical coupling topology and delivers the optimal configuration tailored for sparse-pKa molecular generation.

\subsubsection{Chemical Analysis of Latent Representations}

As mentioned earlier, all latent representations were decoded into actual molecules for multi-model evaluation. The top 100 molecules with low and high pKa values in aqueous phase are presented in the Appendix A12.
In further chemical analysis, we observed that distinct latent codes may correspond to identical molecular structures. Although SELFIES, as a directed molecular encoding scheme, greatly reduces such duplication compared with SMILES, a small number of repeated structures still exist when decoded into undirected molecular graphs. Representative examples are illustrated in FIG.9.

\begin{figure}[H]
\centering
\includegraphics[width=0.9 \textwidth]{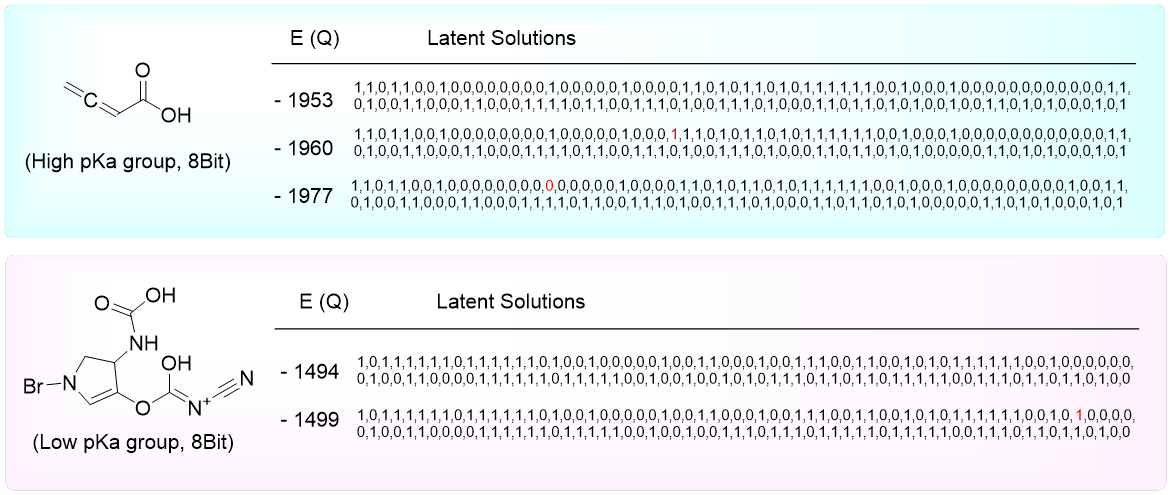}  
\caption{Multiple Latent Codes for Identical Molecules}
\label{fig9}
\end{figure}

Nevertheless, molecules generated from different latent representations exhibit nearly identical energy values throughout our computational workflow. This consistent energy performance strongly verifies the robustness of the proposed pipeline. Specifically, it demonstrates that our strategy of mapping molecules into a 128-dimensional discrete latent space, combined with the long-tail optimization of the factorization machine (FM), can stably characterize the intrinsic molecular properties. Notably, this phenomenon is almost exclusively observed under 8-bit precision, which is closely associated with the unique sampling mechanism and natural sparsification effect of 8-bit quantization. Even when distinct latent vectors are generated during sampling, the core property information of the target molecules remains highly consistent, further proving that our framework is well adapted for sparse-pKa molecular generation tasks.

\subsubsection{Future Prospects}

While the proposed VAE-FM-QUBO-CIM workflow demonstrates advantages in generating sparse-pKa molecules compared to conventional generative approaches, several important limitations remain that warrant further investigation. Most notably, the physical CIM solver successfully discovered a large number of extreme-energy solutions corresponding to the outermost edges of the FM-predicted pKa landscape. However, many of these candidate molecules fall outside the high-confidence evaluation range of our regression-based pKa predictors, as they represent chemical space regions that are vastly underrepresented in all existing experimental datasets. These extreme out-of-distribution candidates have been systematically documented for future reference.

To better harness the physical CIM's capability to abundantly sample extreme-energy solutions and to bring more of these candidates into the high-confidence evaluation window, we attempted to adjust the QUBO energy landscape by translating the linear terms of the trained FM model, which shares the same research purpose as the modification attempt shown in Eq.~\eqref{eq:fm_abs}. This approach was implemented and tested using classical simulated quantum annealing, but it failed to produce a significant increase in the number of high-confidence extreme-pKa molecules. Detailed results of these landscape modification experiments are provided in the Appendix A14.

Given the inherent limitations of computational evaluation for extremely out-of-distribution molecules, we are planning to experimentally synthesize and characterize a subset of these extreme candidates in future work. To prioritize synthetically accessible molecules for experimental validation, we explored the possibility of incorporating the synthetic accessibility (SA) score into the quantum optimization framework. Our previous analysis demonstrated that the 128-dimensional discrete latent space exhibits a strong linear correlation with SA scores, while quadratic interaction terms introduce unnecessary noise and degrade performance. We further rigorously derived the mathematical incompatibility of jointly training an FM model for both long-tail pKa prediction and uniform SA score prediction, as detailed in the Appendix A15. Moreover, we found that directly adding a linear SA term to the QUBO matrix is mathematically equivalent to the linear term translation strategy described above for reshaping the energy landscape. This equivalence indicates that, beyond the expressivity loss caused by discrete latent compression, the QUBO matrix extraction process itself imposes additional representational constraints, and further optimization of this model--quantum coupled generation pipeline is warranted in future work. Within the current framework, the low-cost SA score can still serve as a rapid post-hoc filter for generated candidates. Looking ahead, what may ultimately be required is a fundamentally new, non-additive synthetic feasibility metric---one that is both chemically meaningful and, like high-quality pKa data, exceedingly difficult to acquire at scale.

The current 128‑dimensional discrete latent space, while providing a practical bottleneck for QUBO formulation, inherently suffers from a strong dependence on molecular scaffolds rather than capturing detailed physicochemical information. This scaffold‑biased representation makes precise pKa fitting fundamentally difficult, as subtle electronic and steric effects that govern acid–base equilibria are only indirectly encoded. Nevertheless, this very limitation may become an advantage for exploring novel chemotypes: the coarser, scaffold‑level latent space could promote higher diversity in the generated molecular backbones when performing quantum‑assisted sampling. Consequently, increasing the dimensionality of the discrete latent space (e.g., to 256 or 512 bits) is a logical next step. Such augmentation would allow more fine‑grained physicochemical information to be preserved, potentially improving pKa fidelity, while simultaneously demanding a proportional increase in the number of qubits or spins supported by the quantum hardware. The entire engineering pipeline—from discrete representation learning to QUBO matrix construction and physical CIM execution—will therefore directly benefit from advancements in quantum hardware, especially from the availability of larger‑scale coherent Ising machines with higher bit precision and greater connectivity. Our ongoing work includes systematic exploration of higher‑dimensional latent spaces, with the goal of balancing generative diversity, property predictability, and hardware feasibility.

\section{Conclusion}
The augmentation of high-quality pKa data remains a challenge. To address the limited availability of experimental measurements, large-scale regression-based prediction provides a practical route for constructing substantially enlarged high-confidence pKa datasets from extensive unlabeled molecular libraries. Beyond increasing data volume, such augmentation offers an unprecedented opportunity to examine the large-scale distribution characteristics of molecular pKa values. Our analyses reveal that, although regression-driven augmentation successfully enriches the dominant regions of the pKa distribution and generates abundant high-confidence annotations, sparsely represented extreme-pKa regions emerge with inadequate coverage. This observation suggests that further improvements in pKa data coverage require not only passive prediction-based augmentation but also active exploration strategies capable of targeting rare property regimes. High-fidelity pKa prediction remains essential throughout this process, as accurate property estimation is required for reliable downstream selection, and molecular representation alone is insufficient to guarantee targeted property control.

In this work, we present an engineered models–quantum optimization paradigm that builds upon large-scale high-confidence pKa data augmentation and integrates discrete generative modeling, long-tail surrogate learning, and quantum-inspired combinatorial optimization. In addition to the generation framework, we construct multiple high-accuracy pKa prediction models that provide both the foundation for large-scale data augmentation and a rigorous evaluation backbone for generated candidates.

Our approach overcomes the critical limitations of continuous latent variable models through a synergistic Transformer-VAE architecture with a Gumbel-Softmax binary bottleneck, which maps molecules into a 128-dimensional discrete latent space that avoids posterior collapse and enables direct hardware compatibility. To bridge this discrete space with pKa properties, we develop a tailored factorization machine surrogate optimized specifically for long-tail fitting, resolving the inherent trade-off between tail performance and global accuracy through a combination of bucket-specific sample weighting and a soft \(R^2\)-constrained loss function. This design allows lossless conversion of the molecular generation task into a quadratic unconstrained binary optimization (QUBO) problem, establishing a seamless interface between generative AI and quantum computing.

We demonstrate that physical coherent Ising machines (CIMs) deliver orders-of-magnitude faster sampling and superior extreme-value performance compared to classical simulated quantum annealing, achieving stable Pearson correlations above 0.9 between QUBO energy and predicted pKa in both low- and high-pKa tail regions. Our systematic evaluation of 8-bit and 14-bit precision modes reveals a fundamental trade-off between ranking fidelity and exploration capability, providing critical practical insights for the deployment of quantum-inspired hardware in molecular discovery. 

Notably, this work identifies the first application in generative AI for chemistry that genuinely necessitates quantum-inspired computing, as the rugged energy landscape of sparse property generation plays directly to the unique strengths of quantum annealers and CIMs. As a practically grounded model for generating sparsely labeled molecular properties, the proposed workflow carries meaningful guiding significance for related studies. 

The demonstrated workflow not only produces a substantial collection of novel, high-confidence sparse-pKa molecules, but also complements regression-based data augmentation by extending coverage toward sparsely represented property regions. More broadly, it establishes a generalizable paradigm that combines machine-learning-driven high-confidence data augmentation, distribution-aware property analysis, and quantum-assisted combinatorial optimization, and can be extended to other inherently hard-to-generate molecular properties. By integrating realistic chemical scenarios with advanced machine learning and physical quantum hardware, this work opens new avenues for AI-driven functional molecule discovery and provides a blueprint for the development of next-generation quantum-accelerated scientific computing platforms.

\begin{acknowledgments}
We sincerely thank Academician Hong Ding (Fudan University) for his insightful guidance and inspiration. We are grateful to QBoson Quantum Technology Co., Ltd. for providing the Coherent Ising Machine (CIM) hardware and the Kailua SDK, with special thanks to Liye Miao. We also thank Professor Dunnan Lu and Professor Long Zhang for constructive discussions and generous assistance. We acknowledge Professor Jindong Yang and Academician Jin-Pei Cheng for their contributions to the iBond database,as well as to Professor Sanzhong Luo, Professor Long Zhang, Qi Yang, and Siyuan Liu for their exploratory work on pKa prediction. This work was supported by multiple research practice programs at Tsinghua University.
\end{acknowledgments}

\subsection*{Declaration of competing interest}

The authors have no competing interests to declare that are relevant to the content of this article.

\subsection*{Data Availability}
Given the substantial workload of this work, we will continuously update the code repository and available shared datasets at https://github.com/wrlengo/QUBO-VAE. We welcome requests for specific materials to further enrich this open-source project.

\subsection*{Ethical approval and consent to participate}
Not applicable.
\subsection*{Consent for publication}
Not applicable.


\newpage
\appendix
\renewcommand{\thefigure}{A\arabic{figure}}
\setcounter{figure}{0}

\renewcommand{\thetable}{A\arabic{table}}
\setcounter{table}{0}

\section*{appendix}

\subsection*{A1. General Information}

Except noted, all feature extraction, training, and execution were performed on a Mac Pro with a built-in Apple M4 Pro chip. The multi-dimensional interaction model and the conditional Boltzmann model were executed on an NVIDIA 3090 GPU. Large-scale feature extraction was accelerated using an NVIDIA 3090 GPU, and the discrete latent space training was performed using two NVIDIA H100 GPUs in tandem. Portions of the work were assisted by vibe coding tools such as Trae and Codex with manual verification. Some mathematical derivations were supported by the Lean tool. Molecular visualization was performed using Chemdraw.

\subsection*{A2. Details of Holistic pKa Regression Models}

\paragraph{Detailed Features.}
\needspace{8\baselineskip}
We generated a sufficient number of features, as shown in Table A1.

\begin{table}[htbp]
\centering
\small
\renewcommand{\arraystretch}{0.75}
\begin{threeparttable}
\caption{List of 25 Physicochemical Descriptors (fp\_3751 $\sim$ fp\_3775)}
\label{tab:25_physicochemical_descriptors}
\begin{tabular}{c l l l}
\toprule
No. & Descriptor Name & Definition & Feature Index \\
\midrule
1  & MolLogP               & Octanol-water partition coefficient          & fp\_3751 \\
2  & MolWt                 & Molecular weight                             & fp\_3752 \\
3  & NumHDonors            & Number of hydrogen bond donors               & fp\_3753 \\
4  & NumHAcceptors         & Number of hydrogen bond acceptors            & fp\_3754 \\
5  & TPSA                  & Topological polar surface area               & fp\_3755 \\
6  & HeavyAtomCount        & Number of heavy atoms                        & fp\_3756 \\
7  & NumAtoms              & Total number of atoms                        & fp\_3757 \\
8  & NumBonds              & Total number of bonds                        & fp\_3758 \\
9  & NumRotatableBonds     & Number of rotatable bonds                    & fp\_3759 \\
10 & RingCount             & Number of rings                              & fp\_3760 \\
11 & FractionCSP3          & Fraction of $sp^3$-hybridized carbon atoms    & fp\_3761 \\
12 & NumAromaticRings      & Number of aromatic rings                     & fp\_3762 \\
13 & NumAliphaticRings     & Number of aliphatic rings                    & fp\_3763 \\
14 & NumSaturatedRings     & Number of saturated rings                    & fp\_3764 \\
15 & NumHeteroatoms        & Number of heteroatoms                        & fp\_3765 \\
16 & MolMR                 & Molar refractivity                           & fp\_3766 \\
17 & BalabanJ              & Balaban J index                              & fp\_3767 \\
18 & Ipc                   & Information content index (Ipc)              & fp\_3768 \\
19 & Kappa1                & Kappa shape index 1                          & fp\_3769 \\
20 & Kappa2                & Kappa shape index 2                          & fp\_3770 \\
21 & Kappa3                & Kappa shape index 3                          & fp\_3771 \\
22 & MaxPartialCharge      & Maximum partial charge                       & fp\_3772 \\
23 & MinPartialCharge      & Minimum partial charge                       & fp\_3773 \\
24 & MaxAbsPartialCharge   & Maximum absolute partial charge              & fp\_3774 \\
25 & MinAbsPartialCharge   & Minimum absolute partial charge              & fp\_3775 \\
\bottomrule
\end{tabular}
\vspace{0.1em}
\flushleft
\footnotesize
\textbf{Feature Summary:} 3,776 fingerprint bits + 208 RDKit descriptors + 2 categorical features (solvent, measurement) = \textbf{Total 3,986 features}.
\end{threeparttable}
\end{table}
\FloatBarrier

\paragraph{Feature Selection.}

The initial attempt was based on our previously reported well-performing SPOC XGBoost model, with feature screening results shown in TABLE A2. On this basis, we exhaustively examined models and feature selection, with results presented in FIG A1.

\begin{table}[H]
\centering
\small
\renewcommand{\arraystretch}{0.75}
\begin{threeparttable}
\caption{Initial Feature Selection Attempts}
\label{tab:initial_feature_selection}
\begin{tabular}{l l c c c c c c}
\toprule
Model Type & Selection Method & k Value & Test MAE & Test $R^2$ & Train MAE & Train $R^2$ & Overfitting Index \\
\midrule
XGBoost   & kbest\_mutual & 1200 & 0.8196 & 0.9250 & 0.5184 & 0.9740 & 0.301 \\
XGBoost   & kbest\_mutual & 1500 & 0.9882 & 0.9013 & 0.7594 & 0.9447 & 0.228 \\
XGBoost   & RFE           & 1500 & 1.0772 & 0.8858 & 0.8843 & 0.9237 & 0.192 \\
XGBoost   & kbest         & 1500 & 1.0772 & 0.8858 & 0.8843 & 0.9237 & 0.192 \\
XGBoost   & kbest\_mutual & 2000 & 1.1897 & 0.8631 & 1.0101 & 0.8999 & 0.179 \\
LightGBM  & kbest\_mutual & 1500 & 1.3801 & 0.8141 & 1.2660 & 0.8346 & 0.114 \\
XGBoost   & kbest\_mutual & 2500 & 1.6848 & 0.7289 & 1.5512 & 0.7553 & 0.133 \\
CatBoost  & kbest\_mutual & 1500 & 1.8556 & 0.6728 & 1.7551 & 0.6817 & 0.100 \\
\bottomrule
\end{tabular}
\vspace{0.1em}
\flushleft
\footnotesize
\textbf{Notes:} The "Overfitting Index" is calculated as (Test MAE - Train MAE)/Test MAE, reflecting the generalization gap between training and test sets. No separate CSV files were saved for the training/test splits during these trials.
\end{threeparttable}
\end{table}

\begin{figure}[H]
\centering
\includegraphics[width=0.85\textwidth]{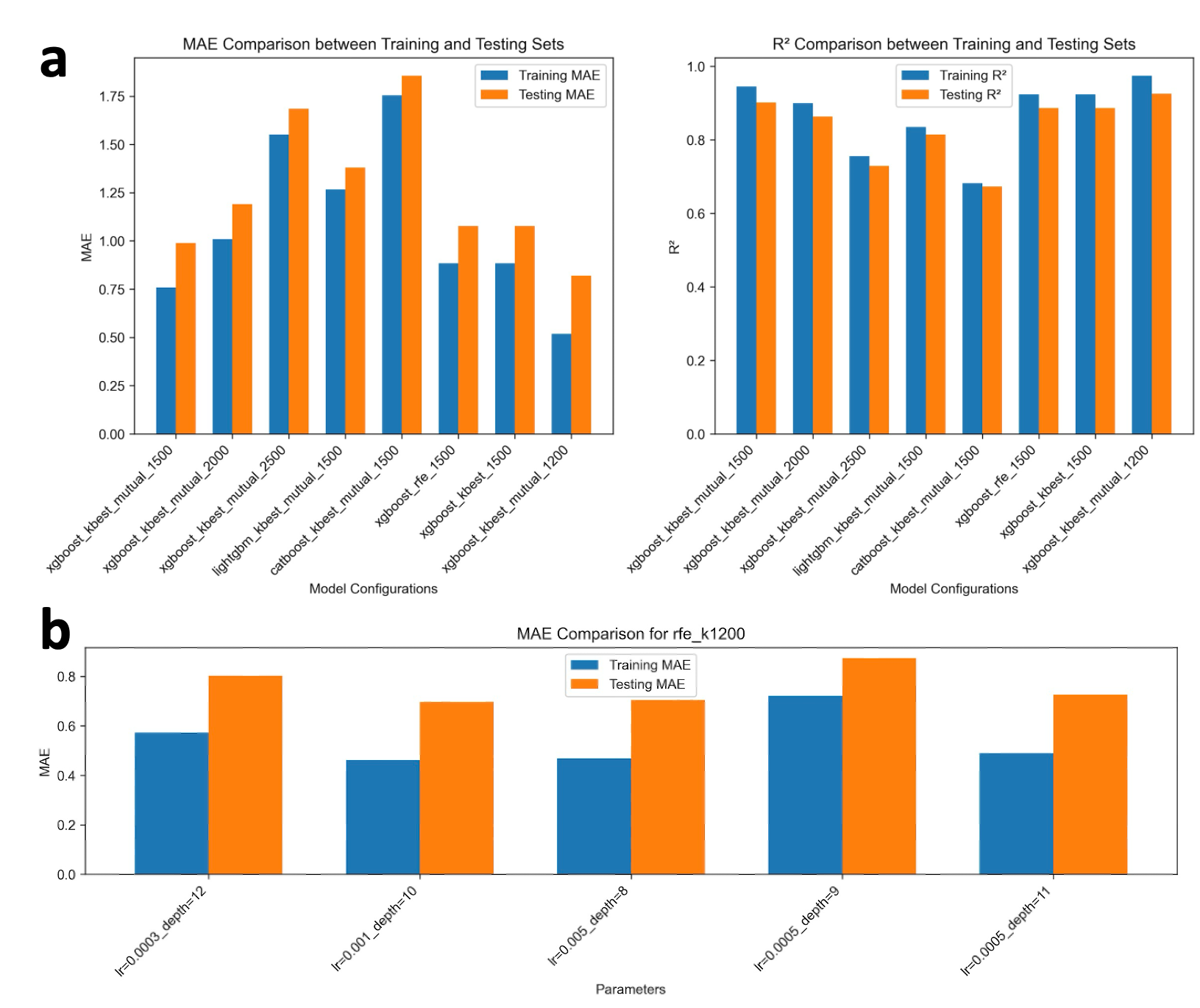}  
\caption{Selected Regression Train Results}
\label{fig2}
\end{figure}

\paragraph{Sub-optimized Results}
Since the features selected by K-best-1200 and rfe-1200 are completely identical, only the two suboptimal models, rfe-1200 and k-best-mutant-1200, combined with model optimization are shown here.
\begin{figure}[H]
\centering
\includegraphics[width=0.85\textwidth]{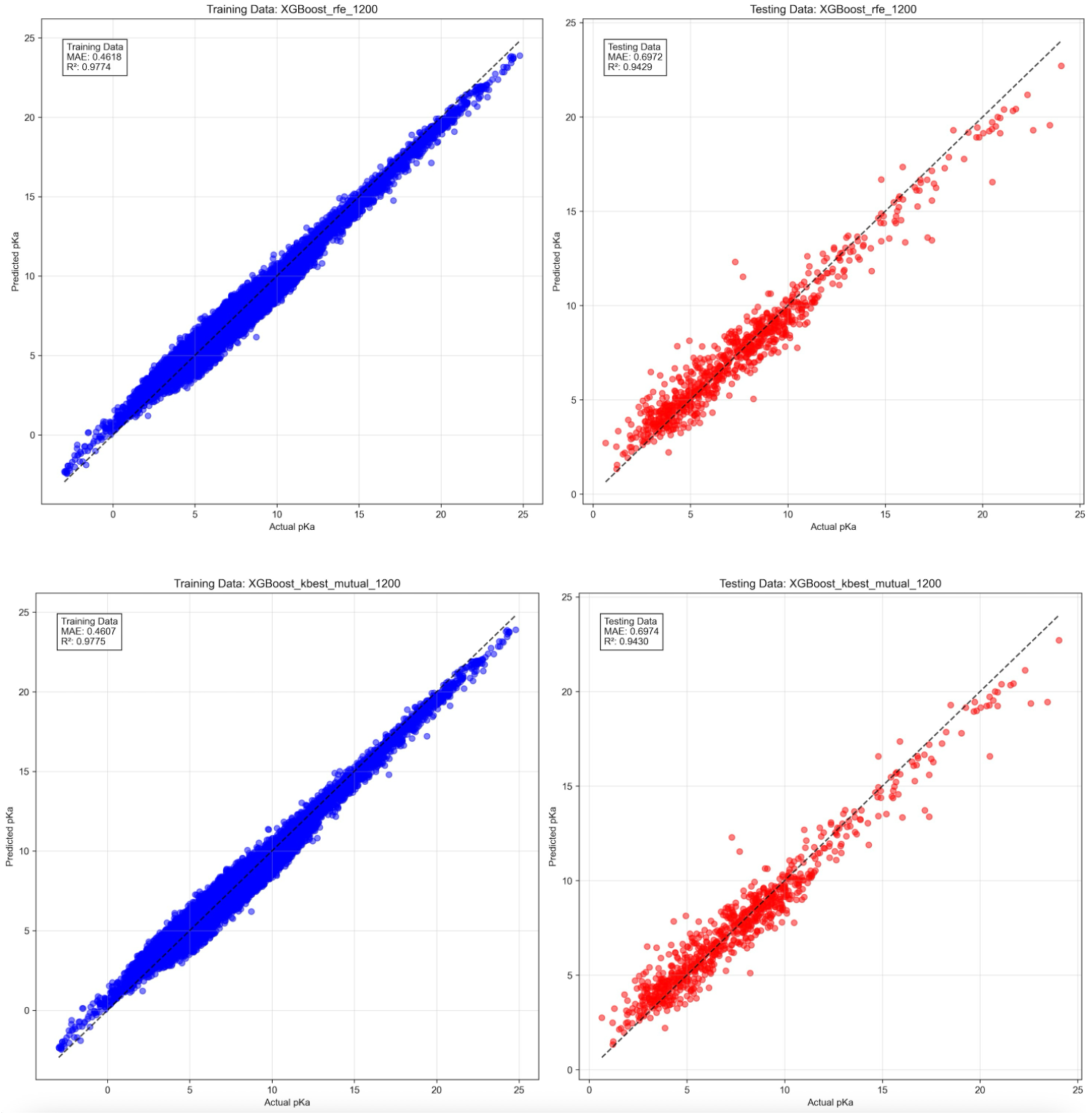}  
\caption{Sub-optimized Results}
\label{fig2}
\end{figure}

\paragraph{Final Model Details.}
From 3,776 fingerprint features, 1,200 were selected (selecting 30.1\%). The 208 RDKit complementary descriptors and the two solvent and measurement features were all retained without participating in feature screening. In total, 1,410 features were used in the model. Parameters: objective = reg:squarederror, learning\_rate = 0.0005, max\_depth = 11, min\_child\_weight = 2, gamma = 0.03, subsample = 0.85, colsample\_bytree = 0.85, reg\_alpha (L1) = 0.03, reg\_lambda (L2) = 0.2, n\_estimators = 6000, random\_state = 42.

\subsection*{A3. Aqueous Specialized Model}
\setcounter{paragraph}{0}

\paragraph{Excluded Model.}
One candidate model exhibited clear overfitting is shown in FIG.A3.

\begin{figure}[H]
\centering
\includegraphics[width=0.85\textwidth]{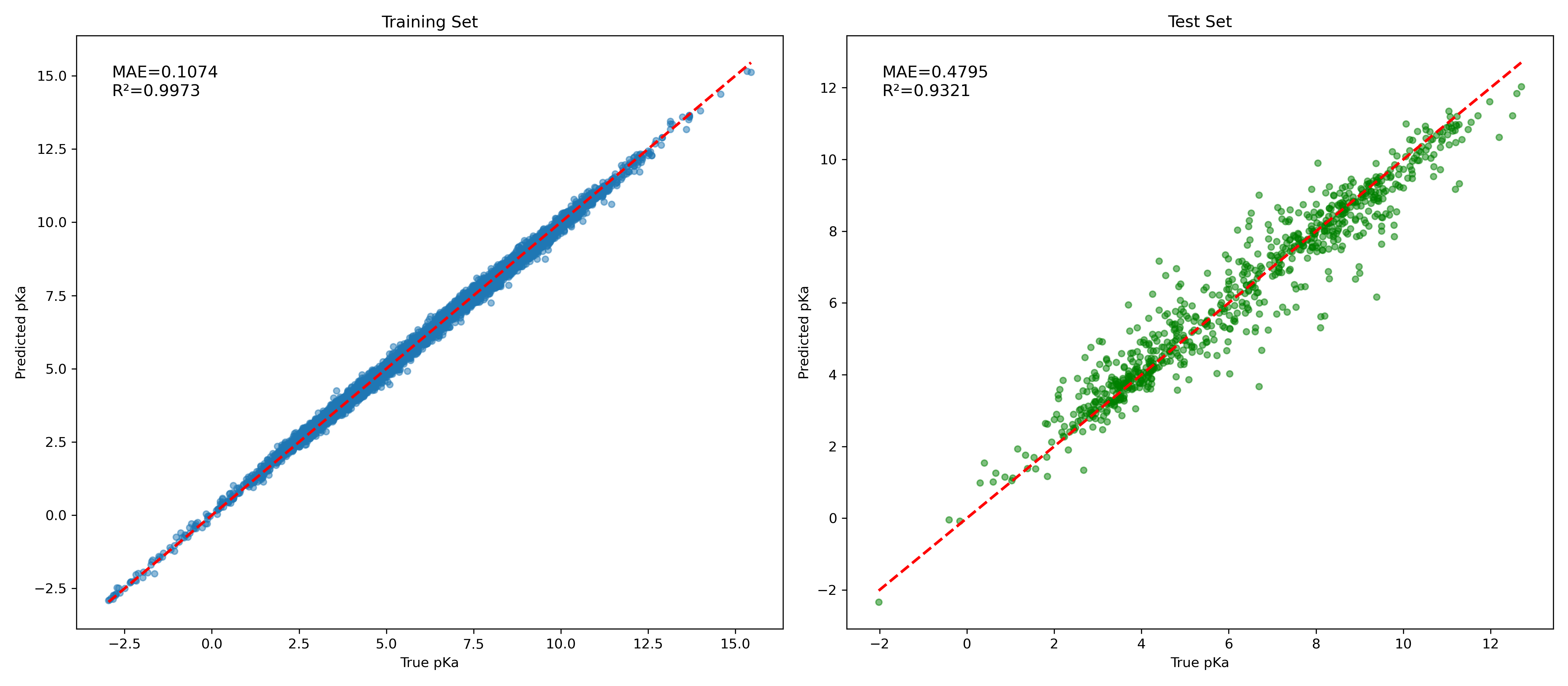}  
\caption{Overfitting Result}
\label{figA21}
\end{figure}

\paragraph{Feature Composition and Statistics}
The overall feature set consists of 3,984 features, including 208 RDKit descriptors and 3,776 fingerprint bits. Table \ref{tab:feature_overview} summarizes the feature counts and proportions.

\begin{table}[H]
\centering
\small
\renewcommand{\arraystretch}{0.75}
\begin{threeparttable}
\caption{Overview of Features}
\label{tab:feature_overview}
\begin{tabular}{lcc}
\toprule
Feature Type & Count & Proportion \\
\midrule
\textbf{Total Descriptors} & \textbf{208} & \textbf{5.22\%} \\
Total Fingerprints & 3,776 & 94.78\% \\
Total Features & 3,984 & 100\% \\
\bottomrule
\end{tabular}
\end{threeparttable}
\end{table}

\footnotesize
\noindent \textbf{Descriptor Category Distribution}:
\begin{itemize}
\setlength{\itemsep}{0pt}       
  \setlength{\parsep}{0pt}       
  \setlength{\parskip}{0pt} 
    \item Physicochemical properties: 61 (29.3\%)
    \item Count descriptors: 19 (9.1\%)
    \item EState descriptors: 4 (1.9\%)
    \item QED: 1 (0.5\%)
    \item Others: 122 (58.7\%)
\end{itemize}
\normalsize 

\paragraph{Feature Selection Comparison of Three Models}

Table \ref{tab:model_feature_comparison} summarizes the selected feature statistics for the two aqueous models and the holistic optimized model.

\begin{table}[H]
\centering
\small
\renewcommand{\arraystretch}{0.75}
\begin{threeparttable}
\caption{Feature Selection Comparison Among Models}
\label{tab:model_feature_comparison}
\begin{tabular}{lccc}
\toprule
Metric & WaterBest & WaterOptimized & XGBoost\_optimized \\
\midrule
Selected Features & 1,500 (37.7\%) & 1,800 (45.2\%) & 1,200 (30.1\%) \\
\textbf{Selected Descriptors} & \textbf{186} & \textbf{194} & \textbf{208} \\
\textbf{Descriptor Retention Rate} & \textbf{89.4\%} & \textbf{93.3\%} & \textbf{100.0\%} \\
Descriptor Ratio in Selected Features & 12.4\% & 10.8\% & 17.3\% \\
\bottomrule
\end{tabular}
\end{threeparttable}
\end{table}

\footnotesize
\noindent \textbf{Key Differences}:
\begin{itemize}
\setlength{\itemsep}{0pt}       
  \setlength{\parsep}{0pt}       
  \setlength{\parskip}{0pt} 
    \item The general model retains all 208 descriptors (100\%), while aqueous models retain 89.4\%--93.3\%.
    \item The general model shows the highest descriptor ratio (17.3\%), indicating a stronger dependence on descriptor features.
\end{itemize}
\normalsize

\paragraph{Overlapping Analysis Between WaterBest and WaterOptimized}
The overlap of selected features between the two aqueous models is relatively low, while the overlap of descriptors is high.

\footnotesize
\begin{itemize}
\setlength{\itemsep}{0pt}       
  \setlength{\parsep}{0pt}       
  \setlength{\parskip}{0pt} 
    \item Total selected features: 1,500 (Best) vs 1,800 (Optimized); overlap = 722 (48\%)
    \item Total descriptors: 208; overlapping descriptors = 180 (86.5\%)
    \item Unique descriptors in Best: 6
    \item Unique descriptors in Optimized: 14
\end{itemize}

\noindent \textbf{Unique Descriptors in WaterBest (6)}:
\begin{itemize}
\setlength{\itemsep}{0pt}       
  \setlength{\parsep}{0pt}       
  \setlength{\parskip}{0pt} 
    \item desc\_NumRadicalElectrons
    \item desc\_fr\_alkyl\_carbamate
    \item desc\_fr\_barbitur
    \item desc\_fr\_thiazole
    \item desc\_fr\_thiophene
    \item desc\_fr\_urea
\end{itemize}

\noindent \textbf{Unique Descriptors in WaterOptimized (14)}:
\begin{itemize}
\setlength{\itemsep}{0pt}       
  \setlength{\parsep}{0pt}       
  \setlength{\parskip}{0pt} 
    \item desc\_EState\_VSA11
    \item desc\_fr\_HOCCN
    \item desc\_fr\_aldehyde
    \item desc\_fr\_azide
    \item desc\_fr\_diazo
    \item desc\_fr\_imidazole
    \item desc\_fr\_imide
    \item desc\_fr\_isocyan
    \item desc\_fr\_isothiocyan
    \item desc\_fr\_lactone
\end{itemize}
\normalsize

\paragraph{Unique Characteristics of the Holistic Model}
The general model differs significantly from the aqueous models as summarized in Table \ref{tab:aq_vs_general}.

\begin{table}[H]
\centering
\small
\renewcommand{\arraystretch}{0.75}
\begin{threeparttable}
\caption{Comparison Between Aqueous Models and Holistic Model}
\label{tab:aq_vs_general}
\begin{tabular}{lcc}
\toprule
Property & Aqueous Models & Holistic Model \\
\midrule
Training Samples & 7,006 & 16,126 \\
Additional Features & Measurement & Solvent + Measurement encoding \\
Descriptor Usage & Selective (89--93\%) & Full set (100\%) \\
Feature Selection Ratio & 37.7--45.2\% & 30.1\% (stricter) \\
\bottomrule
\end{tabular}
\end{threeparttable}
\end{table}

\paragraph{Final Model Details.}
Test MAE decreased progressively through WaterOptimized model optimization: initial model (pre-CV) MAE = 0.5257, post-cross-validation MAE = 0.5028, final optimized MAE = 0.4898.
\begin{table}[H]
\centering
\small
\renewcommand{\arraystretch}{0.75}
\begin{threeparttable}
\caption{Final XGBoost Hyperparameters for Aqueous Models}
\label{tab:water_params}
\begin{tabular}{lcc}
\toprule
Parameter &  WaterBest &  WaterOptimized \\
\midrule
Model Type                & XGBoost Booster & XGBoost Booster \\
n\_estimators             & 5882            & 3990 \\
learning\_rate            & 0.001           & 0.006 \\
max\_depth                & 10              & 6 \\
min\_child\_weight         & 2               & 12 \\
gamma                     & 0.1             & 0.24 \\
subsample                 & 0.8             & default \\
colsample\_bytree         & 0.8             & default \\
reg\_alpha                & 0.1             & 0.28 \\
reg\_lambda               & 1.0             & 2.8 \\
\bottomrule
\end{tabular}
\end{threeparttable}
\end{table}
\normalsize

\paragraph{Feature Importance Analysis.}
\mbox{}\par
\begin{table}[H]
\centering
\footnotesize
\renewcommand{\arraystretch}{0.7}
\begin{threeparttable}
\caption{Top 10 Descriptor Importance for Two Aqueous Models}
\label{tab:feature_importance_both}
\begin{tabular}{l c l c}
\toprule
\textbf{WaterBest}                & \textbf{Imp.} & \textbf{WaterOptimized}          & \textbf{Imp.} \\
\midrule
desc\_BCUT2D\_MWHI                & 2.47\%  & desc\_BCUT2D\_MWHI              & 2.41\% \\
desc\_BCUT2D\_MRLOW               & 2.46\%  & desc\_BCUT2D\_MRLOW             & 2.26\% \\
desc\_MinPartialCharge            & 2.45\%  & desc\_MinPartialCharge          & 2.25\% \\
desc\_BCUT2D\_LOGPLOW             & 1.96\%  & desc\_VSA\_EState2              & 2.02\% \\
desc\_SlogP\_VSA2                 & 1.93\%  & desc\_VSA\_EState5              & 1.90\% \\
desc\_VSA\_EState3                & 1.93\%  & desc\_SlogP\_VSA2               & 1.89\% \\
desc\_VSA\_EState5                & 1.80\%  & desc\_MaxPartialCharge          & 1.88\% \\
desc\_Kappa3                      & 1.78\%  & desc\_MinEStateIndex            & 1.87\% \\
desc\_BCUT2D\_MRHI                & 1.78\%  & desc\_BCUT2D\_MRHI              & 1.85\% \\
desc\_MaxPartialCharge            & 1.73\%  & desc\_BCUT2D\_LOGPLOW           & 1.83\% \\
\midrule
\textbf{Top10 Cumulative}         & \textbf{20.30\%} & \textbf{Top10 Cumulative}      & \textbf{20.16\%} \\
\bottomrule
\end{tabular}
\end{threeparttable}
\end{table}
\footnotesize
\noindent \textbf{Model Comparison}:
\begin{itemize}
\setlength{\itemsep}{0pt}       
  \setlength{\parsep}{0pt}       
  \setlength{\parskip}{0pt} 
    \item Total selected features: 1500 (WaterBest) vs. 1800 (WaterOptimized).
    \item Shared important descriptors in Top 10: 8.
    \item Common key features: desc\_BCUT2D\_MWHI, desc\_BCUT2D\_MRLOW, desc\_MinPartialCharge, desc\_SlogP\_VSA2, desc\_VSA\_EState5, desc\_BCUT2D\_MRHI, desc\_MaxPartialCharge, desc\_BCUT2D\_LOGPLOW.
\end{itemize}
\normalsize

\paragraph{Positive and Negative Top 10 Feature Analysis for Two XGBoost Models.}

Feature direction was determined by correlation between features and predicted pKa, while XGBoost Gain was used to quantify feature importance.

\subparagraph{Model 1: WaterBest}
\noindent \textbf{Top 10 Features by Gain}
\begin{table}[H]
\centering
\footnotesize
\renewcommand{\arraystretch}{0.7}
\begin{threeparttable}
\caption{Top 10 Features by Gain (WaterBest)}
\label{tab:gain1_best}
\begin{tabular}{c l c l}
\toprule
Rank & Feature & Gain (\%) & Description \\
\midrule
1 & desc\_fr\_COO2 & 36.79 & Carboxylate/carboxylic acid count \\
2 & desc\_fr\_ArN & 7.29 & Aromatic nitrogen count \\
3 & desc\_NHOHCount & 4.54 & N-O-H group count \\
4 & desc\_fr\_Ar\_N & 4.07 & Aromatic nitrogen count \\
5 & desc\_FractionCSP3 & 2.48 & Fraction of sp3-hybridized C \\
6 & desc\_VSA\_EState2 & 1.91 & EState VSA type 2 \\
7 & desc\_NumAromaticHeterocycles & 1.90 & Aromatic heterocycle count \\
8 & desc\_NumAromaticRings & 1.25 & Aromatic ring count \\
9 & desc\_fr\_aniline & 1.13 & Aniline group count \\
10 & desc\_MaxAbsPartialCharge & 1.12 & Max absolute partial charge \\
\midrule
\multicolumn{3}{l}{\textbf{Top 10 cumulative contribution}} & \textbf{62.48\%} \\
\bottomrule
\end{tabular}
\end{threeparttable}
\end{table}

\noindent \textbf{Top 10 Positive Correlation (Higher Value → Higher pKa)}
\begin{table}[H]
\centering
\footnotesize
\renewcommand{\arraystretch}{0.7}
\begin{threeparttable}
\caption{Top 10 Positive Correlation Features (WaterBest)}
\label{tab:pos1_best}
\begin{tabular}{c l c c l}
\toprule
Rank & Feature & Correlation & Importance (\%) & Description \\
\midrule
1 & desc\_FractionCSP3 & +0.3943 & 2.48 & sp3 C fraction \\
2 & desc\_VSA\_EState8 & +0.3486 & 0.05 & VSA EState8 \\
3 & desc\_MinEStateIndex & +0.3320 & 0.18 & Min EState index \\
4 & desc\_SMR\_VSA3 & +0.3173 & 0.36 & SMR VSA3 \\
5 & desc\_MinPartialCharge & +0.3116 & 0.39 & Min partial charge \\
6 & desc\_SMR\_VSA4 & +0.2913 & 0.95 & SMR VSA4 \\
7 & desc\_HallKierAlpha & +0.2845 & 0.31 & Hall-Kier alpha \\
8 & desc\_fr\_NH1 & +0.2641 & 0.17 & Primary amine \\
9 & desc\_SMR\_VSA5 & +0.2536 & 0.17 & SMR VSA5 \\
10 & desc\_VSA\_EState5 & +0.2421 & 0.20 & VSA EState5 \\
\bottomrule
\end{tabular}
\end{threeparttable}
\end{table}

\noindent \textbf{Top 10 Negative Correlation (Higher Value → Lower pKa)}
\begin{table}[H]
\centering
\footnotesize
\renewcommand{\arraystretch}{0.7}
\begin{threeparttable}
\caption{Top 10 Negative Correlation Features (WaterBest)}
\label{tab:neg1_best}
\begin{tabular}{c l c c l}
\toprule
Rank & Feature & Correlation & Importance (\%) & Description \\
\midrule
1 & desc\_MinAbsPartialCharge & $-$0.5790 & 0.53 & Min abs partial charge \\
2 & desc\_MaxPartialCharge & $-$0.5656 & 0.41 & Max partial charge \\
3 & desc\_fr\_COO2 & $-$0.5485 & 36.79 & Carboxylate/carboxylic acid \\
4 & desc\_fr\_COO & $-$0.5473 & 0.31 & Carboxylic acid group \\
5 & desc\_fr\_C\_O & $-$0.5045 & 0.15 & Carbonyl group count \\
6 & desc\_fr\_Al\_COO & $-$0.4589 & 0.05 & Aliphatic carboxylate \\
7 & desc\_VSA\_EState2 & $-$0.4397 & 1.91 & VSA EState2 \\
8 & desc\_SMR\_VSA10 & $-$0.4253 & 0.12 & SMR VSA10 \\
9 & desc\_SMR\_VSA1 & $-$0.3791 & 0.30 & SMR VSA1 \\
10 & desc\_MaxAbsEStateIndex & $-$0.3672 & 0.00 & Max abs EState index \\
\bottomrule
\end{tabular}
\end{threeparttable}
\end{table}

\subparagraph{Model 2: WaterOptimized}
\noindent \textbf{Top 10 Features by Gain}
\begin{table}[H]
\centering
\footnotesize
\renewcommand{\arraystretch}{0.7}
\begin{threeparttable}
\caption{Top 10 Features by Gain (WaterOptimized)}
\label{tab:gain1_opt}
\begin{tabular}{c l c l}
\toprule
Rank & Feature & Gain (\%) & Description \\
\midrule
1 & desc\_fr\_COO2 & 35.26 & Carboxylate/carboxylic acid count \\
2 & desc\_fr\_ArN & 5.70 & Aromatic nitrogen count \\
3 & desc\_NHOHCount & 4.83 & N-O-H group count \\
4 & desc\_fr\_Ar\_N & 3.64 & Aromatic nitrogen count \\
5 & desc\_FractionCSP3 & 3.27 & sp3 C fraction \\
6 & desc\_fr\_phenol & 2.19 & Phenol group count \\
7 & desc\_NumAromaticHeterocycles & 1.71 & Aromatic heterocycle count \\
8 & desc\_NumHDonors & 1.50 & H-bond donor count \\
9 & desc\_fr\_allylic\_oxid & 1.46 & Allylic oxidation site count \\
10 & desc\_fr\_COO & 1.19 & Carboxylic acid group \\
\midrule
\multicolumn{3}{l}{\textbf{Top 10 cumulative contribution}} & \textbf{59.75\%} \\
\bottomrule
\end{tabular}
\end{threeparttable}
\end{table}

\noindent \textbf{Top 10 Positive Correlation (Higher Value → Higher pKa)}
\begin{table}[H]
\centering
\footnotesize
\renewcommand{\arraystretch}{0.7}
\begin{threeparttable}
\caption{Top 10 Positive Correlation Features (WaterOptimized)}
\label{tab:pos1_opt}
\begin{tabular}{c l c c l}
\toprule
Rank & Feature & Correlation & Importance (\%) & Description \\
\midrule
1 & desc\_FractionCSP3 & +0.3885 & 3.27 & sp3 C fraction \\
2 & desc\_VSA\_EState8 & +0.3418 & 0.05 & VSA EState8 \\
3 & desc\_MinEStateIndex & +0.3270 & 0.19 & Min EState index \\
4 & desc\_SMR\_VSA3 & +0.3137 & 0.32 & SMR VSA3 \\
5 & desc\_MinPartialCharge & +0.2959 & 0.53 & Min partial charge \\
6 & desc\_SMR\_VSA4 & +0.2901 & 1.19 & SMR VSA4 \\
7 & desc\_HallKierAlpha & +0.2740 & 0.25 & Hall-Kier alpha \\
8 & desc\_fr\_NH1 & +0.2660 & 0.08 & Primary amine \\
9 & desc\_SMR\_VSA5 & +0.2533 & 0.09 & SMR VSA5 \\
10 & desc\_VSA\_EState5 & +0.2385 & 0.21 & VSA EState5 \\
\bottomrule
\end{tabular}
\end{threeparttable}
\end{table}

\noindent \textbf{Top 10 Negative Correlation (Higher Value → Lower pKa)}
\begin{table}[H]
\centering
\footnotesize
\renewcommand{\arraystretch}{0.7}
\begin{threeparttable}
\caption{Top 10 Negative Correlation Features (WaterOptimized)}
\label{tab:neg1_opt}
\begin{tabular}{c l c c l}
\toprule
Rank & Feature & Correlation & Importance (\%) & Description \\
\midrule
1 & desc\_MinAbsPartialCharge & $-$0.5649 & 0.77 & Min abs partial charge \\
2 & desc\_MaxPartialCharge & $-$0.5519 & 0.27 & Max partial charge \\
3 & desc\_fr\_COO2 & $-$0.5375 & 35.26 & Carboxylate/carboxylic acid \\
4 & desc\_fr\_COO & $-$0.5363 & 1.19 & Carboxylic acid group \\
5 & desc\_fr\_C\_O & $-$0.4949 & 0.23 & Carbonyl group count \\
6 & desc\_fr\_Al\_COO & $-$0.4543 & 0.05 & Aliphatic carboxylate \\
7 & desc\_VSA\_EState2 & $-$0.4280 & 1.14 & VSA EState2 \\
8 & desc\_SMR\_VSA10 & $-$0.4196 & 0.15 & SMR VSA10 \\
9 & desc\_SMR\_VSA1 & $-$0.3694 & 0.37 & SMR VSA1 \\
10 & desc\_MaxAbsEStateIndex & $-$0.3558 & 0.00 & Max abs EState index \\
\bottomrule
\end{tabular}
\end{threeparttable}
\end{table}

\footnotesize
\noindent \textbf{Key Points}:
\begin{itemize}
\setlength{\itemsep}{0pt}       
  \setlength{\parsep}{0pt}       
  \setlength{\parskip}{0pt} 
    \item The dominant feature in both models is \textbf{desc\_fr\_COO2} (carboxylate/carboxylic acid), contributing 35–37\% with strong negative correlation.
    \item Key features increasing pKa: \textbf{desc\_FractionCSP3} (most positive, $r\approx+0.39$), EState and SMR–VSA descriptors.
    \item Key features decreasing pKa: partial charge descriptors (\texttt{MinAbsPartialCharge}, \texttt{MaxPartialCharge}) and carbonyl/carboxyl groups (\texttt{fr\_COO2}, \texttt{fr\_COO}, \texttt{fr\_C\_O}).
\end{itemize}
\normalsize

\paragraph{SHAP Feature Analysis.}
Fingerprint features (Morgan, atom pairs, topological torsion) dominate the top SHAP contributions in both models. Descriptor features exhibit negligible contribution in the SHAP-based ranking, which is inconsistent with the dominant role of descriptors observed in the XGBoost gain-based feature importance analysis.
\begin{figure}[H]
\centering
\includegraphics[width=0.85\textwidth]{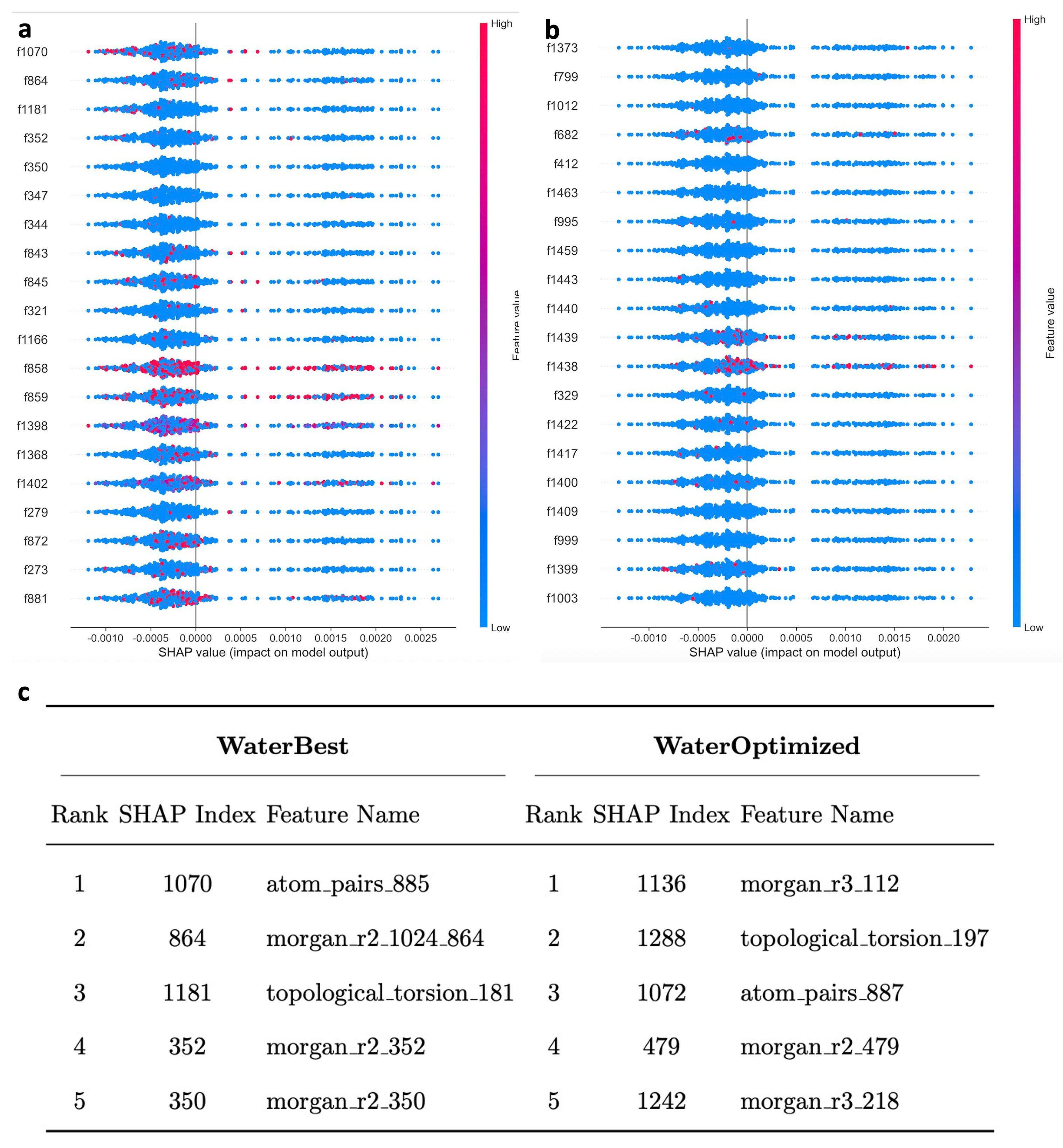}  
\caption{SHAP Analysis. (a) WaterBest, (b) WaterOptimized, (c) Comparisons}
\label{fig3}
\end{figure}

\paragraph{Results of Different Measurements.}
\mbox{}\par
\begin{figure}[H]
\centering
\includegraphics[width=0.95\textwidth]{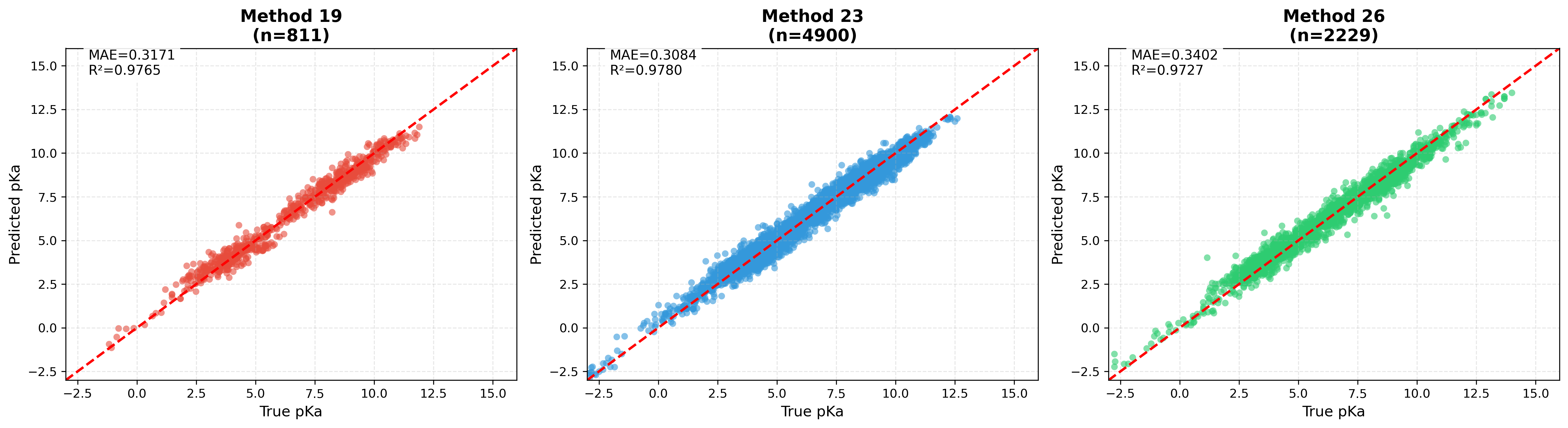}  
\caption{Results of Different Measurements}
\label{figA5}
\end{figure}

\subsection*{A4. RNN-VAE Generative Approches}
\subsubsection*{RNN Architecture and Sequence Parameters}
\begin{itemize}
\setlength{\itemsep}{0pt}       
  \setlength{\parsep}{0pt}       
  \setlength{\parskip}{0pt} 
    \item \textbf{Latent dimension}: 16/32/48 \\
    The latent vector $z$ compressed by the RNN encoder is fixed at 16/32/48 dimensions, which is the core dimension of the entire model.
    \item \textbf{Batch size}: 200 \\
    The RNN is fed with 200 molecular sequences per training step.
    \item \textbf{Total training epochs}: 300 \\
    The RNN-VAE converges after 300 complete training epochs.
    \item \textbf{Network unit}: Basic recurrent neural network (RNN) cells \\
    As explicitly stated, standard RNN cells are used instead of LSTM/GRU for simplified adaptation to molecular strings.
    \item \textbf{Input format}: Tokenized molecular strings (SMILES/SELFIES) are sequentially fed into the RNN.
    \item \textbf{Embedding dimension}: The dimension of vectors converted from molecular tokens by the RNN, designed to match the llatent space.
    \item \textbf{Sequence length}: Molecular strings are uniformly padded or truncated to fit the equal-length sequence input requirement of RNNs.
\end{itemize}

\subsubsection*{RNN-Tied Anti-Collapse and KL Control Parameters}
\begin{itemize}
\setlength{\itemsep}{0pt}       
  \setlength{\parsep}{0pt}       
  \setlength{\parskip}{0pt} 
    \item \textbf{$\beta$ annealing strategy}: \\
    The KL divergence weight $\beta$ increases linearly from 0 to 0.1 over 100 epochs \\
    $\rightarrow$ Prevents KL divergence from dropping directly to 0 and avoids posterior collapse of the latent space encoded by the RNN.
    \item \textbf{Joint loss constraint}: \\
    RNN reconstruction loss + $\beta \cdot$ KL divergence + prediction MSE are used together to update RNN weights \\
    $\rightarrow$ Ensures accurate molecular reconstruction while preserving the validity of the latent space.
\end{itemize}

\subsubsection*{General Results}
For each model configuration, the following pipeline is executed:
Fixed condition $\rightarrow$ Random sampling of 1000 $z$ vectors $\rightarrow$ Reparameterization trick $\rightarrow$ Condition concatenation $\rightarrow$ SimpleRNN decoding $\rightarrow$ Post-processing $\rightarrow$ Diversity metrics calculation $\rightarrow$ Plotting $\rightarrow$ In-situ saving.
The optimal conditional RNN-VAE is selected based on validity rate, novelty rate, scaffold diversity, and condition compliance rate.
Only the optimal model is subjected to latent space interpolation to verify local diversity and latent space smoothness, along with other evaluations. Since the performance differences between models are negligible, the results are summarized in Table \ref{tab:rnn_vae_metrics}.

\begin{table}[H]
  \centering
  \caption{Generative Performance Metrics of RNN-VAE Models}
  \label{tab:rnn_vae_metrics}
  \scriptsize          
  \renewcommand{\arraystretch}{0.7}  
    \begin{tabular}{llll}
    \toprule
    Evaluation Dimension & Metric & Value Range & Description \\
    \midrule
    \multirow{6}{*}{Random Generation}
    & Validity Rate & 0.41-0.45 & Proportion of successfully generated valid molecules \\
    & Error Rate & 0.55-0.59 & Proportion of generated invalid molecules \\
    & Average Tanimoto Similarity & 0.216-0.243 & Structural similarity between generated molecules \\
    & Unique Scaffold Count & 5 & Number of distinct Murcko scaffolds \\
    & Average Novelty & 0.338-0.343 & Difference between generated and training set \\
    & Novelty Rate & 74.4\%-75.5\% & Molecules not present in the training set \\
    \midrule
    \multirow{2}{*}{Latent Space Interpolation}
    & Interpolation Validity Rate & 1.00 & Validity rate of interpolated molecules \\
    & Average Smoothness & 0.818-0.952 & Smoothness of latent space transition \\
    \midrule
    \multirow{3}{*}{Gradient Optimization}
    & Optimization Success Rate & 0.70 & Optimized pKa within target range \\
    & Optimization Accuracy & 0.300 & Mean absolute deviation of pKa \\
    & pKa Stability & 0.072-0.080 & Std of repeated optimizations \\
    \bottomrule
  \end{tabular}
\end{table}

\begin{figure}[H]
\centering
\includegraphics[width=0.85\textwidth]{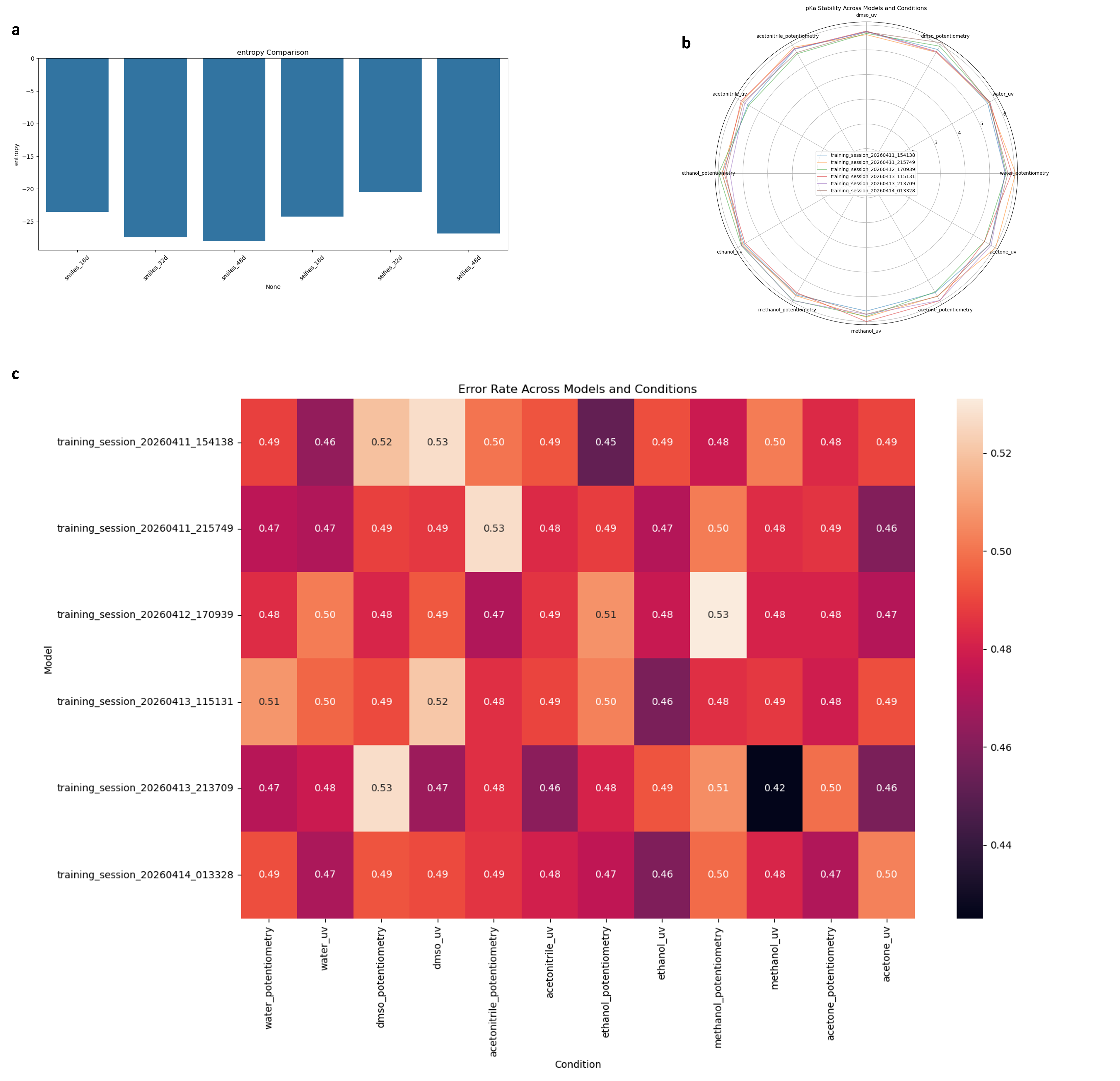}  
\caption{(a) Entropy comparison,  (b) pKa Stability, (c) Detailed error rate.}
\label{figA6}
\end{figure}

\subsubsection*{Detailed Results}
The detailed results of the six models are summarized in FIG. A6. All attempts to fit toward the extreme pKa (preset at 4) failed, and these unsuccessful results are omitted to avoid redundant presentation. 

\begin{figure}[H]
\centering
\includegraphics[width=0.95\textwidth]{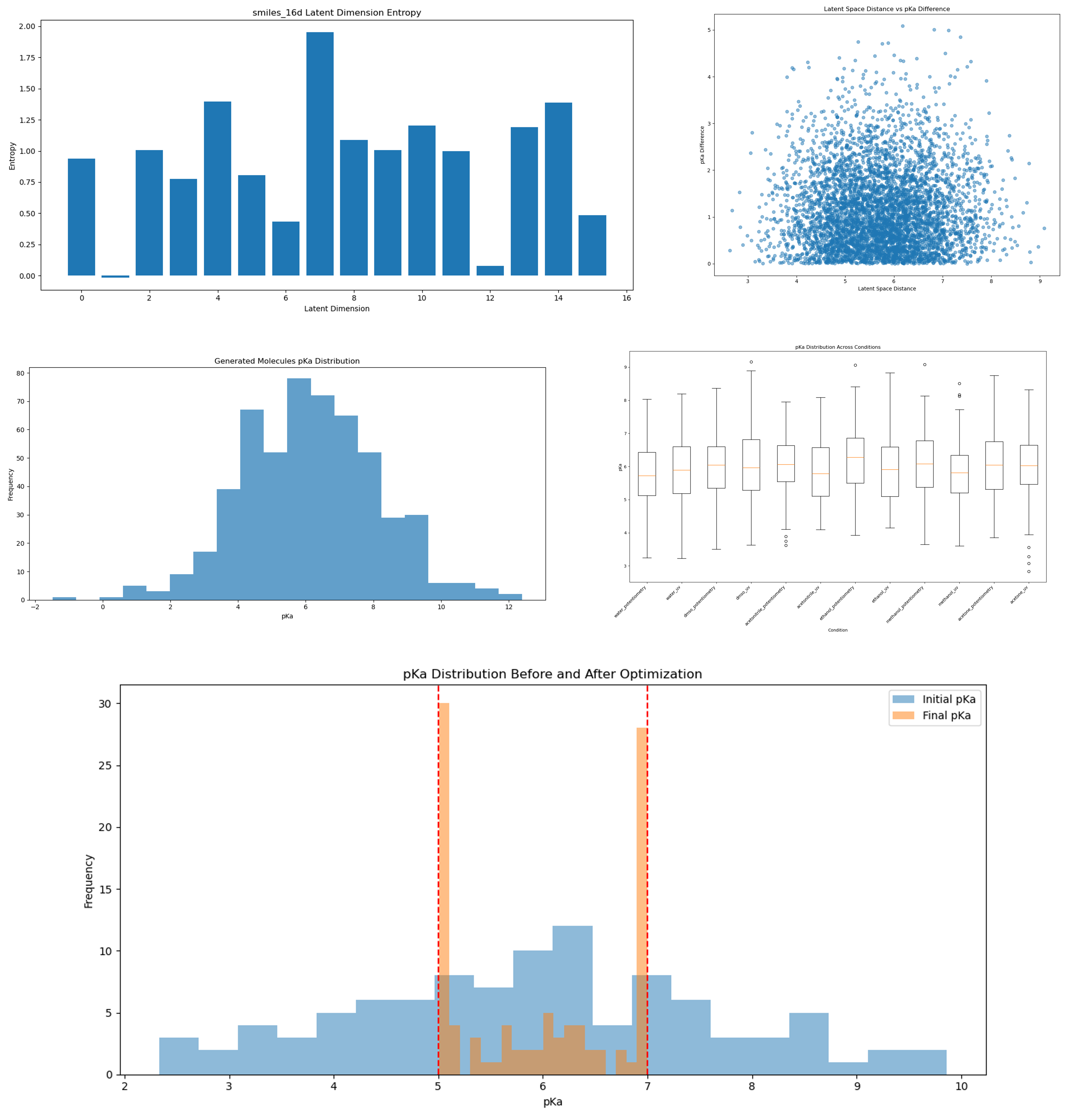}  
\caption{SMILES-16D}
\label{figA71}
\end{figure}

\begin{figure}[H]
\centering
\includegraphics[width=0.95\textwidth]{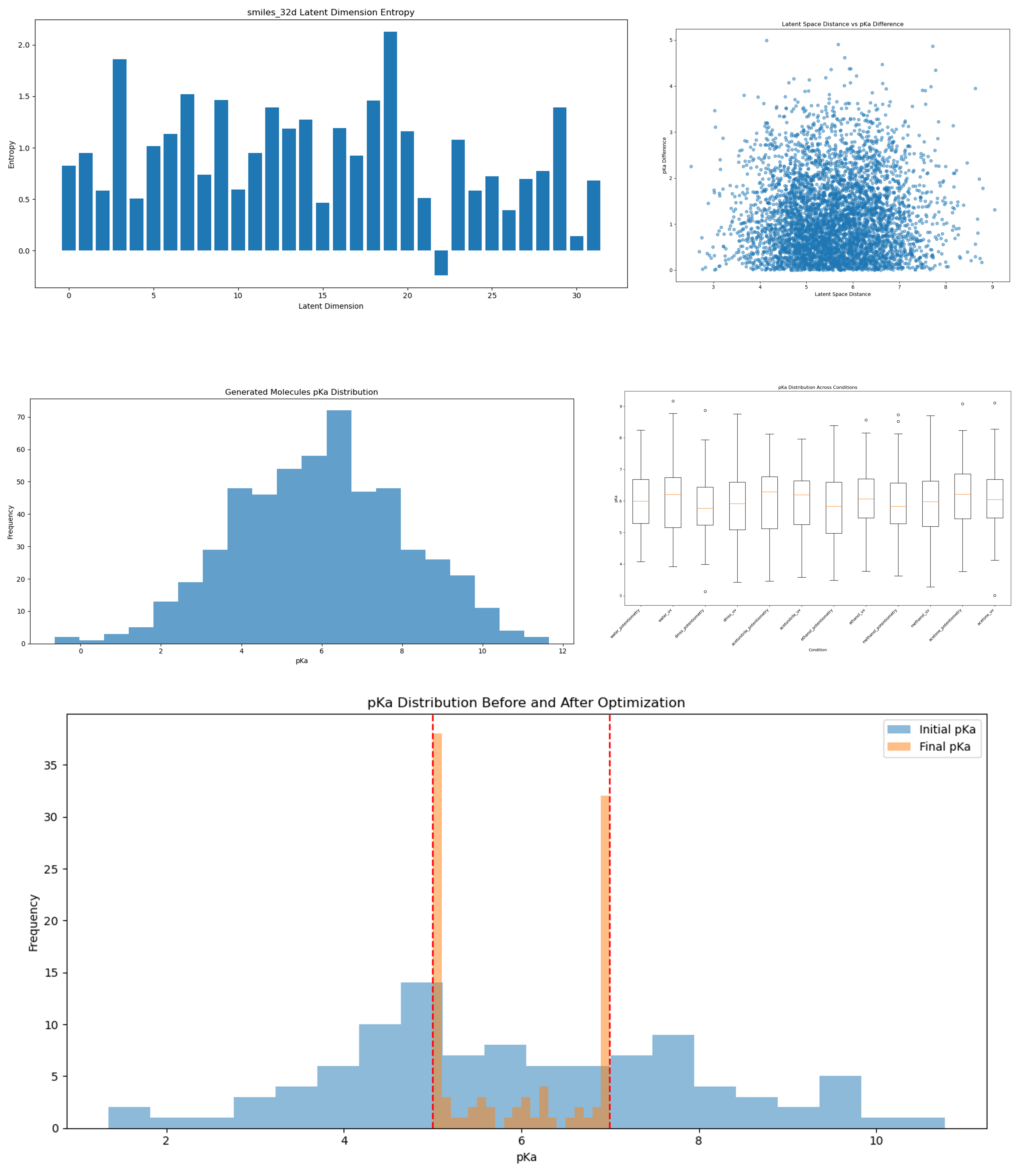}  
\caption{SMILES-32D}
\label{figA72}
\end{figure}

\begin{figure}[H]
\centering
\includegraphics[width=0.95\textwidth]{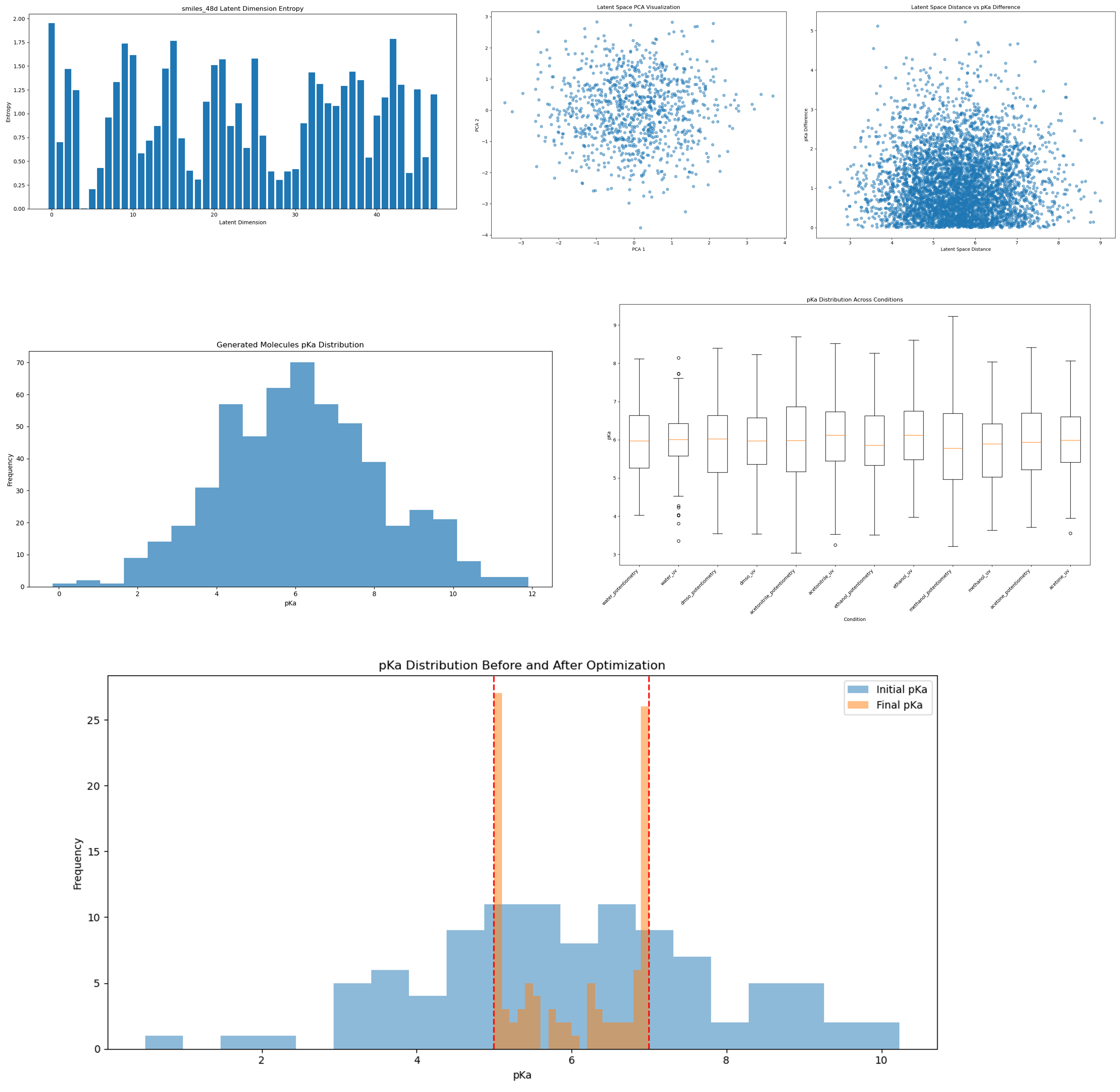}  
\caption{SMILES-48D}
\label{figA73}
\end{figure}

\begin{figure}[H]
\centering
\includegraphics[width=0.95\textwidth]{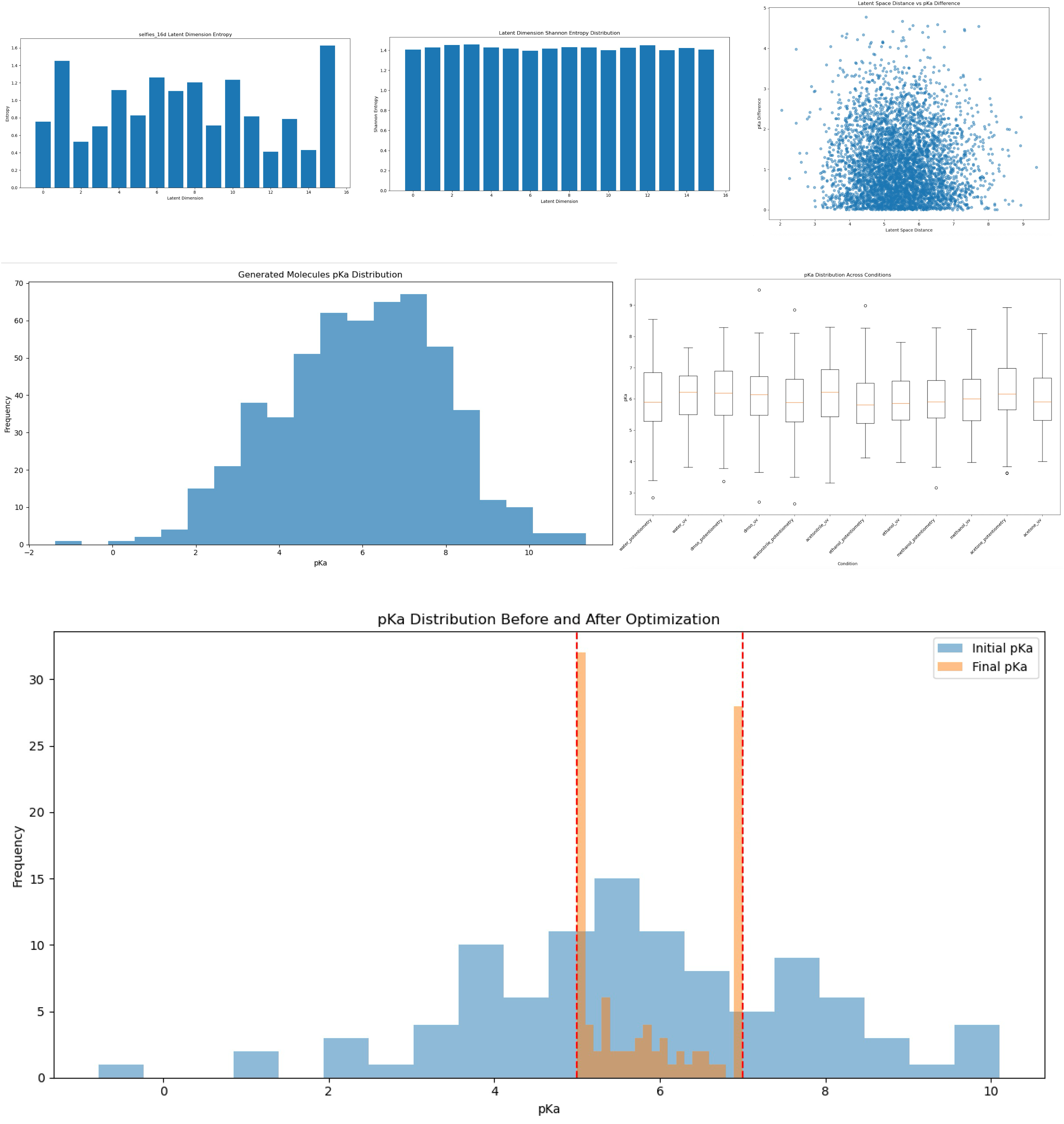}  
\caption{SELFIES-16D}
\label{figA74}
\end{figure}

\begin{figure}[H]
\centering
\includegraphics[width=0.95\textwidth]{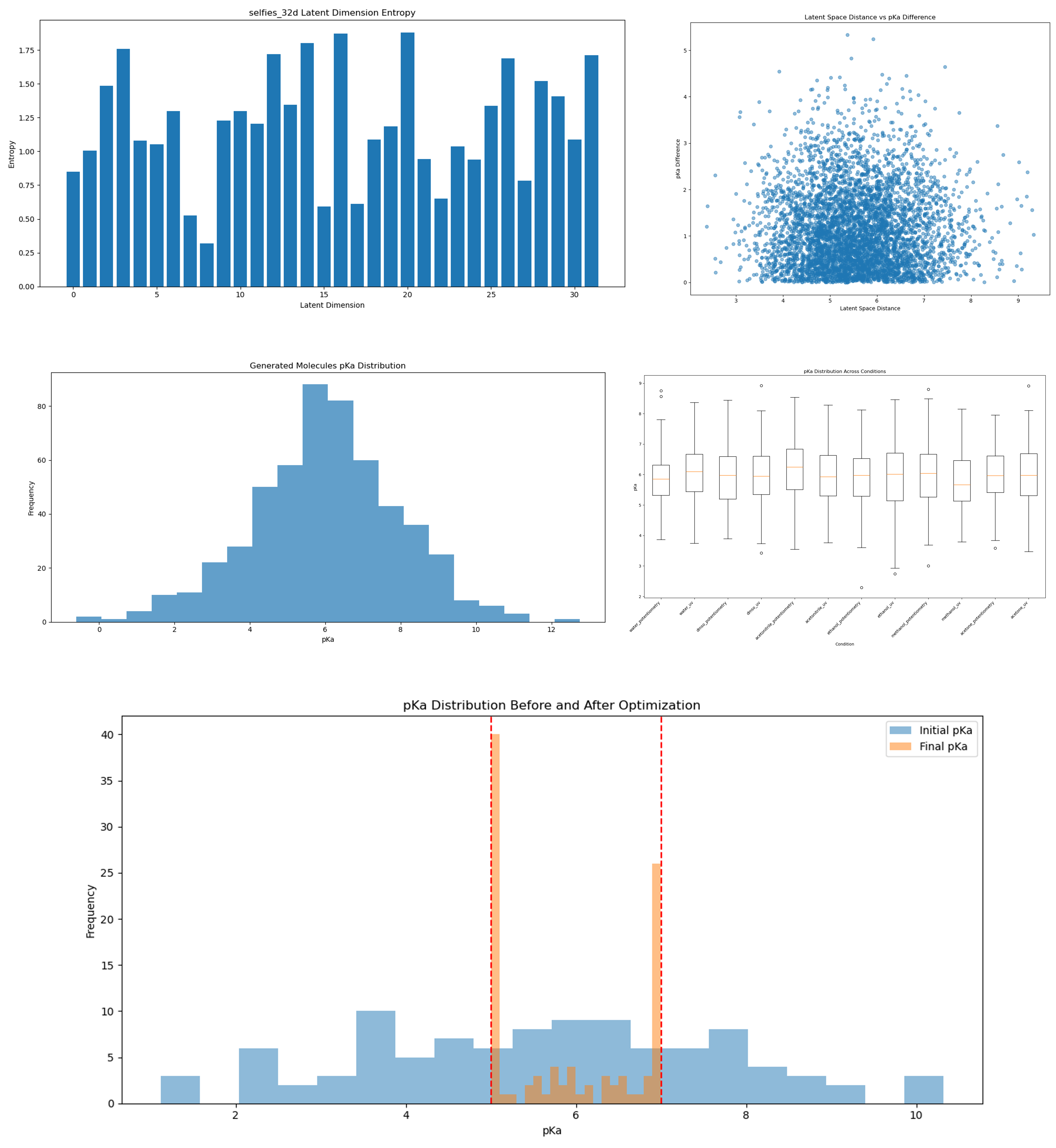}  
\caption{SELFIES-32D}
\label{figA75}

\end{figure}
\begin{figure}[H]
\centering
\includegraphics[width=0.95\textwidth]{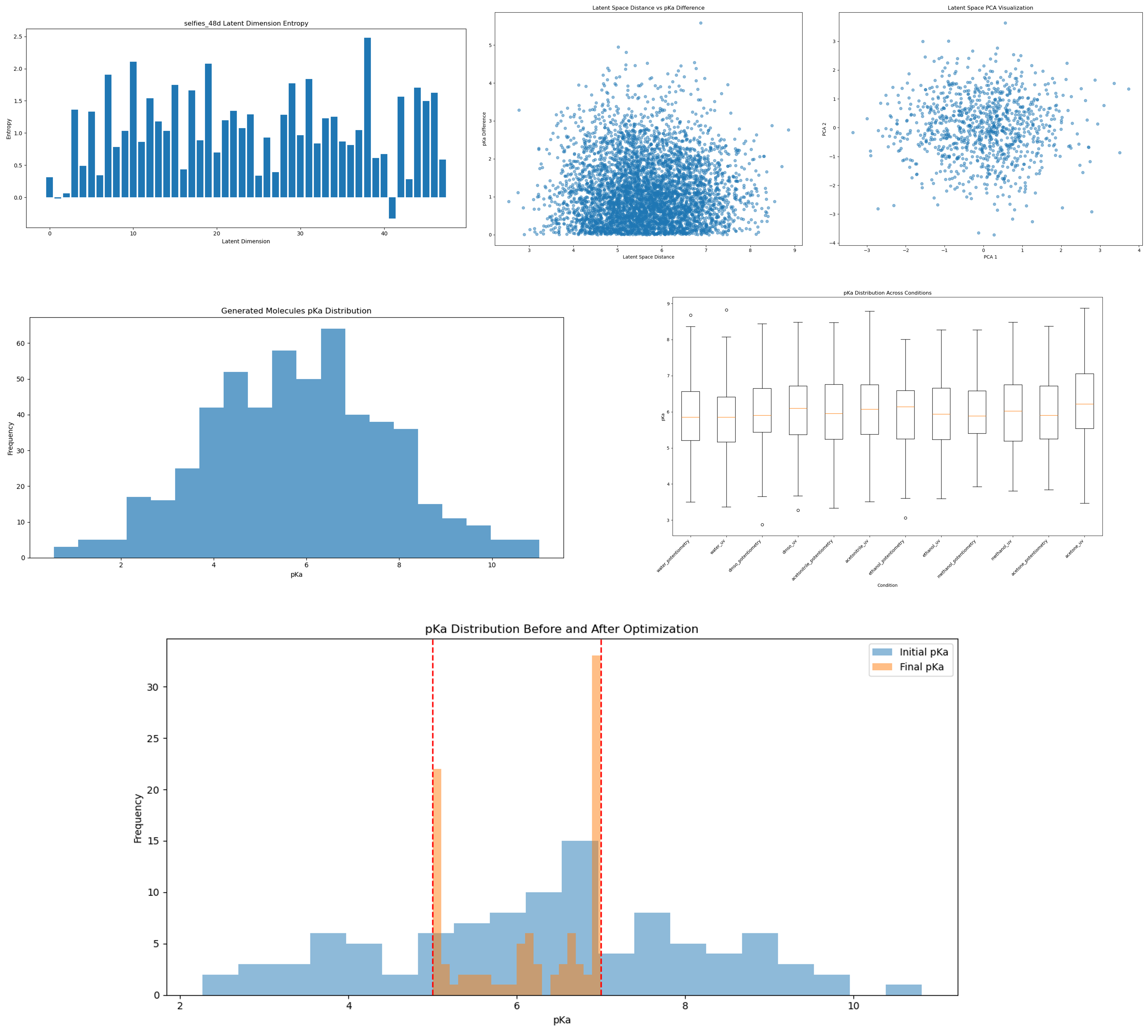}  
\caption{SELFIES-48D}
\label{figA76}
\end{figure}

\subsection*{A5. Pipeline for Binary-Latent Molecular Encoding and Decoding}
\setcounter{paragraph}{0}
This section provides a detailed, step-by-step exposition of the binary discrete latent space construction employed in Janus. In contrast to the high-level overview in the main text, we present the complete algorithmic workflow, including data preprocessing, Transformer-based encoding, Gumbel-Softmax binarization, autoregressive decoding, and the training objective with free-bit KL regularization. A pseudocode summary is provided at the end of the subsection.

\paragraph{Step 1: SELFIES Tokenization and Vocabulary Encoding.}
Given a molecule, we first generate its SELFIES string \(S\) (a robust, 100\% valid molecular grammar). The string is tokenized into a sequence \(S = (s_1, s_2, \dots, s_L)\) where \(L\) is the original length. A fixed vocabulary \(\mathcal{V}\) of all possible SELFIES tokens is precomputed. Each token \(s_i\) is mapped to a \(d_{\text{model}}\)-dimensional embedding via a learnable embedding matrix \(E_{\text{emb}} \in \mathbb{R}^{|\mathcal{V}| \times d_{\text{model}}}\). To handle variable sequence lengths in batch training, we pad every sequence to a maximum length \(L_{\text{max}}\) using a special \([\text{PAD}]\) token, producing a padded sequence \(\tilde{S}\) of length \(L_{\text{max}}\). Positional information is added via sinusoidal positional encodings \(\text{PE}(i) \in \mathbb{R}^{d_{\text{model}}}\):
\[
\mathbf{h}_i^{(0)} = E_{\text{emb}}(\tilde{s}_i) + \text{PE}(i), \quad i = 1,\dots,L_{\text{max}}.
\]

The input tensor to the encoder is \(\mathbf{H}^{(0)} = [\mathbf{h}_1^{(0)}; \dots; \mathbf{h}_{L_{\text{max}}}^{(0)}] \in \mathbb{R}^{L_{\text{max}} \times d_{\text{model}}}\).

\paragraph{Step 2: SwiGLU Transformer Encoder with Pre-LN.}

The encoder consists of \(N\) identical layers, each employing a Pre-Layer Normalization (Pre-LN) scheme and a SwiGLU-based feed-forward network (FFN). For layer \(\ell = 1\) to \(N\):
\[
\begin{aligned}
\mathbf{H}' &= \text{LayerNorm}(\mathbf{H}^{(\ell-1)}), \\
\mathbf{A} &= \text{MultiHeadAttention}(\mathbf{H}', \mathbf{H}', \mathbf{H}'), \\
\mathbf{H}_{\text{interim}} &= \mathbf{H}^{(\ell-1)} + \text{Dropout}(\mathbf{A}), \\
\mathbf{H}_{\text{ffn}} &= \text{SwiGLU}(\text{LayerNorm}(\mathbf{H}_{\text{interim}})), \\
\mathbf{H}^{(\ell)} &= \mathbf{H}_{\text{interim}} + \text{Dropout}(\mathbf{H}_{\text{ffn}}).
\end{aligned}
\]

The SwiGLU operation is defined as:
\[
\text{SwiGLU}(\mathbf{X}) = (\text{SiLU}(\mathbf{X}\mathbf{W}_1) \odot (\mathbf{X}\mathbf{W}_3)) \mathbf{W}_2,
\]

where \(\mathbf{W}_1, \mathbf{W}_3 \in \mathbb{R}^{d_{\text{model}} \times d_{\text{ff}}}\), \(\mathbf{W}_2 \in \mathbb{R}^{d_{\text{ff}} \times d_{\text{model}}}\), and \(\odot\) denotes elementwise multiplication. After \(N\) layers, we obtain \(\mathbf{H}^{(N)} \in \mathbb{R}^{L_{\text{max}} \times d_{\text{model}}}\).

\paragraph{Step 3: Attention Pooling to a Fixed-Length Vector.}

To aggregate the variable-length sequence into a single global representation, we use an attention pooling layer. A learnable query vector \(\mathbf{q} \in \mathbb{R}^{d_{\text{model}}}\) interacts with the sequence through key projections:
\[
\mathbf{K} = \mathbf{H}^{(N)} \mathbf{W}_K, \quad \mathbf{W}_K \in \mathbb{R}^{d_{\text{model}} \times d_k},
\]
\[
\alpha = \text{softmax}\left( \frac{\mathbf{q} \mathbf{K}^T}{\sqrt{d_k}} \right) \in \mathbb{R}^{L_{\text{max}}}, \quad
\mathbf{c} = \alpha \mathbf{H}^{(N)} \in \mathbb{R}^{d_{\text{model}}}.
\]

The vector \(\mathbf{c}\) serves as the continuous latent embedding of the molecule.

\paragraph{Step 4: Binary Discretization via Gumbel-Softmax.}

We project \(\mathbf{c}\) to a \(D_L\)-dimensional logits vector (with \(D_L = 128\)):
\[
\mathbf{q}_{\text{logits}} = \mathbf{c} \mathbf{W}_{\text{to\_latent}} + \mathbf{b}_{\text{to\_latent}}, \quad \mathbf{q}_{\text{logits}} \in \mathbb{R}^{D_L}.
\]

For each dimension \(j\), we define unnormalized probabilities for bit \(0\) and bit \(1\) as \(p_{j,0}=1\) and \(p_{j,1} = \exp(\mathbf{q}_{\text{logits},j})\). During training, we sample from the Gumbel-Softmax distribution with temperature \(\tau\):
\[
\pi_{j,k} = \frac{\exp\left((\log p_{j,k} + g_k)/\tau\right)}{\sum_{k'=0}^1 \exp\left((\log p_{j,k'} + g_{k'})/\tau\right)}, \quad k \in \{0,1\},
\]

where \(g_k \sim \text{Gumbel}(0,1)\) is independent Gumbel noise. In the forward pass, we use the hard (straight-through) estimator: \(x_j = \text{onehot\_argmax}(\pi_{j,\cdot})\) as the discrete output, while in the backward pass the continuous softmax probabilities are used to compute gradients. During inference, we set \(x_j = \mathbf{1}[\mathbf{q}_{\text{logits},j} > 0]\) (deterministic thresholding). The final binary latent code is \(\mathbf{x} = (x_1, \dots, x_{D_L}) \in \{0,1\}^{D_L}\).

\paragraph{Step 5: Autoregressive Decoder.}

The decoder reconstructs the original SELFIES sequence from the binary code \(\mathbf{x}\). First, \(\mathbf{x}\) is projected to the model dimension:
\[
\mathbf{z} = \mathbf{x} \mathbf{W}_{\text{from\_latent}} + \mathbf{b}_{\text{from\_latent}} \in \mathbb{R}^{d_{\text{model}}}.
\]

This vector is repeated \(L_{\text{max}}\) times to form the memory tensor \(\mathbf{M}_{\mathbf{x}} \in \mathbb{R}^{L_{\text{max}} \times d_{\text{model}}}\). The decoder is a stack of \(M\) Transformer decoder layers, each containing:
\begin{itemize}
\setlength{\itemsep}{0pt}       
  \setlength{\parsep}{0pt}       
  \setlength{\parskip}{0pt} 
    \item A masked multi-head self-attention over the previously generated tokens (to preserve causality).
    \item A cross-attention layer where the query comes from the decoder state and the key/value come from \(\mathbf{M}_{\mathbf{x}}\).
    \item A SwiGLU FFN identical to the encoder's.
\end{itemize}

All sub-layers use Pre-LN and residual connections. Starting from a start-of-sequence token \([\text{SOS}]\), the decoder generates tokens autoregressively until an end-of-sequence token \([\text{EOS}]\) is produced or the maximum length is reached. At each step \(t\), the output probability distribution over the vocabulary is:
\[
\mathbf{o}_t = \text{softmax}\left( \text{Linear}(\mathbf{h}_{\text{dec},t}^{(M)}) \right),
\]

where \(\mathbf{h}_{\text{dec},t}^{(M)}\) is the decoder's final hidden state at position \(t\).

\paragraph{Step 6: Training Objective.}

The model is trained end-to-end by minimizing a composite loss:
\[
\mathcal{L}_{\text{total}} = \mathcal{L}_{\text{recon}} + \beta \cdot \max\left(0, \mathcal{L}_{\text{KL}} - \lambda_{\text{fb}}\right),
\]

where:
\begin{itemize}
\setlength{\itemsep}{0pt}       
  \setlength{\parsep}{0pt}       
  \setlength{\parskip}{0pt} 
    \item \(\mathcal{L}_{\text{recon}} = -\sum_{t=1}^{L_{\text{target}}} \log p(s_t^{\text{true}} \mid \mathbf{x}, s_{<t}^{\text{true}})\) is the cross-entropy reconstruction loss.
    \item \(\mathcal{L}_{\text{KL}} = D_{\text{KL}}\left( q(\mathbf{x} \mid S) \parallel \prod_{j=1}^{D_L} \text{Bernoulli}(0.5) \right)\) is the KL divergence between the learned posterior and an isotropic Bernoulli prior. For each dimension, \(q(x_j=1 \mid S) = \sigma(\mathbf{q}_{\text{logits},j})\).
    \item \(\beta\) is annealed linearly from 0 to a final value \(\beta_{\text{max}}\) over the first \(T_{\beta}\) steps.
    \item \(\lambda_{\text{fb}}\) is the free-bit threshold (e.g., 0.1 nats), which prevents penalizing KL values below this threshold, encouraging the model to use all latent bits without over-regularizing.
\end{itemize}

\paragraph{Validation of Latent Space Interpretability.}
After training, we validated the physical interpretability of the learned 128-dimensional binary latent space by examining its relationship with molecular synthetic accessibility (SA) score. Specifically, we encoded a held-out set of molecules into their 128-bit binary vectors and performed linear regression between the full binary vector (as input) and the SA score (as target). A strong linear correlation ($R^2$ = 0.8048, see in FIG. A ) was observed, indicating that the binary latent space captures overall structural scaffold dependencies and additive property trends. This finding confirms that the latent space is not an arbitrary embedding but encodes chemically meaningful, scaffold-related information, supporting the subsequent QUBO-based optimization that relies on a structured and interpretable landscape.

\begin{figure}[H]
\centering
\includegraphics[width=0.5\textwidth]{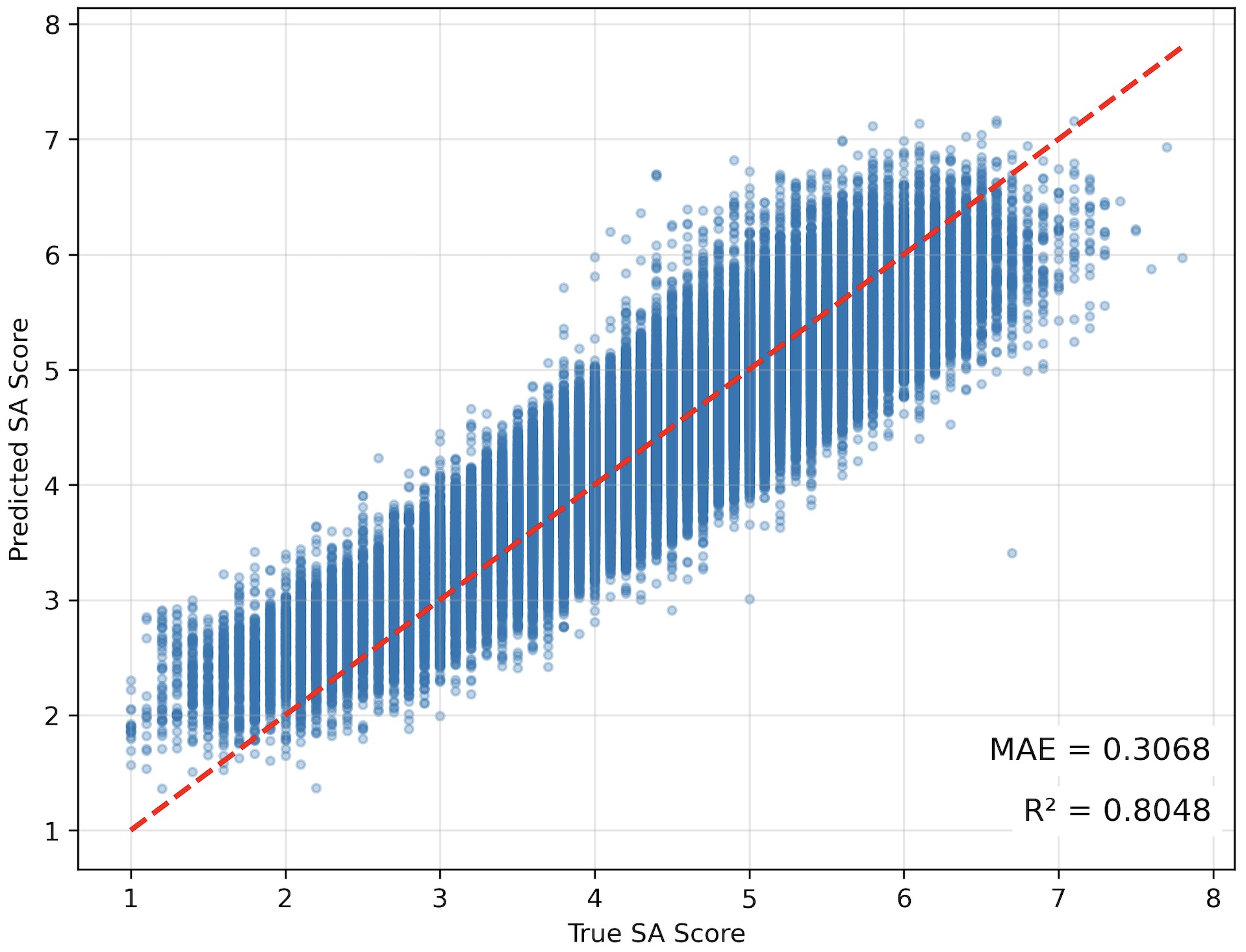}
\caption{True SA Score vs Predicted SA Score}
\label{fig:A1} 
\end{figure}

\paragraph{Remarks on Implementation Details.}
- The temperature \(\tau\) is annealed from an initial value (e.g., 2.0) to a final value (e.g., 0.5) over 10,000 training steps, allowing smooth transition from continuous to discrete sampling.
- The \(\beta\) annealing schedule starts from 0 and increases linearly to \(\beta_{\text{max}}=0.01\) over 50,000 steps, then remains constant.
- The free-bit threshold \(\lambda_{\text{fb}}\) is set to 0.1 nats (or 0.05 per bit on average) to avoid penalizing bits that have already aligned with the prior.
- The model is trained on a large unlabeled dataset of 10 million molecules sampled from Enamine REAL, ZINC, and iBonD using Morgan fingerprint proportional sampling to ensure chemical diversity.
- Post-training, the encoder and decoder are used separately: the encoder produces binary latent codes for any input molecule; the decoder generates molecules from any binary code, enabling property-guided optimization via QUBO or heuristic search.

This algorithmic pipeline establishes a direct, hardware-friendly mapping between chemical structures and fixed-length binary strings, forming the foundation for the subsequent QUBO-based molecular optimization described in the main text.

\subsection*{A6. Regression Results of pKa Prediction from Discrete Latent Space}

\begin{figure}[H]
\centering
\includegraphics[width=0.95\textwidth]{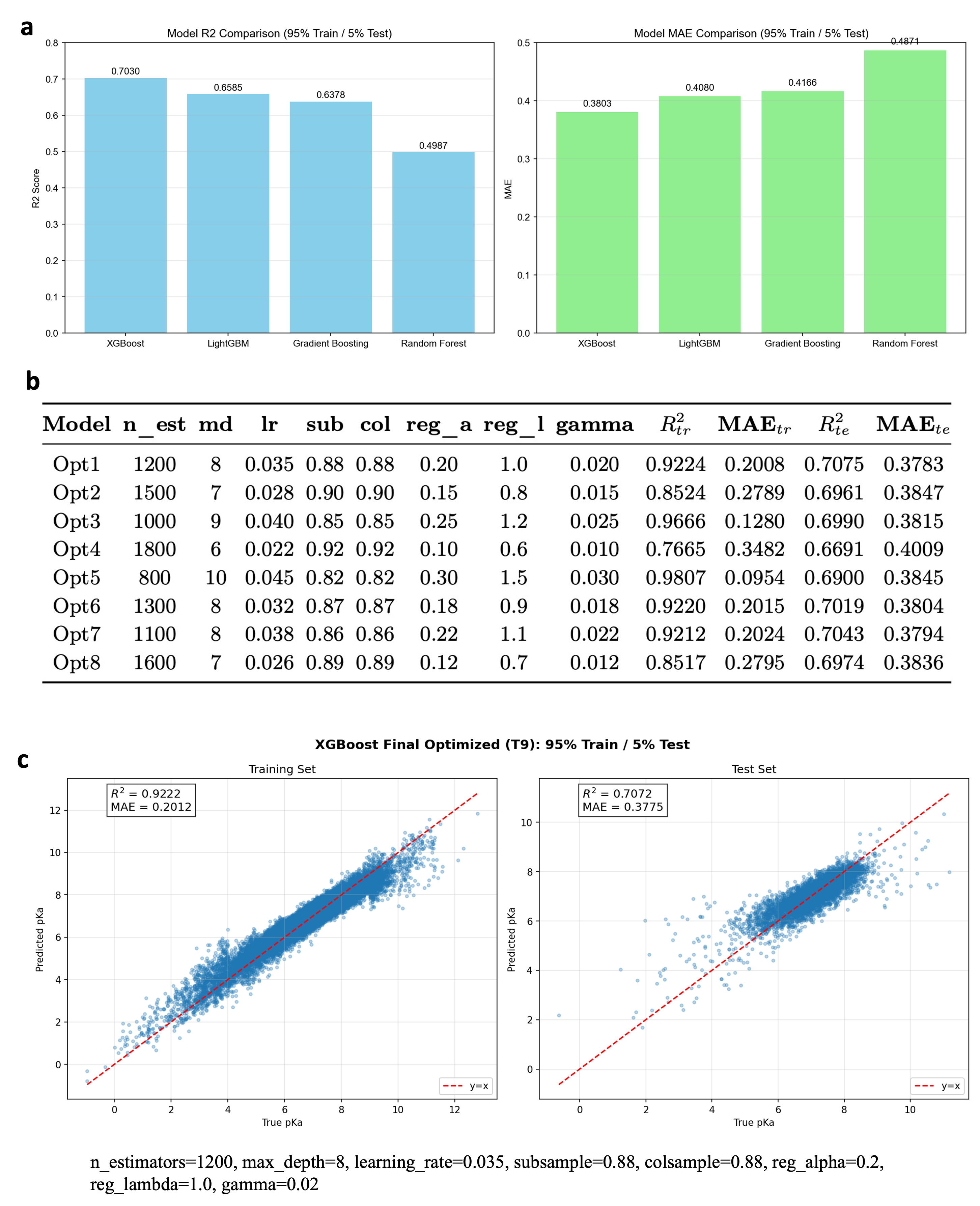}
\caption{pKa Regression Results of Latent Space}
\label{128} 
\end{figure}

\subsection*{A7. Detailed SimulatedAnnealingOptimizer}

The SimulatedAnnealingOptimizer was configured with the following parameters: initial temperature of 1000.0, cooling coefficient of 0.9999, cutoff temperature of 0.0001, 5000 iterations per temperature step, and a solution size limit of 500 per run.

In addition to the optimal results illustrated above, the sampling results of 200 solutions from QUBO matrices extracted from different FM variants are shown in the figure below. Evidently, models with poor long-tail fitting or low $R^2$ scores fail to objectively reflect the valid correlation between QUBO energy and FM-predicted values through QUBO extraction and solving.

\begin{figure}[H]
\centering
\includegraphics[width=0.95\textwidth]{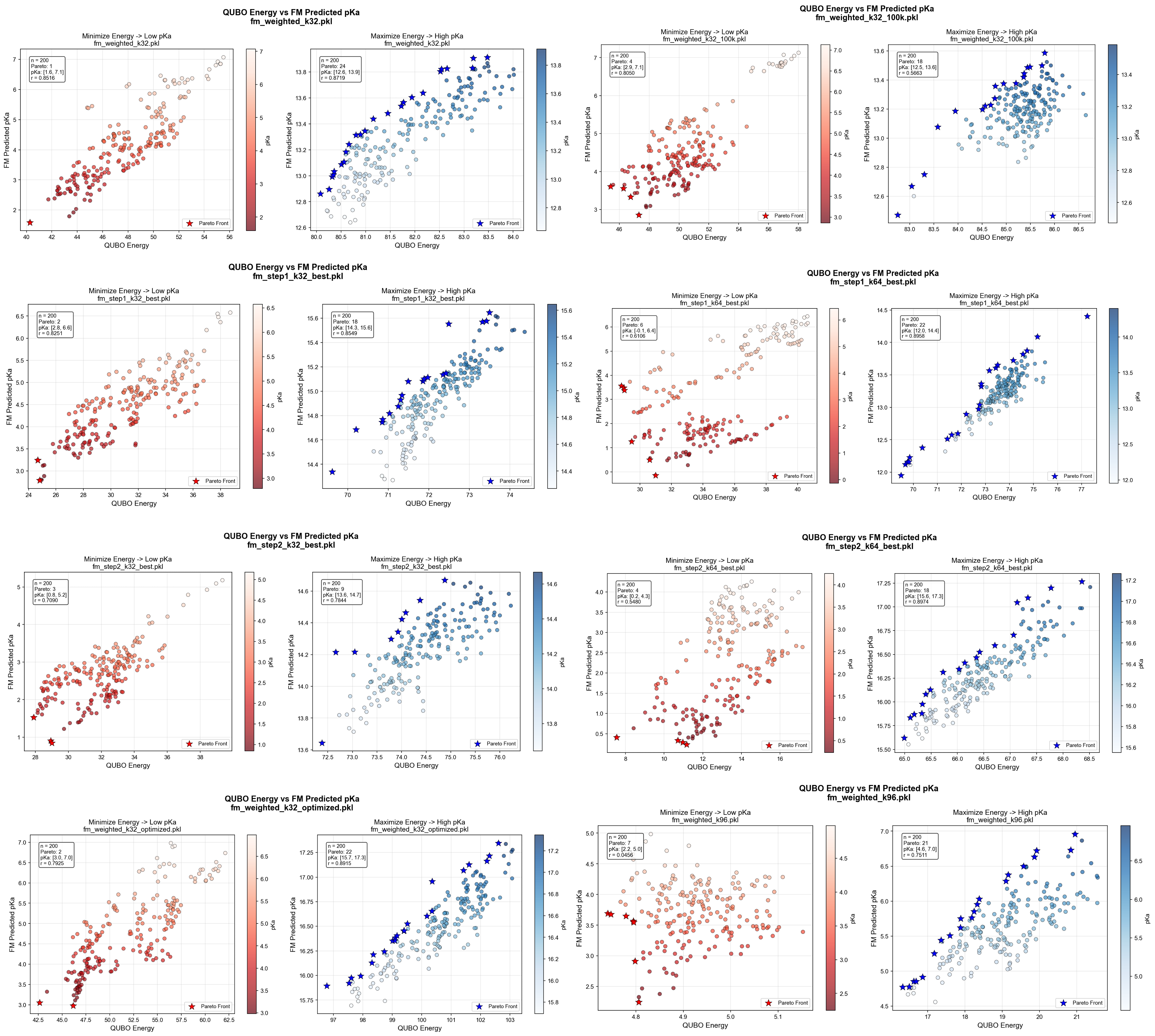}
\caption{SA Results 1}
\label{SA1} 
\end{figure}

\begin{figure}[H]
\centering
\includegraphics[width=0.95\textwidth]{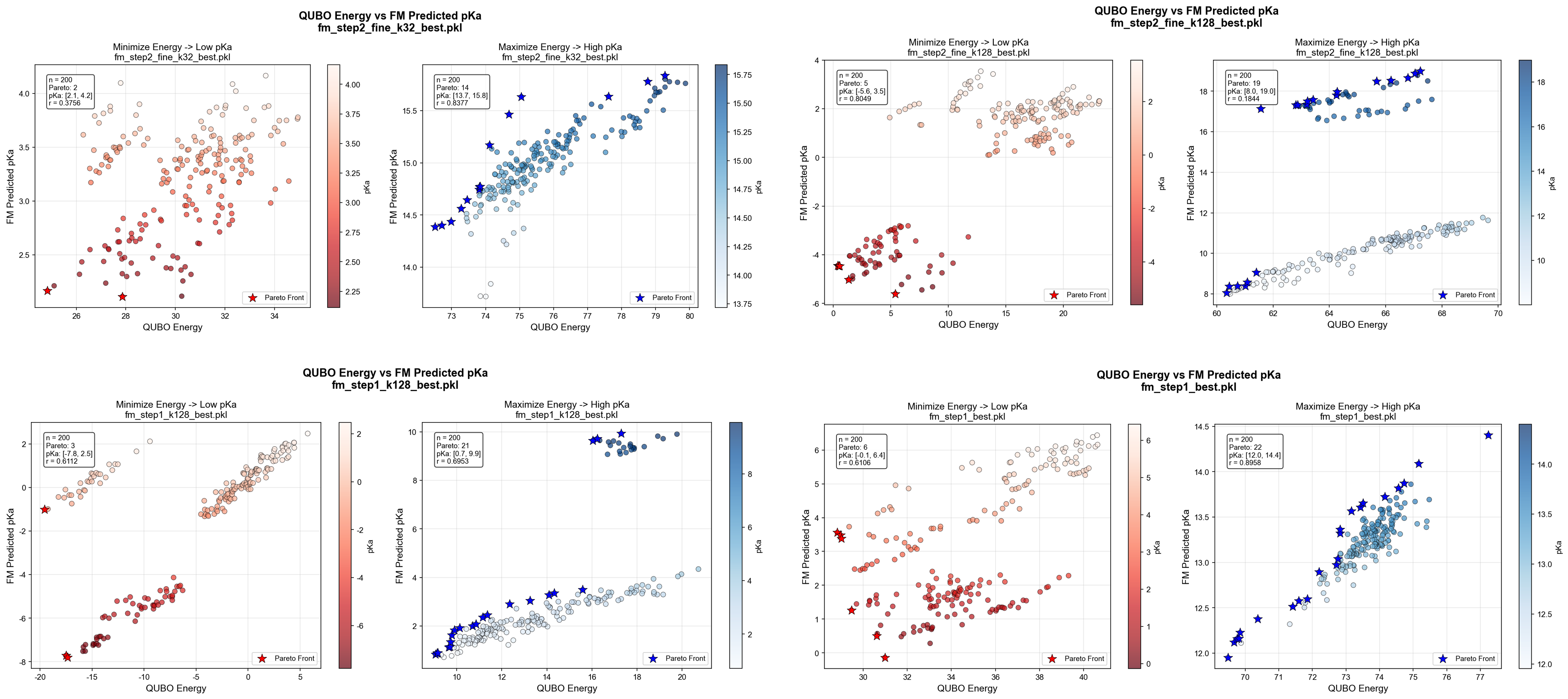}
\caption{SA Results 2}
\label{SA2} 
\end{figure}

\subsection*{A8. Detailed High-gain-based Filtering}

The CIM sampling produced numerous extreme FM energy values that exceeded the statistical confidence range. After removing low-confidence samples, as shown in Figure A16a, the predictions from FM and XGBoost no longer formed naturally separated distributions as observed in simulated quantum predictions; instead, the two sets of results were thoroughly mixed.
On this basis, we merged all valid data (excluding out-of-confidence-range samples) and performed filtering trials using the high-gain parameters derived from the XGBoost model. We first split the full dataset into two equal subsets (top 50\% and bottom 50\%). In each subsequent step, the newly introduced parameters were adopted to evaluate whether the data from the top or bottom half of the original dataset yielded better performance. The results indicate that this screening strategy, built upon the high-gain features of XGBoost, failed to outperform the quantum CIM sampling approach. Further details regarding more complex filtering schemes are available from the authors upon request.

\begin{figure}[H]
\centering
\includegraphics[width=0.95\textwidth]{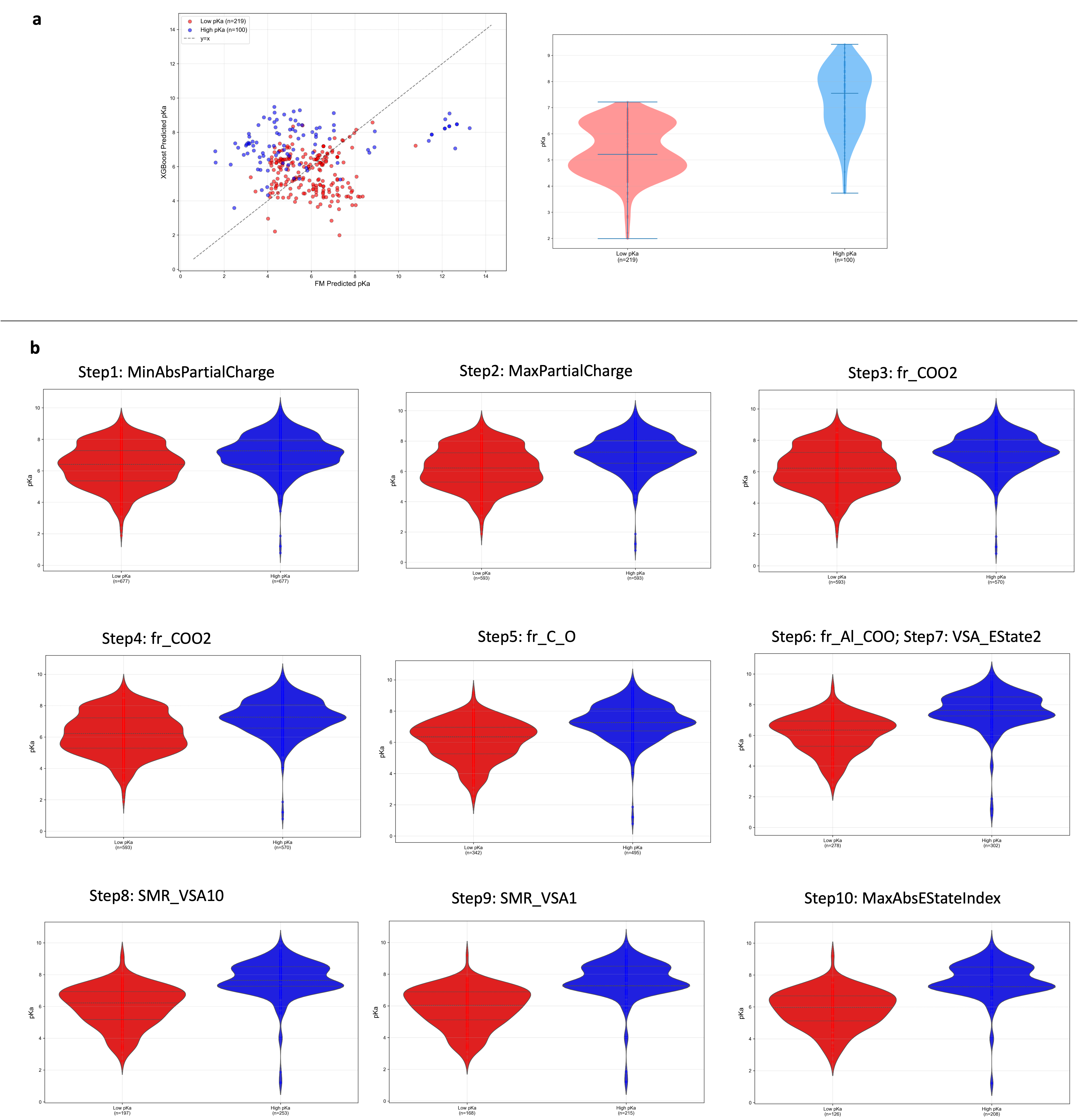}
\caption{a) Quantum Result, b) High-gain-based Filtering}
\label{ob} 
\end{figure}

\subsection*{A9. Detailed ADR Results}

\begin{figure}[H]
\centering
\includegraphics[width=0.9\textwidth]{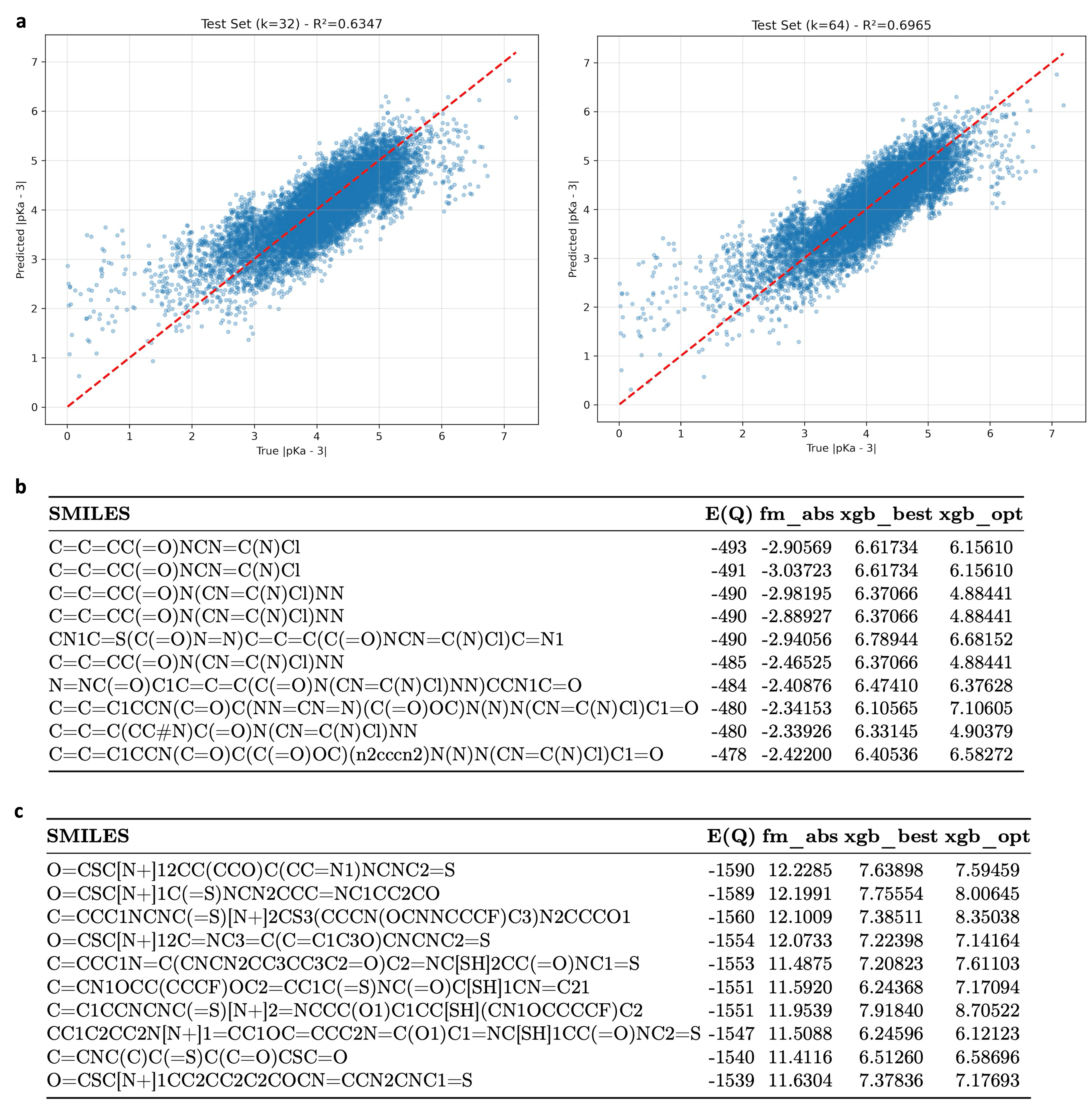}
\caption{ADR Results a) FM Results, b) 8-Bit Quota (low abs group), c) 8-Bit Quota (high abs group)}
\label{ADR} 
\end{figure}

\subsection*{A10. Detailed NDA Results}

\begin{figure}[H]
\centering
\includegraphics[width=0.85\textwidth]{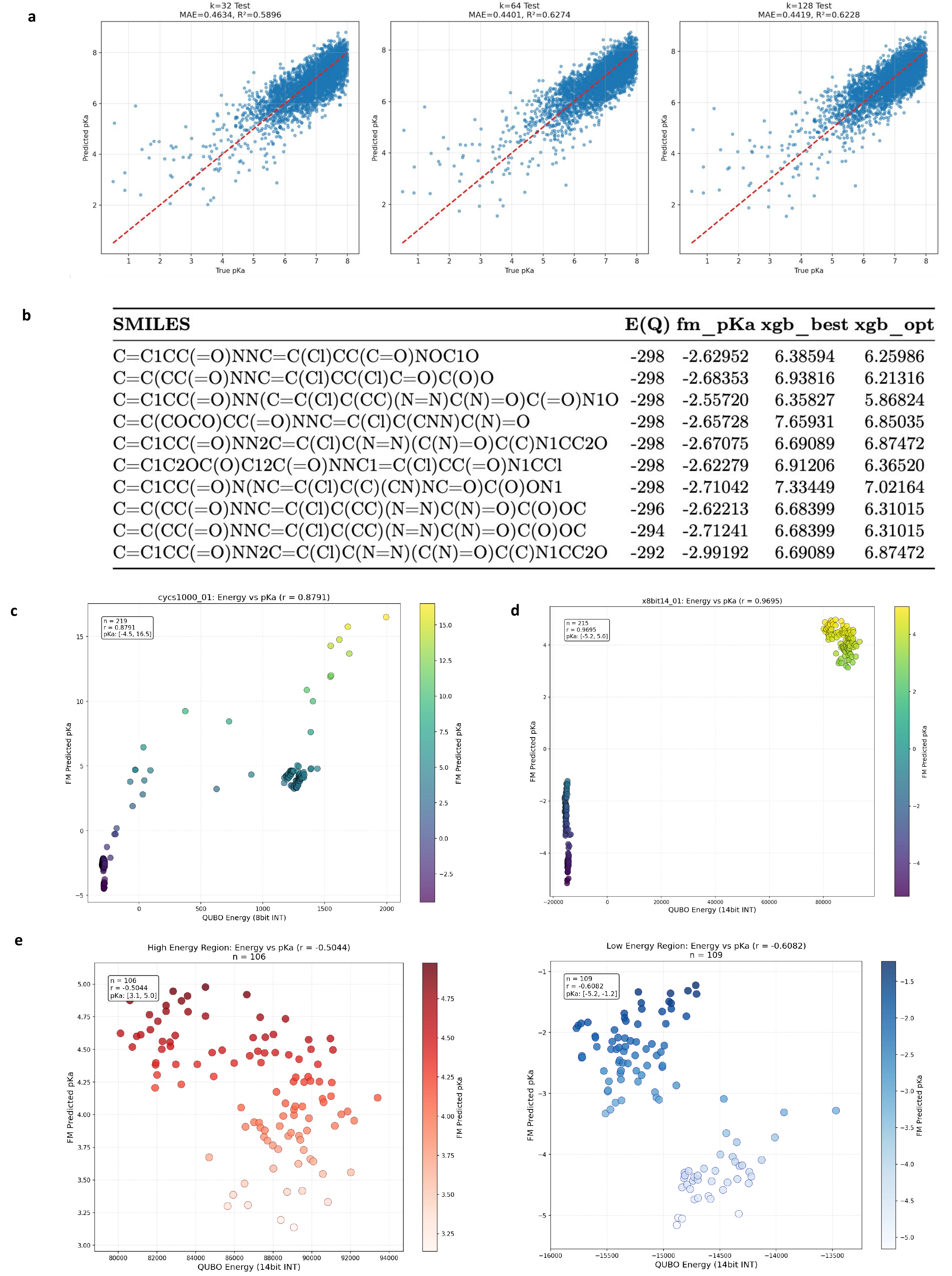}
\caption{NDA Results a) FM Results, b) 8-Bit Quota (low pKa only), c) 8-Bit Sample, d) 14-Bit Sample, e) 14-Bit Sample Fine-scale Features}
\label{NDA} 
\end{figure}

\subsection*{A11. Quantitative Description of Pearson and Spearman Correlation under QUBO Quantization}

Consider a QUBO energy function learned by a factorization machine (FM):

\begin{equation}
E(\mathbf{x})
=
\sum_i h_i x_i
+
\sum_{i<j} J_{ij} x_i x_j,
\qquad x_i \in \{0,1\},
\end{equation}

where \(h_i\) and \(J_{ij}\) are linear and pairwise interaction coefficients, respectively. After integer quantization with bit-width \(B\), the coefficients become

\begin{equation}
\tilde h_i = Q_B(h_i), \qquad \tilde J_{ij} = Q_B(J_{ij}),
\end{equation}

yielding the quantized energy:

\begin{equation}
\tilde E_B(\mathbf{x}) = \sum_i \tilde h_i x_i + \sum_{i<j} \tilde J_{ij} x_i x_j.
\end{equation}

The quantization error is

\begin{equation}
\epsilon_B(\mathbf{x}) = \tilde E_B(\mathbf{x}) - E(\mathbf{x}).
\end{equation}

Assuming a uniform quantizer with step \(\Delta_B \propto 2^{-B}\), we have

\begin{equation}
\mathrm{Var}(\epsilon_B) \propto 2^{-2B}.
\end{equation}

\subsubsection*{1. Pearson Correlation}

The above quantization analysis describes the relationship between the original QUBO energy
\(E\) and its quantized counterpart \(\tilde E_B\).
Under the assumptions

\begin{equation}
\mathbb{E}[\epsilon_B]=0,
\qquad
\mathrm{Cov}(E,\epsilon_B)=0,
\end{equation}

the Pearson correlation between \(E\) and \(\tilde E_B\) can be written as

\begin{equation}
\rho_P(E,\tilde E_B)
=
\frac{\sigma_E^2}
{\sigma_E\sqrt{\sigma_E^2+\sigma_{\epsilon_B}^2}}
=
\frac{1}
{\sqrt{
1+\sigma_{\epsilon_B}^2/\sigma_E^2
}}.
\end{equation}

Since

\begin{equation}
\sigma_{\epsilon_B}^2
\propto
2^{-2B},
\end{equation}

higher-bit quantization always provides a more faithful numerical representation of the original QUBO landscape.

However, in practical molecular design tasks, Pearson correlation is typically evaluated against the target property \(y\),

\begin{equation}
\rho_P(y,\tilde E_B),
\end{equation}

rather than against the original QUBO energy.
Assume the learned energy can be decomposed as

\begin{equation}
E=s+r,
\end{equation}

where \(s\) denotes the component correlated with the target property and \(r\) represents residual interactions.
After quantization,

\begin{equation}
\tilde E_B=s+r+\epsilon_B.
\end{equation}

Therefore, the resulting Pearson correlation depends not only on the quantization error but also on how quantization modifies the residual interaction structure.
Consequently, no general monotonic relationship between Pearson correlation and quantization bit-width is guaranteed.
Depending on the signal-to-noise composition of the learned interaction spectrum, either higher-bit or lower-bit quantization may yield a larger Pearson correlation with the target property.

\subsubsection*{2. Spearman Correlation and Ranking Inversion}

Spearman correlation measures ranking consistency. For two configurations \(\mathbf{x}_a\) and \(\mathbf{x}_b\), define

\begin{equation}
\Delta E = E_a - E_b, \qquad \Delta \epsilon_B = \epsilon_{B,a} - \epsilon_{B,b}.
\end{equation}

The quantized difference is

\begin{equation}
\Delta \tilde E_B = \Delta E + \Delta \epsilon_B.
\end{equation}

A ranking inversion occurs when

\begin{equation}
\mathrm{sign}(\Delta \tilde E_B) \neq \mathrm{sign}(\Delta E).
\end{equation}

If \(\Delta \epsilon_B\) is approximately Gaussian (by the central limit theorem over many independent weak interactions),

\begin{equation}
\Delta \epsilon_B \sim \mathcal{N}(0, \sigma_{\Delta \epsilon,B}^2),
\end{equation}

the ranking inversion probability is

\begin{equation}
P_\mathrm{flip} \approx 2 \Phi\Big(-\frac{|\Delta E|}{\sigma_{\Delta \epsilon,B}}\Big),
\end{equation}

where \(\Phi(\cdot)\) is the standard normal CDF. Spearman correlation decreases as \(P_\mathrm{flip}\) increases.

\subsubsection*{3. Effect of the Reverse Objective}

When training includes a reverse objective:

\begin{equation}
\mathcal{L} = \mathcal{L}_\mathrm{forward} + \lambda \mathcal{L}_\mathrm{reverse},
\end{equation}

the learned interaction spectrum is relatively smooth:

\begin{equation}
J = J_\mathrm{core} + J_\mathrm{weak},
\end{equation}

with many weak interactions still carrying meaningful ranking information. Low-bit quantization truncates these weak interactions, reducing \(|\Delta E|\) and increasing \(P_\mathrm{flip}\), thereby lowering Spearman correlation. High-bit quantization preserves them, yielding higher Spearman correlation:

\begin{equation}
\rho_S^{(14)} \gg \rho_S^{(8)}.
\end{equation}

\subsubsection*{4. Effect after Removing the Reverse Objective}

Moreover, the topology of the extracted QUBO provides direct evidence for this sparsifying effect.

Let

\begin{equation}
\mathcal E_B
=
\{
(i,j):
Q_B(J_{ij}) \neq 0
\}
\end{equation}

denote the set of nonzero couplings after quantization,
and let

\begin{equation}
N_B
=
|\mathcal E_B|
\end{equation}

be the corresponding number of effective edges.

Experimentally, both training objectives produce approximately the same number of couplings at 14-bit precision,

\begin{equation}
N_{14}
\approx
400.
\end{equation}

However, after removing the reverse objective, the number of surviving couplings after 8-bit quantization decreases dramatically,

\begin{equation}
N_8
:
256
\rightarrow
114.
\end{equation}

This observation indicates that a large fraction of interaction coefficients are concentrated near zero and become eliminated by low-bit quantization.

Consequently, 8-bit quantization behaves similarly to a hard-threshold sparsification operator,

\begin{equation}
Q_8(J)
\approx
\mathrm{HardThreshold}(J),
\end{equation}

removing hundreds of weak interactions while preserving only a small set of dominant couplings.

The resulting reduction in effective graph complexity provides direct structural evidence that low-bit quantization acts as an implicit sparsifying regularizer.

\subsubsection*{5. Summary}

\begin{enumerate}

\item Increasing bit-width always reduces the numerical quantization error of the QUBO coefficients and produces a more faithful representation of the original QUBO landscape.

\item Pearson correlation with the target property depends not only on quantization error but also on how quantization interacts with residual weak interactions. Therefore, no universal monotonic relationship between Pearson correlation and bit-width is expected.

\item Spearman correlation is governed by ranking inversion probability, which depends on both the intrinsic energy gaps and the effective perturbation variance.

\item When the reverse objective is present, weak interactions contain meaningful ranking information; preserving them through higher-bit quantization improves ranking consistency.

\item After removing the reverse objective, weak interactions become increasingly dominated by fitting noise. In this regime, low-bit quantization acts as an implicit sparsifying regularizer that suppresses noisy interactions and can improve ranking preservation.

\item Consequently, different Pearson and Spearman trends may emerge under different training objectives, even though higher-bit quantization always yields a numerically more accurate representation of the original QUBO energy function.

\item For the pKa QUBO problem specifically, 8-bit quantization induces a favorable level of natural sparsification (from 406 to 266 edges) that outperforms both higher-precision full-edge configurations and actively sparsified lower-edge configurations in our experiments. Active edge-selection methods such as negative data ablation did not show significant improvement even when preserving a similar number of edges as quantization-induced sparsification, indicating that magnitude-based filtering during quantization may align well with the intrinsic structure of the pKa energy landscape. 

\end{enumerate}

\subsection*{A12. Top 100 Molecules with Low and High pKa }

\scriptsize
\renewcommand{\arraystretch}{0.7}
\setlength{\tabcolsep}{4pt}

\begin{longtable}{@{}l r@{}}
\caption{Low-pKa Molecules} \\
\toprule
\endfirsthead
\multicolumn{2}{c}{{\tablename\ \thetable{} -- continued}} \\
\toprule
\endhead
\bottomrule
\endfoot
\verb|N#CC=C1Cc2cn(Br)cc2O1| & 1.2124 \\
\verb|N#CC=C1Cc2cn(F)cc2O1| & 1.7760 \\
\verb|N#C[N+]=C1Cc2cn(Br)cc2O1| & 1.7779 \\
\verb|O=C([O-])OC1=[N+]=CN1| & 1.8285 \\
\verb|CC(=O)C(CC(=O)O)=[N+]C#N| & 1.8641 \\
\verb|N#CC1OCC2C1=[N+]C(=O)C1(C=CCC1)[N+]=[N+]=C(F)C=CN2C(=O)O| & 1.8756 \\
\verb|O=C(O)C1=CC(=O)C(=O)OC2CCCC2CCCCCCC=C(O)NCS1| & 1.9781 \\
\verb|C=C=C1N(C=O)CC(Cl)C(C2N=N2)=P1(Cl)C(=O)O| & 2.1196 \\
\verb|NC=[N+]=C(F)C(=O)[O-]| & 2.2096 \\
\verb|CN(C=C1CC(Br)CC1OC(O)=[N+]C#N)OC(=O)[O-]| & 2.2490 \\
\verb|CC1CNC(CO)(C(=O)O)N2C=NCN=CN2C(=O)C2=CN(C)N1N2| & 2.3098 \\
\verb|N#C[N+]=C(O)OC1CC(Br)CC1=CNC(=O)O| & 2.5415 \\
\verb|N#C[N+]=C(O)OC1=CN(Br)CC1NC(=O)O| & 2.6775 \\
\verb|CC1=CCC(=O)C(C(=O)O)=NC=C(F)C2(F)N=NC=[N+]=C2O1| & 2.6987 \\
\verb|C=C=CC(=O)O| & 2.7169 \\
\verb|N#C[N+]=C(OCC=O)OC1(C(F)OC(=O)O)CC2(F)CC1C2| & 2.8277 \\
\verb|O=C(O)c1c(C(O)CN=CC2=C=C2)cnnc1Cl| & 2.8342 \\
\verb|NC=[N+]=C(O)OC1CC2=[N+](C=[N+](F)F)NN(C(=O)O)C1C2| & 2.8436 \\
\verb|CC(C)(COC(O)=CC#N)[N+]=C(F)F| & 2.8436 \\
\verb|C[N+](NC=O)[N+](=C1CC(CCCC(=O)O)=C1F)C(F)F| & 2.9061 \\
\verb|CCPC(F)(C=CC(=O)CC(=O)O)C[N+]C(=O)NC(C)=C=O| & 2.9091 \\
\verb|C=C=C(CC1N=N1)NC(=O)C=CNCNC1=CC(Cl)(C(=O)O)O1| & 2.9419 \\
\verb|N#C[N+]=C(O)OC1CC(Br)CC1NC(=O)O| & 2.9593 \\
\verb|CC1=CC(Br)C=C(C=CCC(O)=CC#N)C1| & 3.1451 \\
\verb|O=C(O)C1=C(Cl)N=CC1=CCOCCC=CCC1=C=C1| & 3.1708 \\
\verb|C=C=C(C=C1C2CN=C(C)OCC2(C(=O)O)CC1(C)C)NC(C)=O| & 3.2207 \\
\verb|CN(C)N(C(=O)O)C1CN([N+]C=C(F)F)CC1OC(O)=[N+]C#N| & 3.2208 \\
\verb|[N-]=[N+]=NC(CC(=O)O)CC1CCS(=O)(=O)NC2=C(C=C1CCO)C1C=CC=C1C=C2| & 3.2254 \\
\verb|N#C[N+]=C(O)OC1CC2=[N+](C=[N+](F)F)NN(C(=O)O)C1C2| & 3.2335 \\
\verb|C=CCC1CCC=CC=CC=CC=CC1OC(=O)C1CC1C(=O)C=CCC(F)=[N+]=[N+]=[N-]| & 3.2852 \\
\verb|N#COC1=CC(C(F)=C(F)F)([N+](=O)[O-])NC(=O)C1| & 3.3001 \\
\verb|CC1(F)C=CC(CC=O)C1C(=O)O| & 3.3001 \\
\verb|CC(C)(COC(O)=NC#N)[N+]=C(F)C=CNOOC(=O)O| & 3.3001 \\
\verb|N#C[N+]=C(O)OC1=CN(NC=C(F)Br)CC1N(C(=O)O)C(F)F| & 3.3273 \\
\verb|CCN(CO)C(=[N+]=CN)OC1CC2=[N+](C=[N+](F)F)NN(C(=O)O)C1C2| & 3.3320 \\
\verb|N#C[N+]=C(O)OC1(F)CCC1(F)F| & 3.3392 \\
\verb|CNCC=[N+]C(C)(CF)C(C(=O)O)C(=O)C(=[N+]=CN)OC(F)F| & 3.3513 \\
\verb|N#C[N+]=C1ONC(C(F)=CF)C(F)O1| & 3.3886 \\
\verb|N#C[N+]=C1ONC(C(F)=[N+]F)C(F)O1| & 3.3986 \\
\verb|CC(Cl)(NC(=O)O)C(Cl)=CNC(=O)NC1=C=C1| & 3.4114 \\
\verb|N#CN=C1OC2=C(C(F)(F)F)N=C1C2| & 3.4329 \\
\verb|CCCC1C2CC23C=NC=C(O)C2NN(CCC(F)(C(=O)[O-])S2)C(=O)NC13| & 3.4386 \\
\verb|CC1=NC(CC(=O)O)CCCCN1C(=O)Oc1ccccc1CCCC1C=CCO1| & 3.4451 \\
\verb|CC(=[N+]=CN)C(=O)[O-]| & 3.4470 \\
\verb|N#C[N+]=C(O)OC(F)C1=CC(=[N+]CF)C=C1| & 3.4563 \\
\verb|N#C[N+]=C(CC1C=CC(Br)C1)OCC=O| & 3.4614 \\
\verb|CN1C=C=C(OC2(CCl)C3=COC(=CN2C=O)[PH]3=O)C=C1| & 3.4683 \\
\verb|CC(OC(O)=[N+]C#N)N(F)C(F)=CF| & 3.4750 \\
\verb|O=C(O)CCC=CC(F)CC=C[N+](=O)[O-]| & 3.4836 \\
\verb|N#C[N+]=C(O)OC1CC2(Br)CC1ON2| & 3.4836 \\
\verb|NC=[N+]=C(O)OC(C=CNC(=O)O)C(CF)=C(F)F| & 3.5012 \\
\verb|CN(C=CC(F)(F)OC(O)=[N+]=CN)[N+]CCF| & 3.5121 \\
\verb|CN(CO)C(=[N+]=CN)OC(C)(NC(=O)O)c1ccn(Br)c1| & 3.5229 \\
\verb|NC=[N+]=C1ONC(=C(F)CF)C(F)(F)O1| & 3.5335 \\
\verb|CC1CCCC1C=[N+]=C1NC(=O)C(=O)C2(F)SCN1C2CC=CO| & 3.5471 \\
\verb|CCC1=C2CCCC2=[N+]=CN1C(=O)CC(=O)C(=O)C(C)=[N+]=[N+]=[N-]| & 3.5481 \\
\verb|NC=[N+]=C(O)OC1CC(NCC(F)F)=CC1N=C(C(=O)O)N(F)F| & 3.5629 \\
\verb|CC(Br)=CCC(F)(F)OC(O)=[N+]C#N| & 3.5709 \\
\verb|CCOCC(=O)OC1=C([N+](=O)[O-])CC(F)C1| & 3.5880 \\
\verb|N#C[N+]=C1ONC(C(F)CF)C(F)O1| & 3.5890 \\
\verb|CC1(F)COC(=[N+]C#N)OC1| & 3.6073 \\
\verb|N#C[N+]=C(CC1C=CC(F)C1)OCC=O| & 3.6379 \\
\verb|CN(C)C(CCl)(C(=O)CON=CC#N)[N+](=O)[O-]| & 3.6557 \\
\verb|CC(=O)C(=[N+]=CC(O)=[N+]C#N)C(C#N)=C(F)CCF| & 3.6568 \\
\verb|CC1C(OC(O)=[N+]C#N)=CN1Br| & 3.6626 \\
\verb|CC1=CCCCC2=C(F)SC(N1)C1(CCC(=O)C2=O)CC1| & 3.6686 \\
\verb|NC=[N+]=C(O)OC1C(CC(F)(F)F)CC1(F)F| & 3.6807 \\
\verb|CC1=CC=CC2CCC2C2C=CC(O)=NCS3=C(NC(=O)C(=O)C23F)O1| & 3.6938 \\
\verb|C[N+]=C(F)C1NC(=O)C=CCCOC(=[N+]=CN)C1C=COC(=O)O| & 3.7070 \\
\verb|N#C[N+]=C1ONC(CF)C(F)O1| & 3.7146 \\
\verb|CS(=O)C(=O)N=[N+]=[N-]| & 3.7200 \\
\verb|NC=[N+]=C1OC2C=CCC2=NC(C(F)CF)C(F)O1| & 3.7269 \\
\verb|N#C[N+]=C(O)OC1(F)CCN([N+]CCF)C1| & 3.7318 \\
\verb|CCCCCCCC=CC=CCOC(=O)N1NCSC2(F)C(=O)C=C1C2CCO| & 3.7376 \\
\verb|N#C[N+]=C(O)OC(F)C1=CC(Br)C=C1| & 3.7496 \\
\verb|N#C[N+]=C1ONC(=C(F)Br)C(CCON)O1| & 3.7516 \\
\verb|Cc1cc(N(C=O)CC#N)no1| & 3.7545 \\
\verb|Cn1c#cnn1| & 3.7713 \\
\verb|CC(=O)N(CN=CC=O)C(=O)C1=CC(C(=O)O)=[N+]S1| & 3.7716 \\
\verb|O=CCN(CC1N=N1)C1C(=O)N1C(=O)c1cc(C(=O)O)ns1| & 3.7728 \\
\verb|O=CCN(CC1N=N1)CN(C=O)C(=O)c1cc(C(=O)O)ns1| & 3.8004 \\
\verb|[N+]=C(C=O)C(=O)O| & 3.8023 \\
\verb|O=CC=NCN(C=O)C(=O)c1cc(C(=O)O)ns1| & 3.8065 \\
\verb|O=C1C(Pc2cnon2)C(=O)N1F| & 3.8240 \\
\verb|O=CCN(CNC(=O)c1cc(C(=O)O)ns1)C(=O)C1CC1O| & 3.8317 \\
\verb|CN(CN=CC=O)C(=O)C1=CC(C(=O)O)=[N+][SH]1C=O| & 3.8453 \\
\verb|O=CCN1CN(C(=O)c2cc(C(=O)O)ns2)C1=O| & 3.8453 \\
\verb|C=C(C=O)C(=O)O| & 3.8477 \\
\verb|CCOC(=O)C1=CN(C#N)C23CCCSCC(=C[SH]12)C(O)=NCC3C(=O)O| & 3.8538 \\
\verb|CCC=NC1=CC2=C3C(=O)N4SC(C(=O)O)=NN4C(=C2C)OC(C(C)=O)=[N+]=C13| & 3.8593 \\
\verb|CCCSC1C2=CN(C(=O)N3NC=C(C(=O)O)S3)C1N(C(C)=O)C=NC2=O| & 3.8676 \\
\verb|O=CNC1=[N+]=C(C=O)C(=O)C(C(=O)O)NC1=O| & 3.8677 \\
\verb|O=CCN(CNC(=O)C1=CC(C(=O)O)=[N+]S1)C(=O)C1CC1O| & 3.8916 \\
\verb|CCC(CC12NC=C=C(N=CC=[N+]=NC1=O)C(=O)O2)C(=O)O| & 3.8998 \\
\verb|COC(=O)C1C(F)CCN1CC1N=C2CC(Cl)(C=[N+]=C(NC(=O)O)N1)C2=O| & 3.9222 \\
\verb|COC(OC)(C(=O)O)C(F)=CC=CNC1CC1| & 3.9275 \\
\verb|O=C([O-])OC(=O)N1CC=NC1| & 3.9368 \\
\verb|CSC(C=CC(=O)N(C=O)CN=CC=O)C(=O)O| & 3.9386 \\
\verb|O=CC=C[SH](C[N+](=O)[O-])NC(=O)O| & 3.9429 \\
\verb|CCC1NC=C(C2(C(=O)O)CO2)C2NCC1(F)C(=O)NCSNO2| & 3.9468 \\
\verb|C1=CSCCCCCC(c2n[nH]nc2CCCc2cn[nH]n2)=C1| & 3.9546 \\
\verb|C=C=CSC=C(C=CC(=CC=NC=O)C(=O)O)[N+](=O)C=CN(CF)C1CCC1| & 3.9569 \\
\verb|C=C=C1C(=O)N1C(=O)[O-]| & 3.9738 \\
\verb|CC1(C)NC(=O)CNC2c3[nH]ccc3CC(=CC2(O)C(=O)O)N=C(O)C1Cl| & 3.9745 \\
\verb|CC1NC(=O)CN=CC(O)(C(=O)O)C=C2Cc3cc[nH]c3C1C(Cl)C(O)=N2| & 3.9819 \\
\verb|CC(=O)N1C=[N+]=C2C(=O)C=C21| & 3.9995 \\
\verb|CC(CCO)(C(=O)O)C12NC1=C=CN=CC=[N+]=NC2=O| & 4.0020 \\
\verb|O=C([O-])N1CC(Br)C1=O| & 4.0028 \\
\verb|CN1CCCC(C#N)[SH](NC(=O)O)C(=CC=O)SCC2C=C(CO2)C1=O| & 4.0037 \\
\verb|CC1(C(=O)O)CC2=NC=C=C3NC31C(=O)N=[N+]=C2| & 4.0098 \\
\verb|CCC(CC)NC=CCN(C(=O)O)C(C)OC=O| & 4.0098 \\
\verb|CN(C=CC(=O)Nc1cnoc1C1CC1C#N)C(C=O)C(=O)O| & 4.0098 \\
\verb|CCN(C(=O)O)C1OC2(CN(C)C1=O)CC(C)(C1(O)CC1)CN2| & 4.0117 \\
\verb|CN(O)C1CCCCCN1C=[N+]=C(CO)C1=CC(F)(C(=O)O)OCO1| & 4.0138 \\
\verb|CN(CCCC(C#N)[SH](NC(=O)O)c1cc1=O)C(=O)C1CCOCC1| & 4.0213 \\
\verb|CCC(CC1=Cc2c3c(cn2C(C)=O)C(=O)C13)=NOC(=O)O| & 4.0271 \\
\verb|CC1CC(=NC(C#N)(CF)C(=O)NCCF)N(C(=O)O)C=CC=N1| & 4.0342 \\
\verb|O=CNC1=[N+]=C(CC=CO)C(=O)C(C(=O)O)NC1=O| & 4.0419 \\
\verb|CC1C[N-]N=CCSC1=NC(=O)C1([N+]=CC(C)(CCO)C(=O)O)C=C=CN1| & 4.0454 \\
\verb|C=C1C(=O)CN(C(=[N+])C(=O)O)CN(CC=O)[SH]1C(=O)C1CC1O| & 4.0460 \\
\verb|O=C(O)CNC(=O)C1=CCCCCON(C(=O)F)C1| & 4.0507 \\
\verb|C=C(Cl)CCNC(=O)C=C1OC=C(C(F)C(C(=O)O)c2cocc2F)S1| & 4.0537 \\
\verb|CCOCCCC(=O)N(C)CCCC(C#N)[SH](C=CC=O)NC(=O)O| & 4.0544 \\
\verb|Cc1nnc(C=O)n1C| & 4.0585 \\
\verb|CCC=NSC1=CN2C(=O)NN=C(C(OC)SC=O)C(C=O)=[N+]=C[N+](C(=O)O)=C12| & 4.0640 \\
\verb|CN(C=C1C(=O)Nc2cnoc2CCC1C#N)C(C=O)C(=O)O| & 4.0655 \\
\verb|O=C(O)N1CC=NC1| & 4.0662 \\
\verb|CCN(CC(=O)O)OC(=O)C(Br)=NC(C)(C=CO)CNC(=O)SC=[N+]=CCNC| & 4.0860 \\
\verb|c1nc2cn-2cn[nH]1| & 4.0888 \\
\verb|C=CNC(C=CC(=O)NC(=O)O)=[N+]=CC(=O)NC1(Br)C=C2N3CCN2C31| & 4.0897 \\
\verb|CCCNCC(F)(C=O)ON=C(C)C=C(F)C(=O)O| & 4.0917 \\
\verb|CN1C(C=CC=NC=C(Cl)CF)C(=O)NC1C(=O)O| & 4.1006 \\
\verb|CCCSC1C2=CN(C(=O)N3NC=C(C(=O)O)S3)C1N(C(C)O)C=[N+]=C2O| & 4.1041 \\
\verb|C=C=CC(F)NC(=O)C(=CCC)C(=O)O| & 4.1078 \\
\verb|CCCSC1C2=CN(C(=O)N3NC=C(C(=O)O)S3)C1N(C(=O)O)C=[N+]=C2Br| & 4.1140 \\
\verb|CCCC1(C(C)O)CCCON(C(=O)O)PC=C(F)C1CF| & 4.1451 \\
\verb|C=C(Cl)CCNC(=O)C=C1NC=C(C(F)C(C(=O)O)c2cocc2F)S1| & 4.1938 \\
\verb|O=c1n2cnn1C2| & 4.2086 \\
\verb|CCCSCC1=CN(CC(C#N)N=CC(=O)OC)CN(C(=O)O)C=[N+]=C1O| & 4.2873 \\
\verb|CC1=NC2SCC([N+](=O)C=[N+]=C3CCC3S)=C(C(=O)CC(C)C)C1C2=O| & 4.2935 \\
\verb|CCCNC(=[N+]=C(C=[N+]CC(C=O)(NCCCON)C(F)F)C(=O)O)NC(=O)O| & 4.3057 \\
\verb|CCC(C=NCCF)NC=CCN(C(=O)O)C(CC)OC| & 4.3498 \\
\verb|N#CC(CCCC(O)CCC1CCOCC1)[SH](C=CC=O)NC(=O)O| & 4.3498 \\
\verb|CS(=O)C(C(=O)O)N(O)CC(=O)C1COC2CC2C1| & 4.3498 \\
\verb|C=CC(C)C(=O)C1=C([N+](=O)C=[N+]=C(S)C2CCC2)C2C=C(C=O)C1=NC=CC2=O| & 4.4023 \\
\verb|Cc1cc[nH]c1C1C(C)NC(=O)CSCC(=CC(O)C(=O)O)N=C(S)C1O| & 4.4245 \\
\verb|O=C(CN1N=CS(=O)C1C(=O)O)C1COC2CC2C1| & 4.4251 \\
\verb|O=C(CN1N=CS(=O)C1C(=O)O)C1COCCCN1| & 4.4615 \\
\verb|CCC=NC=C1C=C2C(=O)NN=C(C(OC)SC=O)C(C=O)=[N+]=C[N+](C(=O)O)C12| & 4.4662 \\
\verb|CC(=O)C1=[N+]=CC(c2ccco2)C(=O)N=N1| & 4.4729 \\
\verb|C#CN1N=CC([N+](=O)[O-])=NN=CCCN(C(=O)O)C=NC1(C)O| & 4.4819 \\
\verb|CCC(=CC(C)OC)NC=C1CNC(OC=O)(C(=O)O)C1| & 4.5155 \\
\verb|COC=CC1=CC([N+](=O)C(C)(CCl)N=N)=NC(=O)OC1(F)F| & 4.5444 \\
\verb|CC1=C([N+](=O)[O-])C=C1| & 4.5795 \\
\verb|CCC(=C=O)NC=CCN(C(=O)O)C(CC)OC| & 4.5881 \\
\verb|CC1CCC2C(COC3(C4C=CC=C(NC(=O)O)N4)CCCC23)NC1=O| & 4.5971 \\
\verb|C#CN(N=C(C=NN=CC)[N+](=O)[O-])C(C)(O)N=CN(C)C(=O)O| & 4.6094 \\
\verb|CC1=NC2=CC(=O)C=CC(C)CC(=O)C1=C([N+](=O)C=[N+]=C1CCC1S)C2| & 4.6597 \\
\verb|COC(F)N1N=CC2(N(C)C(=O)NC(=O)O)CC12| & 4.6615 \\
\verb|CN(C)CCN(CC(=O)O)C(=O)N1CC=[N+]=C(S)C1=O| & 4.6708 \\
\verb|CC(=O)ON(CC=O)CNC(=O)C1=CC(C(=O)O)=N[SH]1C| & 4.7060 \\
\verb|C#CC1NC([N+](=O)[O-])=NCN=CCCN(C(=O)O)C=NC1(C)O| & 4.7060 \\
\verb|CCC(=CC(C)OC)NC=C1CCC(OC=O)(C(=O)O)C1| & 4.7073 \\
\verb|CCC(NC=CC1(C(=O)O)CCCCC1OC)C(=O)C=C(C)F| & 4.7306 \\
\verb|C[SH](C)(=C[N+]1=CN(CC(=O)NC2=C3N=CCC=C2O3)N=CC=C1)CC(=O)O| & 4.7382 \\
\verb|CC[N+](C)=C(CC=C(N=CN)C(=O)NN=C1C(=[N+])CC1C(=O)O)C1CC1| & 4.7427 \\
\verb|C=CN=CCC(C(=O)O)C(=O)NC1(F)C=C(C)C=CC=CCCCS(=O)NC1| & 4.7450 \\
\verb|CC=CNC1(C)C=C(C(=O)O)C(=O)N2C3C=C1C2C(F)N(C)C3| & 4.7450 \\
\verb|CN1CC(=O)ON(C=O)C(C=NNC=O)C=[N+]=C1C(=O)C(=O)O| & 4.7450 \\
\verb|CC=CC(C)(C=CNC1(C(=O)C2C3CCOC2N(C)C3)CCCC1CCl)NC(=O)O| & 4.7657 \\
\verb|NCC(CNC=NN=C(CO)C(C=O)(C(=O)O)n1ccnn1)C1(O)C=C1| & 4.7928 \\
\verb|O=[N+]N(C(=O)NCCNC(=O)[O-])C1CCCO1| & 4.7997 \\
\verb|COCCN(C(=O)O)C(=O)N(C)C1CC1| & 4.8025 \\
\verb|CC(Br)C(=O)C(CCBr)(NC=CO)NC(=O)[O-]| & 4.8126 \\
\verb|N#CCCC1CC2(O)OC=NC(=O)C3=COCNC(NC(=O)O)=C3C(=O)NN=C12| & 4.8278 \\
\verb|C=CCC(CC)NC(=O)C(=CC(C)F)C(=O)O| & 4.8310 \\
\verb|CCC(CCNC(=O)Cc1cn(C(C)=O)cc1C=O)=NOC(=O)O| & 4.8509 \\
\verb|C=C(F)C(C(=O)[O-])N(CCC=NC=CC)C(=O)N1C(CC#N)C1(C)C| & 4.8972 \\
\verb|CCC(=C=O)NC=CCNC(CC)(OC)C(=O)O| & 4.9096 \\
\verb|CCC(=CCCO)NC=CCNC(CC)(OC)C(=O)O| & 4.9253 \\
\verb|CCC(=CC=CSN)c1n[nH]nc1CCCc1cn[nH]n1| & 4.9269 \\
\verb|C#CC1C=CC=C(N(C(=O)O)C(CC)OC)OC=C1C| & 4.9337 \\
\verb|O=[N+]N(CCO)C(=O)NCCNC(=O)[O-]| & 4.9507 \\
\verb|CC=CC(=O)C1(Br)CC(=O)N(CO)CC=NN=C(Cc2ccccc2O)N1C(=O)[O-]| & 4.9700 \\
\verb|CCC1N(CN(CC=O)CC2CO2)C(=O)C=CC(C(=O)O)[SH]1C| & 4.9719 \\
\end{longtable}

\scriptsize
\renewcommand{\arraystretch}{0.7}
\setlength{\tabcolsep}{4pt}

\begin{longtable}{@{}l r@{}}
\caption{High-pKa Molecules} \\
\toprule
\endfirsthead
\multicolumn{2}{c}{{\tablename\ \thetable{} -- continued}} \\
\toprule
\endhead
\bottomrule
\endfoot
\verb|CN=C1C=C(CN)N(C)C1NC(C)c1ccccc1| & 9.8899 \\
\verb|C=CCOC1NC=CCNC(F)=NOC(O)CC1C| & 9.8791 \\
\verb|CCC(C=NC=C(C)CO)NC=CCN(C(C)=N)N(C)CCF| & 9.8650 \\
\verb|CCC(C=NC(C)CO)NC=CCN(C(C)=N)N(C)CCF| & 9.8579 \\
\verb|CNC=C[SH](C)(=O)NC1CCNC1| & 9.7809 \\
\verb|C=COC1NC=NC=C1CC1=CSC2=[SH]3(CCC3(CC)OC2C)C1| & 9.7496 \\
\verb|C=C1CCNCC=C(C)c2ccccc2C(C)NC1| & 9.7485 \\
\verb|CCC1=NCC=CN1| & 9.7057 \\
\verb|CC(O)SC(=NNCC(C=CC=C[SH]1C=CC1C)CF)N(C)S| & 9.6858 \\
\verb|C=C=C(C)C1NC(CC)(C(F)F)C(N(C)F)NC(=S)NCC1C=CC| & 9.5841 \\
\verb|CCCNC=CC(NC(C)CO)=C(C)CCO| & 9.5836 \\
\verb|C=C(C)CCCC1SC=CC2=NC(Br)=C1CN2CC1CC1=CC1(C)CC1O| & 9.5788 \\
\verb|C=CC(F)=C1C=C2CCCN(ONC2(C)C)[SH]2C=C(CCCNCC2)C1| & 9.5775 \\
\verb|CN1CCC2CNCCNC(Cl)(C=C3NC=CC=C3C2)N1| & 9.5725 \\
\verb|CC(=NC(=N)N(C)CCCO)OCN1CC=CN2NCN=CCC(=O)OC12| & 9.5574 \\
\verb|CCCCC(CC=N)(CNC1=C[SH](C(CC)CN(C)C)N1)N(C)O| & 9.5451 \\
\verb|CCOCNCCCCCCN(C)C(=O)C=C1NC=CC1=N| & 9.5336 \\
\verb|C=C1CCN2CCCC(C)C2C2(CN3CC4CC3C4)CC(C)(NC1)c1ccccc12| & 9.5218 \\
\verb|CN(C)CC1=C(F)C2=C3C=NC3CCCNC=C(N(C)C)CCNCC2=C1| & 9.5152 \\
\verb|CNCCC1(C)C=CC(=C2C=NC(N)=C2Br)C=CN1CC(CN)CC(C)C| & 9.4766 \\
\verb|C=CNC=CCCC(C)NC1(Br)CCNC1| & 9.4688 \\
\verb|CNC(CN(C)C)C1=[N+]=CC=CC1CCF| & 9.4476 \\
\verb|S=C=[SH]C=CCC1=CN=CNCOC=CC1| & 9.4428 \\
\verb|CNCCC1(C)C(=CC=C2C=NC(N)=C2Br)C(C)=CN1CC(CN)CC(C)C| & 9.4279 \\
\verb|C1=CCNC2CC(=C1)C2| & 9.4266 \\
\verb|CON(N)C(O)(OC)N1CCCCOCC(=N)N(CC2=COC(CCCO)C2C)C1| & 9.4079 \\
\verb|CCN(C)C(C)C(C=CCC=C(C)NC(C)C)=C(F)CNC(C)C(N)OC| & 9.3867 \\
\verb|CC=CC1CCN([SH]2C=C(N(C)F)N2CC)SCCC=C=C1C| & 9.3800 \\
\verb|N=C(C=CN1CC=NOCCNC1=O)C1=CC=CON=CC=CC1| & 9.3793 \\
\verb|CC=CC=O| & 9.3686 \\
\verb|C=NC1(NC=CC=N)OCC(C)(NC)C1N(C)C=O| & 9.3453 \\
\verb|CC=CC1CNC1C(=O)NCN=CC(NC=NC=N)OC| & 9.3453 \\
\verb|C=C1CC(C=C(F)CNCC=C(C(C)=S)N(C)C)S1| & 9.3415 \\
\verb|C=COCC(O)(C=CN=CC=C[NH])NC| & 9.3365 \\
\verb|COC(C)(F)CS(N)(=O)=O| & 9.3339 \\
\verb|N=CC1COC(C=CNC2=CN2)C1C(F)F| & 9.3244 \\
\verb|C=CCNCCPc1ccccc1C1(CN(C)C)NC(O)CN1C| & 9.3096 \\
\verb|CCCNCCPc1ccccc1C1(N(NC)C(C)O)CN(C)C1| & 9.3060 \\
\verb|CC1N[N+]2=C(Cl)NC(CC2)C(C)C(C)C1NCCCN=[N+]=[N-]| & 9.3016 \\
\verb|CCC1(CNC(=O)C2=CN2CC=CC=N)CN1C1C=CCN1| & 9.2981 \\
\verb|CC(=CC(=N)ON)OCN=CNCC=CN1COC(=O)CC=NCN1| & 9.2977 \\
\verb|C=NC1CN2CC=C(F)C=CC1(C(C)CNCCC(O)NC(C)O)C2| & 9.2897 \\
\verb|N=CON1CCC(NCC=CNC2(CCO)CC2)COC1| & 9.2806 \\
\verb|CC(=N)C(=O)C=CCNC=NCNCC1COCC(=O)C=C(O)CC(C#N)O1| & 9.2764 \\
\verb|C=COCCNCCCN1CCS(=O)C12C=CC(CC=C(Cl)OC)=N2| & 9.2756 \\
\verb|C=C1CNC23C=C2NC3N=C1C1(N)NCCCCN1| & 9.2751 \\
\verb|CCC1NCP2C3=CC=[N+]=C4C(C=CC=CN2C3F)CC1C1CC1C41CN(C)N1| & 9.2603 \\
\verb|C=NCNC(C)C=[N+]=C1CC=C[SH]1C| & 9.2518 \\
\verb|CN(C)N1C2c3ccccc3CCC3CCCC14C=C(COCC4N3)C2N| & 9.2497 \\
\verb|CC(NCC=CNC1(CCO)CC1)C1OCN(C=N)CO1| & 9.2449 \\
\verb|C=NCNC(=CC=CS)CCC#N| & 9.2400 \\
\verb|CCCN1CCCC1CC(CF)=C(NN(C)C)C(C)C| & 9.2346 \\
\verb|CNCC1=CC(=O)N(NCC=CC(=N)N(C)C)C(=O)C1=NC=N| & 9.2289 \\
\verb|CC=CC(C(=O)NC1NCC1(O)CC(C=N)C(=O)C1(COC)CC1)N(C)C[NH]| & 9.2277 \\
\verb|CN1CC2=CC=[N+]=C(C=C(CNCC=C(F)N(CO)CON)N2)C1| & 9.2268 \\
\verb|CCNC=C=N[SH](C)(=O)C=CNC| & 9.2222 \\
\verb|CC1=C(N)OC2=NCNC3(COCN1C)CC3OC=C=C2| & 9.2199 \\
\verb|CC1=C(N)OC2=NCNC3(COCN1C)CC3OC=C=C2| & 9.2199 \\
\verb|C=C(C)C(=O)N1[SH]=C(NO)NC=C1C(=CCCNCC=NC)P1CCNC1C| & 9.2092 \\
\verb|CC1(NCCCON=N)CCNCCC=NNC=C(N)COC1| & 9.2071 \\
\verb|CCCN(CCCCCC1CCNCC2(CC)C1=CCC=[N+]2C(C)NN(C)C)NC| & 9.2025 \\
\verb|CNCCC1(C)C=C2C=CC([N+](F)F)=CC=CC(C(C)C)C(CN)N1N2| & 9.1979 \\
\verb|C#CC(C)(C)S(N)(=O)=O| & 9.1854 \\
\verb|NC(CF)CC1CNCO1| & 9.1748 \\
\verb|COCCC(CC1CCCO1)NCC1=NC=C2CC21| & 9.1719 \\
\verb|N#CC(C=NO)=C1CC2=CC3=[N+]=C2CC3=COC=CN1| & 9.1695 \\
\verb|C=CCNC1C2NC3CC4CN(C)C=C(N4)OC(=N3)C12| & 9.1684 \\
\verb|C=C(C)CCCC=NC(CCCO)=C(CC(C)C)CN1CC=CC1C(C)(N)O| & 9.1658 \\
\verb|CC(N)NC(O)=CC=CC=CPC(F)F| & 9.1559 \\
\verb|C=C1C(N(C=C2C=C(CCO)C=N2)CCCCNCC2CC2)=CC=CC1NN(C)C| & 9.1556 \\
\verb|CC=C1OCC2=CCNCCC(=O)C1(C)CN(C)C(C)(O)N2| & 9.1447 \\
\verb|CC1=CC2CNCC3OCC(C)N3C2=CCC=C(CF)N1| & 9.1412 \\
\verb|C=CNC(C=C1C=[N+]=CC2(CN(C)C2)N(C(CO)C2CCC2CC)N=N1)N(C)C| & 9.1406 \\
\verb|CC1=CC(C)(O)N(C)C1| & 9.1344 \\
\verb|CCCNCCPc1c2cccc1C1(CN(C)C1)N(C(C)O)CC=CC(F)=C2| & 9.1342 \\
\verb|C=CNC(N)C(=O)C=CN1CN(CCCCC=C(CO)CC=N)C(C)N1| & 9.1340 \\
\verb|CNS(N)(=O)=O| & 9.1300 \\
\verb|CN(C)C=[N+]=C1C=CC=C(F)C2CCNC(NC=CN=CCCC1)C2NC(O)CO| & 9.1248 \\
\verb|NC1CSC1| & 9.1245 \\
\verb|CPc1ccccc1C(=CNC1CC12C(=COCCF)C2C)CN(C)C| & 9.1233 \\
\verb|CCC(=CNCCC(O)C1C=CC(F)=C(CN(C)C)C1)N(C)C| & 9.1227 \\
\verb|C=C(C)NC1(Cc2ccccc2F)C(CNCC2CC2)C2(CO)C1N1CCC12C| & 9.1206 \\
\verb|CC(C)CC(C=CN(C(F)(F)F)C12CON(C1)C(NO)=CN2)CC#N| & 9.1150 \\
\verb|CCC1=NC(C(C)N)OCC(N(C)O)=CC=CN1| & 9.1142 \\
\verb|CNC1CC2CN(C(C)(CO)C3CC3C)C3=CCNC=C3C1(C)c1cccc2c1C| & 9.1131 \\
\verb|C=CCC(CO)C(=O)CC(CC)N(C)C(=O)C1(CC=N)CCOC1| & 9.1063 \\
\verb|CC1=[N+]=CC=C(NCCNCC=C(F)N(CO)CON)C1N(C)C| & 9.1052 \\
\verb|C=C1NC(N(C)C)C=CC=[N+]=C2NN=CN(C23CN(C)C3)C1(CO)N=N| & 9.1042 \\
\verb|CN1CC(=N)C(O)C(Nc2ccsn2)C2=CNC(C#N)(C=COC=CN1)CN=C2| & 9.1000 \\
\verb|C=CCC1N=CC=C1C(C)S(=O)C1CC2CC2CC(OC)(OC=C)C2=C1CNN2| & 9.0985 \\
\verb|NCC1NC(CN2CCCC2)C(CNC2OC3C(F)CC23)=C1c1ccccc1| & 9.0975 \\
\verb|CN1CCN2CCC2(C)C(CO)(CC2CC2)C(CN)C(C)(c2ccccc2)NC1| & 9.0964 \\
\verb|N=CN=CC(=O)CN=C(N)NCN=CC(=O)NCCOCNC=CC=NON=C1CCCC1| & 9.0944 \\
\verb|COC(Cl)=CCC1=NC2(C=C1)S(=O)CCC21C=COCCNCCC1| & 9.0926 \\
\verb|C=C=COC1CC12CCN=CN=C(O)C(=N)N(C)COC2| & 9.0889 \\
\verb|CN(CNC(Cl)Cl)c1cc2cn1C(Cl)=[N+]=C2N1CCOCC2(CCOC2)C1| & 9.0855 \\
\verb|CCCN1C=C2CNC(=O)COC(C2)NC1=CCO| & 9.0826 \\
\verb|CCCN1C2=CCNC(=O)COC(C2)NC1=CCO| & 9.0825 \\
\verb|CC=C=O| & 9.0794 \\
\verb|CNC1C=CC=CC(Cl)N(CC(CO)(NC)OC)CCCC(F)N1| & 9.0792 \\
\verb|COC(Cl)=CCC1=NC2(C=C1)S(=O)CCC21C=COCNCCCC1| & 9.0714 \\
\verb|[N+]COC1=C[SH](N2C=CC(F)=CC=CCCC2)C[SH]=NCCC1| & 9.0697 \\
\verb|CCSC1CCCCCNC(F)=CC(=CN)C2(CNN=N2)O1| & 9.0673 \\
\verb|CC1CC(CC#N)C=CN(C2(C)NC3=CNC(C3)O2)C1| & 9.0666 \\
\verb|C=NC(O)CCC(NC1CC1)C(F)(F)F| & 9.0661 \\
\verb|CCCC=CC12C=CNC(=CCN=C(F)C=C1)C(O)=C1SC(N(C)C)=CNC12| & 9.0660 \\
\verb|CC1NCN(C)CCNCC2OCC(CNCC3CC3)C2(C)c2ccccc21| & 9.0642 \\
\verb|CC1=CNCCPNN(C(C)O)C=CC=NC=CC(F)=CC=CC(NC(C)C)=C1| & 9.0641 \\
\verb|CCNC1=C(C=CN2C(CC)CC2CCO)CC(O)OCO1| & 9.0637 \\
\verb|CNC(O)(O)C1=[N+]NCC=CN(C(F)(Cl)OC=CCN)CCNCN=C1| & 9.0633 \\
\verb|CC1CC=C(CO)C1| & 9.0622 \\
\verb|CC1=[N+]=CC=CCN(CCCN2CC=C(O)N=C2CO)CC(C)C1| & 9.0603 \\
\verb|CCCC=CC12C=CC=C(C=C(F)C=C1)C(O)=C1NC(N(C)C)=CNC12| & 9.0558 \\
\verb|CC1=C2C=CC(N2)C2=CC3(O)COC=CCC3OC(=NN2)CN1| & 9.0536 \\
\verb|CC=C(F)C12C=CC(NCCc3ccccc3C3(CN(C)C3)C1)OC2| & 9.0532 \\
\verb|COC=CN1CN=CCC2CC(CCOC(OCCCO)N1)CN2| & 9.0457 \\
\verb|CC1=CCCNCN(C)C=CCC1(C=CNC1CCC(F)CN1)CCOCC1CN1C| & 9.0445 \\
\verb|CN(C)CC(F)c1ccccc1CCNC1CC1(C)C=CCO| & 9.0430 \\
\verb|CC1C=NC(CO)=CC1| & 9.0407 \\
\verb|CCc1n[nH]cc1OC1C=C(C)OCCN1CC(=O)C(C)=[N+]=NNC(=O)CC=NC=N| & 9.0347 \\
\verb|N=CN=CC(=O)CN=C(N)NCN=CC(=O)NCCOCNC=NC=CON=C1CCCC1| & 9.0326 \\
\verb|CC(=NC(=N)CON=CCCCO)OCN1CC=CN2NCN=CCC(=O)OC12| & 9.0314 \\
\verb|CN1CC(C=CSC=CC(F)N(C)C)(CC2=C(CCO)C=N2)C1| & 9.0302 \\
\verb|CON(N)C(O)(OC)N1CCCOCC(=N)C1NCC1=COC(C(C)C)C=[N+]=C1| & 9.0299 \\
\verb|C=COC1=NOC=CC[SH]2C(C=CC(C=COC)C2C)CN2CCCCC2CCN1| & 9.0291 \\
\verb|C=C=C(C(N)c1ccco1)C(F)(CCNCCO)N=CCC| & 9.0226 \\
\verb|CCC=CC1CC=CNC=CCN(C(=O)CO)CO1| & 9.0204 \\
\verb|C=CNCC=C(F)C=CC1=[N+]=CC(CN(C)C)NC=C(CO)N(C)CC(C)=C1| & 9.0164 \\
\verb|C=CC1C(=CC=N)C(=O)N2[SH]=C1NC=C2C1=CCCNCC=NC1| & 9.0149 \\
\verb|CCC1CN(C)C=C2NC3=CS1(CO2)CC(C(C)C)(N(C)O)CN3| & 9.0149 \\
\verb|CC=NN(C=CC=NC=CC(F)=CC(CN1CCCC1)=C(C)NN(C)C)C(C)O| & 9.0121 \\
\verb|C=C(N)CC1C2=CCN1CNC=CC21NC(=N)N1C(CC)CCONCC| & 9.0114 \\
\verb|CN(CC1=CCN1)NCCS| & 9.0108 \\
\verb|CCCNC=CC(NC(=O)C1CCO1)=C(C)CCO| & 9.0092 \\
\verb|CC1=NC(CNCCCC=CC=CC=C(F)NN(C)C)C1| & 9.0078 \\
\verb|C=C1C=C2CNCN(C=NC=CCCON=CC)CCOC(=NC=N)CC(=O)N2C1| & 9.0058 \\
\verb|CCC1CN(S(=O)(=O)CCC2CC23NCC=NCCc2[nH]cc4c2C43)C1| & 9.0042 \\
\verb|CC1C(C=[N+]=C2C=C3C=CC2(NN(C)C)O3)NC2CC21C| & 9.0040 \\
\verb|CCCNCCPc1c2cccc1C1(CN(C)C1)N(C(N)O)CC=CC(F)=C2| & 9.0028 \\
\verb|N=CN=CCC(=O)OCNC=CNc1cnoc1CC1=CC=CCON=CCC1| & 9.0002 \\
\verb|CCCNCCPc1ccccc1C1(N2C3CC32O)CN(C)C1| & 8.9990 \\
\verb|CNCCC1(C)C=CC2=CC([N+](F)F)=CC=CC(C(C)C)C(CN)N1CN2| & 8.9977 \\
\verb|CC1=CCN(CCCC=CNN(O)C(C)NN(C)C)O1| & 8.9960 \\
\verb|N=CC=C1OC2=[SH]N(C1=O)C(C1=CCCNCC=NC1)=CN2| & 8.9933 \\
\verb|COC1=CCC2=NC(S(=O)CCC3C=COCNCCCC3)(C=CNN=C1)C=C2| & 8.9877 \\
\verb|CCC(F)C=CC=NC(CO)C(C)(NC1CC1)N(C)CNC(=CC=CS)CC#N| & 8.9859 \\
\verb|CN(C)NCCC=[N+]=CC1CCCO1| & 8.9852 \\
\verb|CC=NNCCCCCNC1=CC2(CC2)C(CNC2=CC=CCC2C(F)CO)=NO1| & 8.9834 \\
\verb|CCC1C=CNC(N(C)C)C=CC=[N+]=CC2(CC(C)C2)N(F)C12CNC2| & 8.9823 \\
\verb|CC(CC(C)C)=[N+]=CC=C1C=CC(C)CN=CC2(CCO2)NC=CNCCCN(C)C1| & 8.9818 \\
\verb|CNC(O)(O)C1=[N+]N(CC=CN(CCN)C(F)(Cl)OC=CCN)CN=C1| & 8.9808 \\
\verb|CN(C)NC(c1ccccc1CCCF)N1CC12CC(O)CCC2CCN| & 8.9802 \\
\verb|CCC1CCC1C(CO)N(N=N)C1(C=[N+]=CC=CC(F)N(C)C)CC(C)C1| & 8.9767 \\
\verb|C=C1COC=CC=NC=C(C(F)=CC(NCC2CC2)=C(C)NN(C)C)C(C)O1| & 8.9736 \\
\verb|CC1C(N(C)C)N2Cc3cccc(c3C(F)F)C(C)(C2)C1(CN)N(C)C| & 8.9697 \\
\verb|CCN(C)C1(C(=O)C(=NC=N)C(=O)NCCNCc2ccccc2C)C=CN1| & 8.9676 \\
\verb|C=COC=CNC(O)N(CC)S(=O)C1C=CC(CC=C(Cl)OC)=N1| & 8.9673 \\
\verb|O=S1CCC2=COC3C(CC2)CNCOC=CC1C=C1C=NC2=C1CC3O2| & 8.9640 \\
\verb|CCC=C1OC(CNN=C(SCC(F)F)S(C)=S)C=CC1=C[SH]1C=CC1C| & 8.9625 \\
\verb|C=C(C)CN(Cc1ccccc1F)C1N2CCC2(C)C1(CO)CCNCC1CC1| & 8.9599 \\
\verb|CCC12CN(CC3=CCOC=C3C)CNCCC3COC4OC3C=C4C1=[N+]2| & 8.9589 \\
\verb|CCN1CCCCCC(NCC(=S)N(CO)NN)N(CN(C)C)c2cc1ccc2F| & 8.9577 \\
\verb|CCC1(C)CCOC2CCCOCCN(CN(O)C(O)(OC)N(N)OC)CC=C2N1| & 8.9534 \\
\verb|CC1COC2CNCC3CN(C)NC(CF)=CCC=C3N12| & 8.9513 \\
\verb|COC1=CCC2=NC(S(=O)CCC3C=CCCNCOC=C3)(C=CNN=C1)C=C2| & 8.9496 \\
\verb|CC1=[N+]=CC=CCN(CCCNCC2=C(F)N=C(N)CO2)CC(C)C1| & 8.9480 \\
\verb|CC1CC2NC(CCO)N(C=CCNC(=O)CO2)C1| & 8.9479 \\
\verb|CC=C1CNCCCCC(CCN2NNC2(C)O)C1c1cncs1| & 8.9472 \\
\end{longtable}

\normalsize

\subsection*{A14. Transformation of Linear Terms in Trained FM Model}
\begin{figure}[H]
\centering
\includegraphics[width=0.9\textwidth]{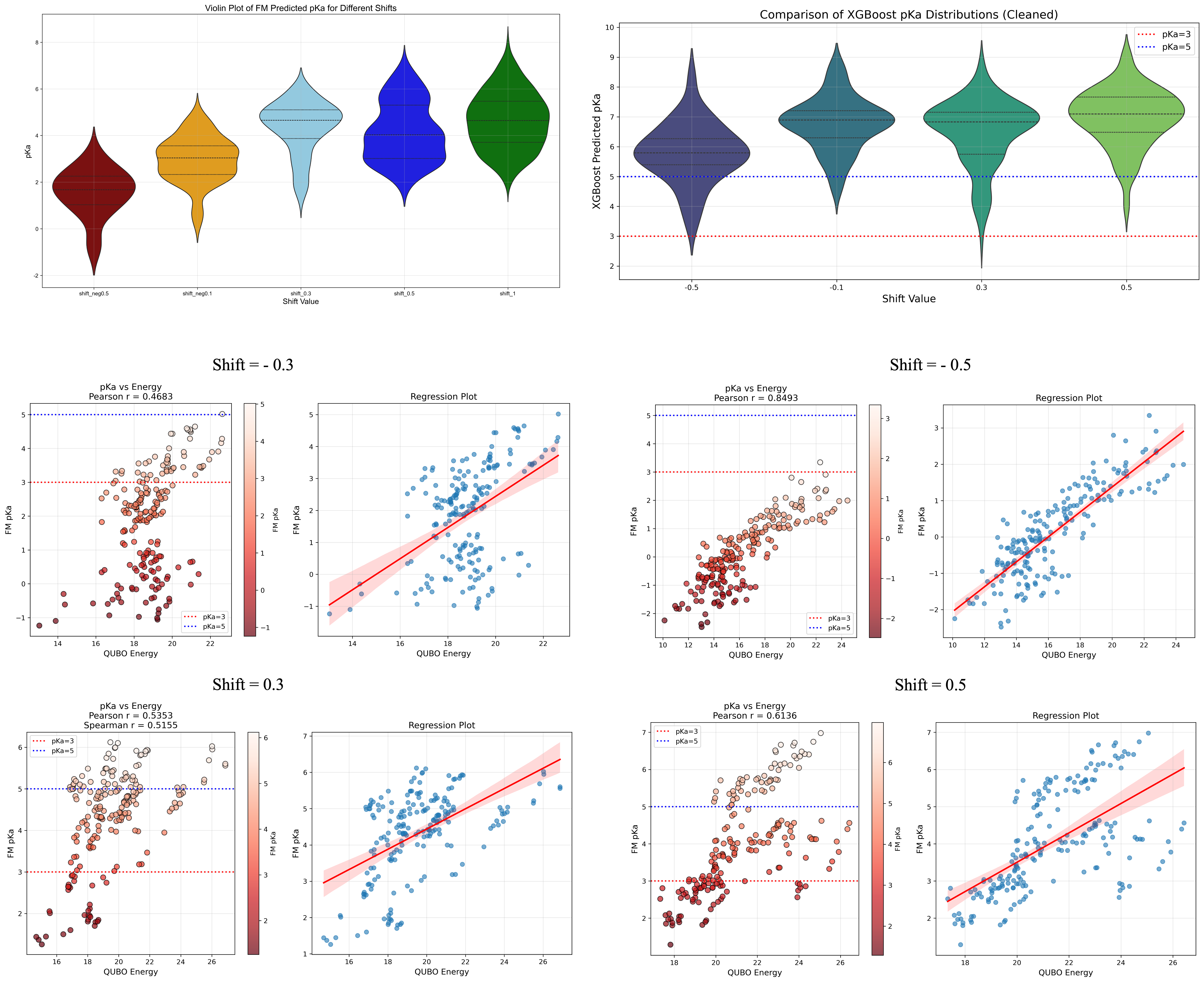}
\caption{Representative Shift Results}
\label{shift} 
\end{figure}

\subsection*{A15. Structural Bias in Joint Training of Factorization Machines with Long-Tail and Uniform Features}
\label{sec:conflict}

\subsubsection*{Problem Setup and Notation}
\label{sec:setup}

\paragraph{Feature 1: Long-Tail, Bucketized.}
Feature~1 is discretized into $m$ high-frequency (head) buckets $\{h_1,\dots,h_m\}$ and one low-frequency (tail) bucket $t$. Each sample belongs to exactly one bucket, represented by a one-hot vector.  
Let $n_{h_i}\approx N_h/m$ and $n_t = N_t$ be the sample counts, with $N_t \ll N_h$.

\paragraph{Feature 2: Uniform Continuous.}
Feature $f$ is a continuous variable, standardized such that $\mathbb{E}[x_f]=0$ and $\mathbb{E}[x_f^2]=1$. The true relationship involves only a linear effect of $f$:
\begin{equation}
y = \beta^* x_f + \varepsilon,
\end{equation}
where $\varepsilon$ is noise. Crucially, $f$ itself requires no interactions with any bucket.

\paragraph{Tail-Specific Effect of Feature~1.}
On top of the linear trend, samples from the tail bucket exhibit an additional interaction:
\begin{equation}
y = 
\begin{cases}
\beta^* x_f + \varepsilon, & \text{head samples},\\[2pt]
\beta^* x_f + \gamma^* x_f + \varepsilon, & \text{tail samples},
\end{cases}
\label{eq:true-model}
\end{equation}
with $\gamma^*\neq 0$. Equivalently,
\begin{equation}
y = \beta^* x_f + \sum_{b\in\mathcal{B}_1} \gamma^*_b \, x_b x_f + \varepsilon,
\end{equation}
where $\mathcal{B}_1 = \{h_1,\dots,h_m, t\}$, $\gamma^*_t = \gamma^*$, and $\gamma^*_{h_i}=0$ for all $i$.

\subsubsection*{FM Model and Weighted Loss}
\label{sec:model}
\setcounter{paragraph}{0}
We employ a factorization machine that captures interactions between the buckets of Feature~1 and the continuous Feature~$f$. The prediction for a sample in bucket $b$ is
\begin{equation}
\hat{y} = w_f x_f + \sum_{b\in\mathcal{B}_1} \langle \mathbf{v}_b, \mathbf{v}_f \rangle \, x_b x_f,
\label{eq:fm-pred}
\end{equation}
where $w_f\in\mathbb{R}$ and $\mathbf{v}_b,\mathbf{v}_f\in\mathbb{R}^k$ are latent vectors. Define the interaction scores $s_b := \langle \mathbf{v}_b, \mathbf{v}_f \rangle$. Then $\hat{y}_b = (w_f + s_b)x_f$.

The ideal parameters that perfectly fit the true model~\eqref{eq:true-model} are
\begin{equation}
w_f = \beta^*, \qquad s_t = \gamma^*, \qquad s_{h_i}=0 \;\;(\forall i).
\label{eq:ideal}
\end{equation}

To handle the long-tail distribution, sample weights $c_b$ are assigned based on bucket frequency, e.g., inversely proportional to $n_b^\alpha$. The tail weight $c_t$ is much larger than the head weight $c_h$. We define the effective total weight per bucket as $W_b := n_b c_b$:
\begin{itemize}
\setlength{\itemsep}{0pt}       
  \setlength{\parsep}{0pt}       
  \setlength{\parskip}{0pt} 
\item Head buckets: $W_{h_i} = \frac{W_h}{m}$, where $W_h = \sum_{i=1}^m W_{h_i}$.
\item Tail bucket: $W_t = N_t c_t$, with $W_t \gg W_h/m$ by design.
\end{itemize}

The training objective is the weighted mean squared error with $L_2$ regularization:
\begin{equation}
\mathcal{L} = \frac{1}{2} \sum_{b\in\mathcal{B}_1} W_b \, \mathbb{E}_{\text{bucket }b}\!\big[ (y - \hat{y}_b)^2 \big] + \frac{\lambda}{2} \Big( \sum_{b\in\mathcal{B}_1} \|\mathbf{v}_b\|^2 + \|\mathbf{v}_f\|^2 \Big).
\label{eq:loss}
\end{equation}
(We omit regularization on $w_f$ without loss of generality.) Substituting the true model and the FM prediction, and using $\mathbb{E}[x_f^2]=1$, the expected squared errors become
\begin{align}
\text{head } h_i &: \; (w_f + s_{h_i} - \beta^*)^2, \\
\text{tail } t   &: \; (w_f + s_t - \beta^* - \gamma^*)^2 .
\end{align}
Define the residuals
\begin{equation}
\delta_i := w_f + s_{h_i} - \beta^*, \qquad \delta_t := w_f + s_t - \beta^* - \gamma^*.
\label{eq:resids}
\end{equation}
The loss can then be written as
\begin{equation}
\mathcal{L} = \frac{1}{2} \sum_{i=1}^m \frac{W_h}{m} \delta_i^2 + \frac{1}{2} W_t \delta_t^2 + \frac{\lambda}{2} \Big( \sum_{i=1}^m \|\mathbf{v}_{h_i}\|^2 + \|\mathbf{v}_t\|^2 + \|\mathbf{v}_f\|^2 \Big).
\label{eq:loss-simplified}
\end{equation}

\subsubsection*{Stationary Point Conditions}
\label{sec:stationarity}

At any local minimum (or stationary point) of $\mathcal{L}$, the gradients vanish. We compute:

\noindent\textbf{(a) Gradient w.r.t.\ $w_f$:}
\begin{equation}
\frac{\partial \mathcal{L}}{\partial w_f} = \sum_{i=1}^m \frac{W_h}{m} \delta_i + W_t \delta_t = 0. \label{eq:grad-wf}
\end{equation}

\noindent\textbf{(b) Gradient w.r.t.\ $\mathbf{v}_{h_i}$:}
\begin{equation}
\nabla_{\mathbf{v}_{h_i}} \mathcal{L} = \frac{W_h}{m} \delta_i \, \mathbf{v}_f + \lambda \mathbf{v}_{h_i} = \mathbf{0}
\;\Longrightarrow\; \mathbf{v}_{h_i} = - \frac{W_h}{m\lambda} \delta_i \, \mathbf{v}_f. \label{eq:grad-vhi}
\end{equation}

\noindent\textbf{(c) Gradient w.r.t.\ $\mathbf{v}_t$:}
\begin{equation}
\nabla_{\mathbf{v}_t} \mathcal{L} = W_t \delta_t \, \mathbf{v}_f + \lambda \mathbf{v}_t = \mathbf{0}
\;\Longrightarrow\; \mathbf{v}_t = - \frac{W_t}{\lambda} \delta_t \, \mathbf{v}_f. \label{eq:grad-vt}
\end{equation}

\noindent\textbf{(d) Gradient w.r.t.\ $\mathbf{v}_f$:}
\begin{equation}
\nabla_{\mathbf{v}_f} \mathcal{L} = \sum_{i=1}^m \frac{W_h}{m} \delta_i \mathbf{v}_{h_i} + W_t \delta_t \mathbf{v}_t + \lambda \mathbf{v}_f = \mathbf{0}.
\end{equation}
Inserting~\eqref{eq:grad-vhi} and~\eqref{eq:grad-vt} yields
\begin{equation}
-\frac{1}{\lambda}\Big( \frac{W_h^2}{m^2} \sum_{i=1}^m \delta_i^2 + W_t^2 \delta_t^2 \Big) \mathbf{v}_f + \lambda \mathbf{v}_f = \mathbf{0}.
\end{equation}
If $\mathbf{v}_f = \mathbf{0}$, all $s_b = 0$ and $\gamma^*$ cannot be captured, contradicting the tail requirement. Hence $\mathbf{v}_f \neq \mathbf{0}$, and we obtain the scalar condition
\begin{equation}
\frac{W_h^2}{m^2} \sum_{i=1}^m \delta_i^2 + W_t^2 \delta_t^2 = \lambda^2. \label{eq:stationary-vf}
\end{equation}

\subsubsection*{Inherent Residual Bias under Shared Latent Regularization}
\label{sec:conflict-core}

\begin{proposition}[Non-vanishing residual bias]
\label{prop:bias}
Under the shared-embedding FM formulation with $L_2$ regularization and nonzero tail interaction $\gamma^*\neq0$, any permutation-symmetric stationary point satisfying $\mathbf{v}_f\neq\mathbf{0}$ necessarily exhibits nonzero residual bias:
\[
(\delta_h,\delta_t)\neq(0,0).
\]
\end{proposition}

\begin{proof}
The uniform Feature~2 demands that all interactions with $f$ be zero, i.e., $s_{h_i}=0$ and, in the absence of the tail effect, also $s_t=0$. However, the tail bucket of Feature~1 requires $s_t = \gamma^*$. Both requirements are imposed through the same latent vector $\mathbf{v}_f$, creating a tension that cannot be fully resolved under the current regularized shared-embedding formulation.

From~\eqref{eq:grad-vhi} we compute the head interaction scores:
\begin{equation}
s_{h_i} = \langle \mathbf{v}_{h_i}, \mathbf{v}_f \rangle = -\frac{W_h}{m\lambda} \delta_i \|\mathbf{v}_f\|^2.
\end{equation}
Combined with the definition $\delta_i = w_f - \beta^* + s_{h_i}$, we solve for $\delta_i$:
\begin{equation}
\delta_i = \frac{w_f - \beta^*}{1 + \frac{W_h}{m\lambda} \|\mathbf{v}_f\|^2}. \label{eq:delta_i}
\end{equation}
Due to the permutation symmetry of the head buckets in the loss function, any permutation-symmetric stationary point must satisfy $\delta_i = \delta_h$ for all $i$. Similarly, for the tail:
\begin{equation}
s_t = -\frac{W_t}{\lambda} \delta_t \|\mathbf{v}_f\|^2, \qquad
\delta_t = \frac{w_f - \beta^* - \gamma^*}{1 + \frac{W_t}{\lambda} \|\mathbf{v}_f\|^2}. \label{eq:delta_t}
\end{equation}

Equation~\eqref{eq:grad-wf} gives the balance between head and tail residuals:
\begin{equation}
W_h \delta_h + W_t \delta_t = 0 \;\Longrightarrow\; \delta_t = -\frac{W_h}{W_t} \delta_h. \label{eq:balance}
\end{equation}

Substituting into~\eqref{eq:stationary-vf} and using $\sum_i \delta_i^2 = m\delta_h^2$ yields
\begin{equation}
W_h^2 \delta_h^2 \left( \frac{1}{m} + 1 \right) = \lambda^2,
\end{equation}

hence
\begin{equation}
\delta_h = \pm \frac{\lambda}{W_h \sqrt{1 + 1/m}}, \qquad
\delta_t = \mp \frac{\lambda}{W_t \sqrt{1 + 1/m}}. \label{eq:delta-solved}
\end{equation}

These residuals are strictly non-zero for any $\lambda>0$ and are entirely determined by the regularization strength and the effective weights, \emph{independent of $\gamma^*$}.
\end{proof}

We now examine whether the tail interaction $s_t$ can still attain the ideal value $\gamma^*$. From the definition of $\delta_t$,
\begin{equation}
s_t = \gamma^* + \delta_t - (w_f - \beta^*). \label{eq:st-exact}
\end{equation}

To achieve $s_t = \gamma^*$, one would need $w_f - \beta^* = \delta_t$. Using~\eqref{eq:delta_i} and~\eqref{eq:balance},
\begin{equation}
w_f - \beta^* = \delta_h \left( 1 + \frac{W_h}{m\lambda} \|\mathbf{v}_f\|^2 \right)
= -\frac{W_t}{W_h} \delta_t \left( 1 + \frac{W_h}{m\lambda} \|\mathbf{v}_f\|^2 \right).
\end{equation}

Setting this equal to $\delta_t$ requires
\begin{equation}
-\frac{W_t}{W_h} \left( 1 + \frac{W_h}{m\lambda} \|\mathbf{v}_f\|^2 \right) = 1,
\end{equation}
which is impossible because the left-hand side is strictly negative while the right-hand side is strictly positive (all quantities are positive). Thus $w_f - \beta^*$ can never cancel $\delta_t$, and $s_t$ is systematically biased away from $\gamma^*$.

More explicitly, substituting the expressions for $w_f - \beta^*$ and $\delta_t$ into~\eqref{eq:st-exact} gives
\begin{equation}
s_t = \gamma^* - \delta_t \left( 1 + \frac{W_t}{W_h} \cdot \frac{1 + \frac{W_h}{m\lambda} \|\mathbf{v}_f\|^2}{1 + \frac{W_t}{\lambda} \|\mathbf{v}_f\|^2} \right).
\end{equation}
Since the correction factor is strictly positive, $s_t$ is always biased away from $\gamma^*$. In the empirically typical case where $\delta_t > 0$ (which occurs when the tail signal is positive and regularization dominates), the tail interaction is underestimated, i.e., $|s_t| < |\gamma^*|$. Even for large $\|\mathbf{v}_f\|^2$, the residual bias remains non-zero due to the stationary constraints. The root cause is that $\delta_t$ cannot be zero, thereby preventing exact recovery of $\gamma^*$ under the given parameterization.

\subsubsection*{Extension: A Quadratic Term in $y$ without Providing the Quadratic Feature}
\label{sec:quadratic-y}

One might attempt to incorporate non-linear effects of Feature~2 by modifying the target variable itself, e.g., letting
\begin{equation}
y = \beta^* x_f + \gamma^*_b x_f + \eta \, (x_2 - c)^2 + \varepsilon,
\label{eq:y-quadratic}
\end{equation}
while keeping the FM model unchanged, i.e., still using \eqref{eq:fm-pred} without any explicit quadratic feature. We show that this modification does not alleviate the structural bias; the same inherent tension remains.

Since $x_f$ is the standardized version of $x_2$, we can express $(x_2-c)^2$ as a quadratic polynomial in $x_f$:
\[
x_f = \frac{x_2 - \mu}{\sigma},\qquad
(x_2-c)^2 = a x_f^2 + b x_f + d,
\]
with constants $a=\sigma^2$, $b=2\sigma(\mu-c)$, $d=(\mu-c)^2$. Substituting into~\eqref{eq:y-quadratic},
\[
y = (\beta^* + \eta b) x_f + \gamma^*_b x_f + \eta a x_f^2 + \eta d + \varepsilon.
\]

For a sample in bucket $b$, the FM prediction is $\hat{y}_b = (w_f+s_b)x_f$. The error becomes
\[
y - \hat{y}_b = \Delta_b x_f + \eta a x_f^2 + \eta d + \varepsilon,
\]
where $\Delta_b = \beta' + \gamma^*_b - w_f - s_b$, with $\beta' := \beta^* + \eta b$. Assuming the standardized feature has negligible skewness (i.e., $\mathbb{E}[x_f^3]\approx 0$, as is common for many real-world continuous variables after normalization), and using $\mathbb{E}[x_f]=0$, $\mathbb{E}[x_f^2]=1$, and denoting $\mathbb{E}[x_f^4]=\kappa$, the expected squared error for bucket $b$ is
\begin{align*}
\mathbb{E}\big[(y-\hat{y}_b)^2\big] 
&= \mathbb{E}\big[ (\Delta_b x_f + \eta a x_f^2 + \eta d + \varepsilon)^2 \big] \\
&\approx \Delta_b^2 \cdot 1 + \eta^2 a^2 \kappa + \eta^2 d^2 + 2\eta^2 a d + \sigma_\varepsilon^2 \\
&= \Delta_b^2 + \text{const},
\end{align*}
where the constant does not depend on the bucket $b$. Hence, up to an additive constant, the loss function retains the same form as before, with the residuals redefined as
\begin{equation}
\delta_i := w_f + s_{h_i} - \beta', \qquad
\delta_t := w_f + s_t - \beta' - \gamma^*.
\label{eq:resids-quad}
\end{equation}

Thus the weighted loss $\mathcal{L}$ is structurally identical to \eqref{eq:loss-simplified}, merely replacing $\beta^*$ by $\beta'$.

All stationary point conditions \eqref{eq:grad-wf}--\eqref{eq:stationary-vf} carry over verbatim, leading once again to the non-zero forced residuals
\[
\delta_h = \pm \frac{\lambda}{W_h \sqrt{1 + 1/m}}, \qquad
\delta_t = \mp \frac{\lambda}{W_t \sqrt{1 + 1/m}},
\]
and the same systematic deviation in the tail interaction score $s_t$,
\[
s_t = \gamma^* - \delta_t \left( 1 + \frac{W_t}{W_h} \cdot \frac{1 + \frac{W_h}{m\lambda} \|\mathbf{v}_f\|^2}{1 + \frac{W_t}{\lambda} \|\mathbf{v}_f\|^2} \right).
\]

Therefore, even when the target variable contains a quadratic component of the uniform feature, as long as the FM model does not explicitly include the corresponding quadratic feature, the optimization exhibits the same structural bias. The tail interaction $s_t$ cannot exactly recover $\gamma^*$, and the long-tail modeling capability is damaged. The presence of the quadratic term merely shifts the effective linear coefficient $\beta'$ and adds a constant to the loss, without resolving the core tension originating from the shared latent vector $\mathbf{v}_f$ and the $L_2$ regularizer.

\subsubsection*{Why Separate Training on Feature~1 Succeeds}
\label{sec:solo}

When Feature~1 is trained alone, there is no pressure from a uniform feature to force all interactions to zero. The head buckets can be heavily down-weighted (or their interactions can be structurally excluded by not introducing $\mathbf{v}_f$ for them), allowing the model to focus almost entirely on the tail residual. In that case, the loss simplifies to balancing the tail regression term and the regularizer, and the stationary point can closely approach $s_t \approx \gamma^*$ with $s_{h_i}\approx 0$. The bias described in Sections~\ref{sec:conflict-core} and~\ref{sec:quadratic-y} does not arise.

\subsubsection*{Discussion and Practical Implications}
\label{sec:conclusion}

Our stationary-point analysis under the regularized shared-embedding FM formulation shows that jointly training a Factorization Machine with a uniformly distributed feature—whether the target $y$ contains only a linear effect or also a quadratic term of that feature—while sharing the latent vector with the interaction embeddings of a long-tail feature, leads to an unavoidable residual bias. This bias systematically prevents the tail interaction from attaining its ideal value and introduces spurious interactions on head samples. It is a structural consequence of the regularized objective and cannot be eliminated under the current regularized shared-embedding FM formulation by merely increasing tail weights or by engineering the target variable.

Intuitively, the abundant head samples favor suppressing interaction magnitudes through the shared embedding, whereas the sparse tail samples require amplified interactions to accurately fit the rare regime. Under shared regularization, these competing objectives cannot simultaneously attain their respective optima, inevitably pulling the embedding toward a compromise that compresses the tail signal.

It is important to note that our analysis identifies a specific limitation of the shared-latent architecture; it does not preclude alternative architectures with decoupled or feature-specific embeddings, adaptive regularization, or higher-order interaction parameterizations. As such, this result should be interpreted as a diagnostic insight rather than a fundamental impossibility theorem. Rather, it highlights that in the standard shared-embedding FM with $L_2$ regularization, competing gradient signals inevitably pull the embedding toward a compromise solution, compressing the tail interaction.

This theoretical result is consistent with our empirical observation that introducing uniformly distributed auxiliary features or quadratic target engineering degrades sparse-tail fitting performance, even though it may improve global regression smoothness. Specifically, it explains why simply adding the quadratic feature as an independent linear term (without sharing $\mathbf{v}_f$) fully preserves the previously optimized long-tail interactions. The practical implication is clear: the uniform feature should be incorporated solely via separate, independent terms (a linear coefficient or an explicit quadratic feature, e.g., $(x_2-c)^2$), with all parameters related to the long-tail feature kept frozen. This preserves the previously acquired tail modeling capacity and circumvents the inherent bias present in any joint training that shares latent vectors under $L_2$ regularization.

\end{document}